\documentclass[12pt]{article}

\usepackage{amsmath}
\usepackage{amsfonts}
\usepackage{amssymb}
\usepackage{eurosym}
\usepackage{amsthm}

\usepackage{graphicx}
\usepackage{float}
\usepackage{subfig}
\usepackage{caption}
\usepackage{booktabs}
\usepackage{multirow}
\usepackage{enumerate}
\usepackage{setspace}
\usepackage{pdflscape}
\usepackage{lscape}
\usepackage[hang]{footmisc}

\usepackage{comment}
\usepackage{lipsum}
\usepackage{titlesec}
\titleformat*{\subsection}{\normalsize\bfseries} 
\titleformat*{\section}{\large\bfseries} 

\usepackage{pgfplots}
\pgfplotsset{compat=1.18}

\usepackage{xcolor}
\definecolor{airforceblue}{rgb}{0.36, 0.54, 0.66}
\definecolor{battleshipgrey}{rgb}{0.52, 0.52, 0.51}
\definecolor{charcoal}{rgb}{0.21, 0.27, 0.31}

\usepackage{hyperref}
\hypersetup{
    colorlinks=true,
    linkcolor=airforceblue,
    citecolor=airforceblue,
    urlcolor=battleshipgrey
}

\usepackage[
backend=biber,
style=authoryear,
natbib=true,
sorting=nyt,
sortcites=true,
maxcitenames=4,
maxbibnames=99,
uniquename=false,
giveninits=true,
hyperref=true
]{biblatex}

\DeclareFieldFormat{bibhyperref}{\bibhyperref{#1}}

\DeclareCiteCommand{\citet}
  {\usebibmacro{prenote}}
  {\usebibmacro{citeindex}%
   \printtext[bibhyperref]{\printnames{labelname}}%
   \addspace%
   \printtext[bibhyperref]{\mkbibparens{\printfield{year}}}}
  {\multicitedelim}
  {\usebibmacro{postnote}}

\DeclareCiteCommand{\citep}[\mkbibparens]
  {\usebibmacro{prenote}}
  {\usebibmacro{citeindex}%
   \printtext[bibhyperref]{%
     \printnames{labelname}%
     \addcomma\space%
     \printfield{year}}}
  {\multicitedelim}
  {\usebibmacro{postnote}}

\DeclareCiteCommand{\citeauthoronly}
  {\usebibmacro{prenote}}
  {\usebibmacro{citeindex}%
   \printtext[bibhyperref]{\printnames{labelname}}}
  {\multicitedelim}
  {\usebibmacro{postnote}}

\let\cite\citet

\AtEveryBibitem{%
    \hypertarget{cite.\thefield{entrykey}}{}%
}

\newtheorem{lemma}{Lemma}

\newtheorem{proposition}{Proposition}

\def\eproof{\hbox{\hskip3pt\vrule width4pt height8pt depth1.5pt}}

\renewcommand{\baselinestretch}{1.5}
\usepackage[top=1in, bottom=1in, left=1in, right=1in]{geometry}

\begin{document}

	\title{The Economic Benefits and Costs of AI and Policies to Mitigate AI's Impact on Inequality}

	\author{Matthew O. Jackson  \qquad  Zafer Kanik\thanks{
			Jackson: Department of
			Economics, Stanford University, and external faculty member at the Santa Fe Institute.
			Kanik: Adam Smith Business School, University of Glasgow.}}

\vspace{1.5cm}

\date{\today}

\maketitle
	
\thispagestyle{empty}		
\setcounter{page}{0}

\begin{abstract}	
	\begin{singlespace}
   We examine the economic impact of increasingly productive AI and policies that spread its benefits across the economy.
   Improvements in AI productivity trigger labor reallocation and changes in absolute and relative wages for different types of labor.  Wages of labor that is essential for building AI increase faster than overall GDP.   Wages of labor that is substituted for by AI decrease in both absolute and relative terms.  Wages of labor that is used only in final goods production and is not displaced by AI increase in line with overall GDP.   
    We contrast the impact of productivity gains depending on whether AI production is competitive or monopolistic. Monopoly production of AI restricts its deployment, slowing the transition and impact of AI.
    Optimal tax and regulatory policies that 
    achieve Pareto-improvements differ depending on whether there is competition in AI production. 
    \end{singlespace}


\end{abstract}
	
	\newpage
	
\renewcommand{\baselinestretch}{1.1}

\section{Introduction}

\label{SEC:INTRO}

Artificial intelligence is already becoming one of the most consequential economic forces in recent history. Investments are in the trillions of dollars, and technology advances along dimensions that directly intersect with labor---from writing code and analyzing data to generating text and assisting decision-making---are leading AI to increasingly perform tasks once exclusively carried out by workers across a wide range of occupations \citep[e.g.,][]{IdeTalamas2024,IdeTalamas2025,IdeTalamas2026,autor2024,brynjolfsson2025canaries,brynjolfssonLiRaymond2025}.

This raises first-order questions for economists and citizens: How will AI reshape wages and wage inequality across different types of workers? What policies can ensure that AI's gains are shared, and how do these policies depend on the pace of technological progress, production technologies, and market structure? These questions are already being debated by the very people building AI. Anthropic's
CEO Dario Amodei has argued that those creating the technology should take
responsibility for its effects, and that this is a time to worry less about disincentivizing growth and worry more about ``making sure that everyone gets a part of that growth'' \citep{hagey2026}. That debate increasingly turns to taxes, and regarding this,
\cite{amodei2025} said ``it's going to involve taxes on people like me, and maybe specifically on the AI companies.'' Likewise, \cite{openai2026} has called for
modernizing the tax base---particularly mentioning ``policymakers could rebalance the tax base by increasing reliance on
capital-based revenues—such as higher taxes on capital gains at the top, corporate income, or
targeted measures on sustained AI-driven return.''

 We address these questions within a model in which labor is tracked by its role in production and whether it competes with AI, rather than by skill level.  The model has two sectors---AI production and other (``final'') goods production.  The three types of labor are as follows. \emph{Substitutable workers} perform tasks that AI can replace---e.g., customer service, routine coding, legal clerking, standardized data processing and analysis, content generation, translation, editing, etc. \emph{Complementary workers} are required to produce AI itself---e.g., machine learning engineers, system architects, researchers---and who also have productive uses in the final goods sector. \emph{Final goods-specific workers} perform tasks that are neither substituted by AI nor required for AI production, but who benefit from AI adoption in their production---e.g., doctors, teachers, managers, etc. 
Where workers fit in this production structure determines how they fare from AI adoption.

The fact that AI is itself produced using labor is an important feature of our model.  Thus, it consumes labor in addition to being complementary to labor in some productive tasks and substituting for labor in other tasks. 
Other researchers have modeled labor-substituting technologies via capital adoption, without modeling the production of that capital.  Here, AI is built using workers---e.g., engineers, researchers, data scientists---who have productive uses elsewhere in the economy. 
The important aspect of this is how it pins down wages.  Complementary workers who are used in AI production end up having wages that move with AI productivity ( regardless of how many are employed there), due to the fact that other sectors have to compete for such workers with the AI sector.   On the other hand, substitutable workers' wages are tied down by the price of AI, with which they compete for jobs, and hence their wages  fall with AI productivity---again, regardless of how many are displaced.

These wage dynamics  have no direct analogue in previous analyses. 
This allows us to characterize the tax policies that create Pareto-improvement for all workers, and to analyze how such policies depend on the pace of technological progress, the structure of AI production, and the relative roles of different types of labor in various parts of the economy.

Modeling AI as labor-produced gives rise to a distinctive transition phase during which AI is deployed in the economy.  
Once AI passes a minimum productivity threshold, it begins to compete with substitutable labor, and the economy reallocates labor. During this transition, wages of substitutable workers fall in absolute terms even as aggregate output rises---their wage is tied to the competition they face from AI for their jobs.   Substitutable workers do not drop out of the economy---but instead see lower wages and get pushed out of the tasks where AI substitutes for them into ones where it does not.
By contrast, as AI production increases, complementary workers' wages grow faster than aggregate output, reflecting their dual role as builders of AI and workers who benefit from working alongside it.
The third category of workers---final goods-specific workers (who are not substitutable)---end up with wages that track aggregate output one-for-one.

A central insight that emerges from our model is that the wage ratio of complementary to substitutable labor is pinned down by AI productivity alone---independent of the relative populations of the two labor types. This labor-supply invariance is intuitive and follows from the observations above that the complementary wage increases with the productivity of AI, while the substitutable wage falls as the relative price of AI falls as its productivity rises.%
\footnote{Although we show this in a Cobb-Douglas setting, the result extends well beyond Cobb-Douglas as discussed in Section \ref{SEC:ROBUSTNESS}.} 

The wage dynamics above raise the natural question of which policies ensure that AI's gains are broadly shared.   The three types of labor experience widely different outcomes: substitutable wages fall in both relative and {\sl absolute} terms, final goods specific workers' wages rise in line with GDP, and wages of complementary workers rise faster than GDP.   
Because aggregate output grows during the transition, policy interventions can make all workers better off. 

Given that final goods sector specific workers' wages are rising in line with GDP, taxes that offer welfare (Pareto) improvements relative to pre-AI adoption effectively involve taxing AI-complementary workers and then subsidizing those who are being substituted by AI. We characterize those Pareto-improving tax policies.  
Given that there are overall gains to be shared, they can be distributed in many different combinations.  We examine two natural targets that involve contrasting tax rates as
a function of AI productivity.   One target is simply to maintain AI-substitutable workers at their absolute pre-AI consumption level.  A more ambitious target is to share the AI gains equally among all workers, so as to maintain substitutable workers' {\sl share} of GDP.

The first target of simply maintaining substitutable workers at their pre-AI consumption level involves a tax on complementary workers that is \emph{hump-shaped} in AI productivity:  it initially rises as AI becomes more productive and the wages of substitutable workers fall, but then the tax rate eventually decreases as the increased productivity of AI means that smaller fractions of richer workers' wages are needed to maintain poorer workers' consumption constant. 
By contrast, ensuring that all workers gain equally from the AI expansion requires a tax rate that is  \emph{increasing} in AI productivity, because the market wage ratio widens as AI productivity advances. The choice of the welfare target therefore determines whether a welfare-improving tax policy involves more or less taxation as AI advances.

We also consider what happens if AI production becomes monopolized. A monopolist slows AI adoption by raising its price.  Also, by restricting AI production, the monopolist hires fewer complementary workers and displaces fewer substitutable workers, with aggregate output being below the competitive level. Wage inequality is thus compressed relative to competition---the substitutable wage falls less and the complementary wage rises less---but the total surplus available for redistribution also shrinks. Redistribution must now be carried out among the monopolist and the affected worker classes.   

If one can simply regulate the monopoly price to be the competitive price (or, alternatively, fully tax monopoly profits), then the problem simplifies to the previous competitive policy analysis.  
Our additional policy analysis concerns situations in which  price controls or full profit taxes are not (politically) viable, but partial taxes are.  

Near the onset of AI adoption, monopoly profit grows more slowly than the substitutable-worker deficit, so a profit tax alone cannot finance Pareto-improving transfers; it must be supplemented by a tax on complementary workers, whose wages rise during the transition. As AI productivity rises, monopoly profit eventually overtakes the deficit and profit taxation alone suffices to keep substitutable workers at pre-AI consumption levels. Comparing redistribution regimes under Pareto-improving taxation, we further show that restoring competition---via a price cap or full profit taxation---always leaves complementary workers better off after taxes, with the gap widening as AI advances.

\paragraph{Related Literature}

Our paper relates to literatures that have separately studied technology adoption, wage inequality, taxation, and market power. We bring these issues together in a general-equilibrium model with three role-based labor types, endogenous labor reallocation, and AI produced within the economy---rather than introduced as an exogenous capital input or imported input. This allows us to study how AI productivity affects wages and adoption under both competitive and monopolistic AI production, and how market structure shapes the surplus available for redistribution.

This setting also allows us to ask how the gains from AI are distributed across workers; and, when AI production is monopolized, between workers and monopoly rents. Labor-produced AI creates general-equilibrium wage dynamics that are absent from standard models of technology adoption. In the tax analysis, a key redistributive distinction is between taxing the factor income generated by AI production and taxing monopoly profits from AI production. These instruments are not interchangeable: near the onset of adoption, monopoly rents may be too small to compensate displaced workers, so factor-income taxation may be necessary; later in the transition, as rents grow, profit taxation can become sufficient on its own.

A growing literature studies the impact of AI on labor allocation in knowledge- and task-based economies. \cite{IdeTalamas2025} model AI in knowledge hierarchies, building on \citet{garicanoRossiHansberg2006,garicanoRossiHansberg2015}. \cite{acemoglu2025simple} studies AI through automation and task complementarities. 
We take a different approach, and introduce a role-based labor classification that distinguishes labor by whether they can work in the AI sector and whether they are substitutes or complements to AI.\footnote{We focus on three of the possible four labor types that this two-by-two classification provides.  The missing box is substitutable labor that can only work in the final goods sector.  The impact on that type of labor is a straightforward extension as AI simply lowers their wages.  To keep the model less cluttered, we omit that category.}\footnote{This also distinguishes our model from the job-polarization literature, which emphasizes the hollowing-out of the middle \citep{autorDorn2013,goosManningSalomons2014}.}
This way of classifying labor aligns directly with how it ends up being impacted and allows us to develop a clear policy analysis.

Our model and the foundation for our analysis that it generates in terms of labor allocation and wages relates to a broad literature on automation, skill-biased technical change, capital-deepening, and directed-technical-change. The models of \cite{zeira}, \cite{alm}, and \cite{ar1,ar3,ar4} provide foundations for displacement, reinstatement, and the task content of automation; \cite{hemous2022rise} studies endogenous automation through horizontal innovation; and related work studies automation's implications for employment, labor shares, retraining, wealth inequality, and growth \citep{zn,as,jaim,moll2022,aghionJonesJones2019}. The skill-biased and directed-technical-change literatures study how technology, capital, and skill shape wage inequality \citep{greenwood19971974,acemoglu,acemoglu2002directed}; \citet{ac2,aa} survey this broader literature. 
The results that our model produce are complementary but distinct: wage inequality is generated by endogenous adoption of labor-produced AI, rather than by a standard displacement, capital-deepening, or directed-technical-change channel. Finally, our labor-supply invariance result complements the partial-equilibrium findings of \cite{lewis} and \cite{clp} on endogenous automation absorbing labor-supply changes.

Existing work that studies robot/AI taxation, technology regulation, and redistribution in response to automation/AI \citep{korinekStiglitz2019,thuemmel2023,guerreiro2022,costinot2023,korinekLockwood2026,acemogluManeraRestrepo2020} mainly focuses on taxing or regulating the technology itself---robot taxes, AI taxes, or broader technology regulation---and, when labor-income taxation is present, it is typically part of a distortionary tax system.\footnote{\cite{bastani2024} survey AI taxation, and \cite{bryan2026} reviews the economic impacts of AI.}
In our setting, taxes on production or products are distortionary, while labor taxes are not.  This enables us to derive explicit redistribution schedules across worker classes whose gains and losses arise endogenously from their roles in AI production and deployment without harming growth or GDP. The policy problem is how to redistribute income among workers without limiting technology adoption. Because the transition is endogenous, AI productivity determines both the size of the losses and the available tax base and so, as we show, the required tax path depends in interesting ways on AI productivity and on the redistribution objective.

Our monopoly analysis connects the paper to recent work on AI market structure. \cite{korinekVipra2025} study how scale and scope economies can push AI markets toward concentration. Closest to our monopoly analysis, \citet{atheyScottMorton2025} study upstream AI market power in an open economy general-equilibrium model where AI is a priced imported input, showing how strategic AI pricing affects adoption, displaced-worker wages, and welfare. Our analysis differs by embedding monopoly AI production in a labor-role-based general-equilibrium model, making the monopolist's problem quite distinct from theirs. While \citet{atheyScottMorton2025} focus on the incidence of upstream AI market power and its effects on wages and welfare, we study and compare Pareto-improving redistribution policies that use taxes on workers and/or monopoly profit.

The paper proceeds as follows. Section~\ref{SEC:MODEL} introduces the model. Section~\ref{SEC:TRANSITION} characterizes wage dynamics. Section~\ref{SEC:INCOME_TAX} analyzes Pareto-improving taxation. Section~\ref{SEC:MONOPOLY} extends to monopoly. Section~\ref{SEC:ROBUSTNESS} discusses extensions and robustness of the modeling choices. Section~\ref{SEC:CONCLUSION} concludes.

\section{Our Model}
\label{SEC:MODEL}

We consider a competitive economy composed of two sectors and three distinct types of labor.

The AI sector, indexed by $a$, produces the AI technology ($Y_a$) that serves as an input in final goods production. 
The final goods sector, indexed by $f$, produces a representative consumption good ($Y_f$) using a combination of labor and AI technology.

Labor is classified into three categories ($S,C,F$). 
Aggregate endowments of labor are denoted $\overline L^S$, $\overline L^C$, and $\overline L^F$, all supplied inelastically.

The Cobb-Douglas production functions are: 
\begin{equation}
	Y_a = A_{a} (L^S_{a})^{\alpha^S_{a}} (L^C_{a})^{\alpha^C_{a}},
	\label{eq:tech_production}
\end{equation}
\begin{equation}
	Y_f = A_{f} \left[ L^S_{f} + X_a \right]^{\alpha^S_{f}} (L^C_{f})^{\alpha^C_{f}} (L^F_{f})^{\alpha^F_{f}},
	\label{eq:final_production}
\end{equation}
where $A_a$ and $A_f$ represent productivity in the AI and final goods sectors, respectively. $X_a$ denotes the quantity of AI technology used in final goods production. The term $[L^S_{f} + X_a]$ captures the perfect substitution of AI for some forms of labor ($L^S_{f}$). 

We use a Cobb-Douglas formulation because it leads to very interpretable and intuitive expressions.  We show that it does not affect the qualitative results in Section \ref{SEC:ROBUSTNESS}.

Both sectors exhibit constant returns to scale, with share parameters satisfying $\alpha^S_{a}+\alpha^C_{a}=1$ and $\alpha^S_{f}+\alpha^C_{f}+\alpha^F_{f}=1$. The superscripts $\alpha$s represent the relative importance (or intensity) of these different forms of labor in different sectors.

Our labor classification reflects the
fact that AI, robotics, and other forms of automation interact with labor in three different ways---there is labor for which it substitutes, there is labor involved in final goods production that it does not substitute for and is used together with the AI inputs, and then there is labor that is used in the production of AI inputs (and cannot be substituted by it).

More specifically:
\begin{itemize}

\item  \textit{Final Good-Specific Labor ($ \overline L^F$)} encompasses labor employed exclusively in the final goods sector ($L^F_{f}$).  
This labor category includes for instance, managers, retailers, nurses, and some service workers, who can utilize AI tools to enhance their productivity, but whose skills are not easily substituted for AI systems, nor do they produce AI. 
  
\item \textit{Complementary-Labor ($ \overline L^C$)} comprises workers who can work in both the final goods sector ($L^C_{f}$) and the AI sector ($L^C_{a}$)
and who make use of AI tools, but who are not substituted for by AI. 
For example, this includes software architects and various computational scientists, as well as some engineers, product designers, and managers. 

\item \textit{Substitutable-Labor ($ \overline L^S$)} consists of workers performing   tasks that can also be partly or wholly done by AI. In the final goods sector ($L^S_{f}$), these workers are perfect substitutes for AI. However, they can also serve as inputs for producing AI in the AI sector ($L^S_{a}$). 
For instance, a coder or data analyst  competes with AI tools, but can also be useful in the AI sector in training, validating, and deploying AI.
\end{itemize}

\cite{alm} study wage inequality and skill acquisition using an aggregate production function $Y=[L_r+ K]^{\alpha}{L_n}^{1-\alpha}$, where $L_r$ represents routine labor, $K$ denotes capital that is perfectly substitutable for $L_r$, and $L_n$ signifies non-routine labor. Our model extends their framework not only by including a third form of labor but also by 
modeling the production of automation/AI technologies via Equation \ref{eq:tech_production}. 
Modeling the production of AI---a large and rapidly growing sector of the economy and potentially eventually the largest sector---introduces distinct general equilibrium effects with a transition phase and reallocation of labor across sectors.

\subsection{Competitive Equilibrium}

We normalize the price of the final good to $p_f=1$.  This makes things very easy to interpret, as then all other price and wage levels directly translate into how much GDP they cost or can purchase.  

In particular, let $p_a$ denote the price of AI technology, and let $w_S$, $w_C$, and $w_F$ denote the equilibrium wages for the respective labor types.   

A \textit{competitive equilibrium} is a set of prices $\{p_f, p_a, w_S, w_C, w_F\}$ and allocations and outputs $\{Y_f, Y_a, X_a, L^S_{f}, L^C_{f}, L^F_{f}, L^S_{a}, L^C_{a}\}$ such that producers of AI and final goods maximize profits, households supply labor inelastically, and output and labor markets clear {($L^C_{f}+L^C_{a}=\overline L^C$, $L^S_{f}+L^S_{a}=\overline L^S$, $L^F_{f}=\overline L^F$, $Y_a=X_a$)}. As the definition is completely standard, we provide its formal statement in the Appendix.

\section{The Impact of Productivity Gains in AI on Employment, Wages, and GDP}

\label{SEC:TRANSITION}

We first analyze  how the economy ``transitions'' as 
the productivity of AI, $A_a$, increases. 

This transition is a feature of our framework: because AI is itself produced by labor---unlike in existing models where AI enters purely as capital---increases in AI productivity trigger a reallocation of workers across sectors, leading to changes in wages and wage inequality across different types of labor.
Specifically, there exist two threshold levels of AI productivity, $A_a^{*}\leq A_a^{**}$, such that: if $A_a \leq A_a^{*}$ then there is no AI production, if $A_a \geq A_a^{**}$ then AI has fully displaced substitutable-labor in final good production,
and in between the allocation of labor changes with AI productivity.

Letting $\Phi = (\alpha^C_{a})^{\alpha^C_{a}} (\alpha^S_{a})^{\alpha^S_{a}}$, these threshold productivity levels are:
\begin{equation}
	\label{threshold}
	A^* = \frac{1}{\Phi} \left( \frac{\overline{L}^S}{\overline{L}^C}\right)^{\alpha^C_{a}} \left( \frac{\alpha^C_{f}}{\alpha^S_{f}} \right)^{\alpha^C_{a}}, \qquad A^{**} = \frac{1}{\Phi} \left( \frac{\overline{L}^S}{\overline{L}^C}
\right)^{\alpha^C_{a}} \left(\frac{\alpha^S_{f}\alpha^C_{a} + \alpha^C_{f}}{\alpha^S_{f}\alpha^S_{a}} \right)^{\alpha^C_{a}}.
\end{equation}

Equation~\ref{threshold} reveals the crucial role of complementary-labor in AI production, and the importance of modeling AI production. 

If $\alpha^C_{a}=0$, then there is no transition phase and $A_a^{*}=A_a^{**}=1$.  
In that case, there is a bang-bang transition, when $A_a>1$ substitutable labor is fully substituted for in final production and becomes fully employed in AI production.  Effectively substitutable labor can either be used to its own substitute or the final good, and when   $\alpha^C_{a}=0$ it means that $\alpha^S_{a}=1$ and so substitutable labor produces its own substitute linearly, and hence the bang-bang transition at $A_a=1$. Then, wages for different labor classes all track GDP and there is no change in relative wage levels.  

The interesting case is where there is non-substitutable labor employed in the production of AI: $\alpha^C_{a}>0$.  Then, the economy exhibits a gradual substitution (or transition) phase, as labor (both substitutable and complementary) is gradually reallocated to the production of AI. This substitution phase features both a lower bound ($A_a^{*}$) and an upper bound ($A_a^{**}$), giving rise to three distinct phases of AI adoption: a pre-AI phase, a transition phase, and a post-phase. 

Also note that a finite upper bound $A_a^{**}<\infty$---and hence the existence of a post-phase---arises only if there exist tasks performed by substitutable-labor in both the AI sector and the final goods sector: $\alpha^S_{a}>0$. In that case, as AI productivity $A_a$ rises, substitutable-labor, together with complementary-labor, shift from the final goods sector to the AI sector until $L^S_{f}=0$. Otherwise, for $\alpha^S_{a}=0$, substitutable-labor $L^S$ is only ever employed in the final goods sector and simply experiences a  declining wage as $A_a$ increases.

We document the wage levels as $A_a$ increases in Proposition \ref{prop:equilibrium_wages}.

First, let 
$E_a= \frac{p_aY_a}{p_fY_f}$
denote the share of GDP spent on AI.
During the transition phase, the share of GDP spent on AI takes on the closed form:
\[
E_a=\frac{p_aY_a}{p_fY_f}
=
\frac{
\alpha^S_{f}\frac{\overline L^C}{\overline L^S}\left(A_{a}\Phi\right)^{\frac{1}{\alpha^C_{a}}}-\alpha^C_{f}
}{
\alpha^C_{a}\left(1+\frac{\overline L^C}{\overline L^S}\left(A_{a}\Phi\right)^{\frac{1}{\alpha^C_{a}}}\right)
}.
\]
This share is increasing in $A_a$ during the transition phase; i.e., $\frac{d E_a}{dA_a}>0$.

\begin{proposition}
\label{prop:equilibrium_wages}{\rm [Wages]}
The equilibrium wages for the three types of labor are:
\begin{align}
    w_S &= Y_f  \frac{\alpha^S_{f} - E_a \alpha^C_{a}}{\overline{L}^S}, \label{eq:wS_GDP} \\
    w_C &= Y_f  \frac{\alpha^C_{f} + E_a \alpha^C_{a}}{\overline{L}^C}, \label{eq:wC_GDP} \\
    w_F &= Y_f  \frac{\alpha^F_{f}}{\overline{L}^F}. \label{eq:wF_GDP}
\end{align}
The corresponding share of GDP spent on AI in the economy,  $E_a $, is 0 in the pre-AI phase ($A_a \leq A_a^*$), 
    $\alpha^S_{f}$ in the post-phase ($A_a \geq A_a^{**}$), and increasing in between. 
\end{proposition}

Proposition~\ref{prop:equilibrium_wages} expresses wages as shares of aggregate output $Y_f$, revealing how AI adoption redistributes both absolute and relative income across labor types through a single sufficient statistic: the interaction of AI expenditure share with the intensity of complementary labor in AI production, $E_a \alpha^C_{a}$. The interaction of AI adoption, $E_a$, with the complementary labor intensity in AI production, $\alpha^C_{a}$, is the channel through which labor-produced AI delivers wage effects on different types of labor that are absent in existing models of automation/AI.

The wage of complementary labor (\ref{eq:wC_GDP}) is increasing in both absolute and relative terms during transition, reflecting the dual role of $\overline{L}^C$ workers---they are needed in building AI and also benefit from AI adoption in the final goods production.

The wage of final goods labor (\ref{eq:wF_GDP}) remains a fixed fraction $\alpha^F_{f}$ of output regardless of AI adoption.  The absolute final-goods workers' wage rises with GDP reflecting the increased productivity, and moves one-for-one with GDP as there is no interaction between the employment of this type of labor and AI adoption. This one-for-one movement is a
Cobb--Douglas result, but it highlights another unique implication of our
framework: a tractable Cobb-Douglas specification produces
three economically distinct labor classes purely from differences in
production roles. In particular, the wages of $L^C$ and $L^F$ workers behave
differently not because of any difference in their elasticity of
substitution in the final goods sector, but because one is used in AI
production and the other is not. Were $L^C$ confined to the final goods
sector, it would behave exactly like $L^F$---earning a constant share
of GDP, with no premium from AI adoption.\footnote{The asymmetric distributional
consequences of AI adoption for $L^C$ and $L^F$ workers therefore arise
endogenously from \emph{who produces AI}, rather than from imposed
differences in elasticities. In a richer CES specification,
the absolute wage of $L^F$ would no longer move exactly one-for-one
with GDP---its growth rate would depend on the elasticity of
substitution between $L^F$ and the other inputs in final-goods
production. The qualitative results here would nonetheless carry
through as long as the key distinction between $L^C$ and $L^F$ is that $L^C$ is used in AI production while $L^F$ is not.}

In contrast, the substitutable-labor wage (\ref{eq:wS_GDP}) {\sl decreases}
in both relative and absolute levels during the transition. As AI captures a larger fraction of production, the effective demand for substitutable labor falls, depressing $w_S$. 
This holds even though substitutable labor is used in AI production ($0<\alpha^S_{a}<1$).  This is because their share in the production of AI is less than 1, and the displacement effect for substitutable workers in the final goods sector---through a higher $\frac{p_a Y_a}{p_f Y_f}$---dominates the rising demand for these workers in the production of AI, and their wages decrease as AI adoption rises. 
Once the transition is complete, then their wages become a constant fraction of GDP and begin rising again.

\begin{figure}[t]
\centering
\includegraphics[width=1\textwidth]{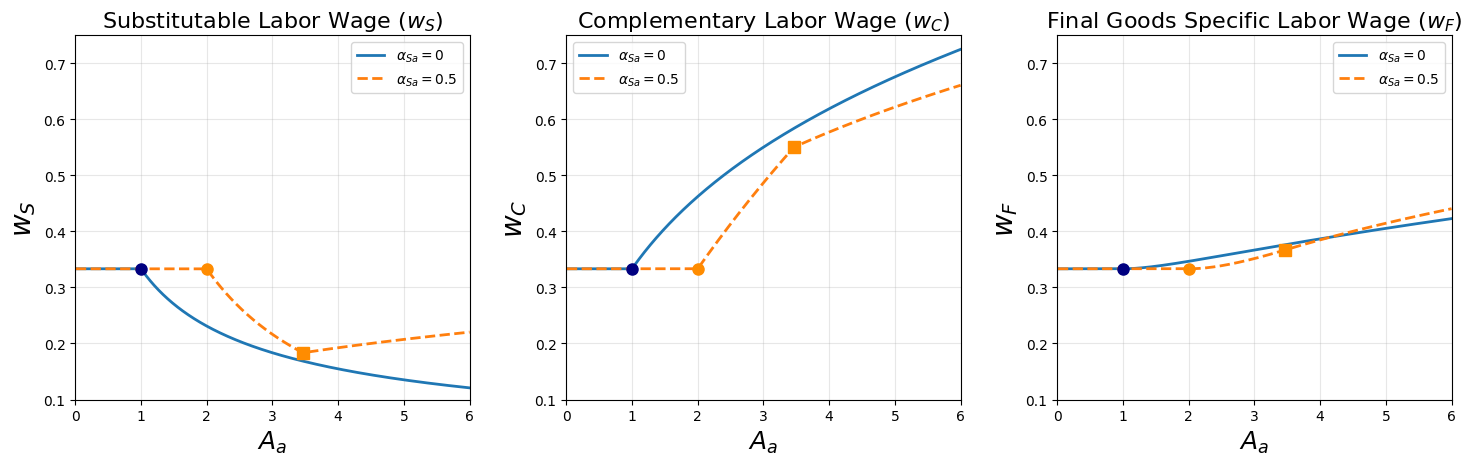}
\caption{AI Adoption and Equilibrium Wages}
\label{fig_wages}

\vspace{0.2cm}
\begin{minipage}{0.95\textwidth}
\footnotesize
\textit{Notes:} This figure displays wages as AI productivity ($A_a$) increases from $A_a=0$ to $A_a=6$. The three panels show wages for substitutable ($w_S$), complementary ($w_C$), and final goods specific ($w_F$) labor. Two AI technologies are compared: the solid blue line  depicts the case where AI is produced exclusively by complementary labor ($\alpha^S_{a}=0$) resulting in an endless transition phase; the dashed orange line depicts balanced factor shares in AI production, allowing a finite transition ($\alpha^S_{a}=0.5$). Dots mark the onset of AI adoption ($A_a^*$); the square marks the completion of transition ($A_a^{**}$) for the $\alpha^S_{a}=0.5$ case. Other parameters: $\alpha^S_{f}=\alpha^C_{f}=\alpha^F_{f}=1/3$, $\overline{L}^S=\overline{L}^C=\overline{L}^F=1$, and $A_f=1$.
\end{minipage}
\end{figure}

Figure \ref{fig_wages} illustrates the wage dynamics during the transition, plotting them as a function of AI productivity, $A_a$. The blue solid line ($\alpha^S_{a}=0$) depicts the case in which AI is produced exclusively by complementary-labor.  In that case, the transition never terminates because $\overline L^S$ cannot be reallocated to the AI sector, and its wage simply continues to fall.
In contrast, the orange dashed line represents the case in which the AI sector employs both labor types ($\alpha^S_{a}=0.5$).  In that case, wages of substitutable labor fall during the transition and eventually stabilize and become a constant fraction of GDP again once the post-AI phase is reached. 
The center shows that complementary labor benefits unambiguously from AI advancement, with wages rising throughout the transition and afterwards in both cases. The right panel shows that final goods sector specific labor gains from the expanding AI adoption, but simply as a constant fraction of GDP---neither displaced by nor attracted to AI production.

Proposition~\ref{prop_wagegap} characterizes the corresponding changes in wage inequality between different labor types in response to a change in $A_a$ during transition.

\begin{proposition}
	\label{prop_wagegap} {\rm [Growing Wage Inequality]}
During the transition phase ($A_a^{*} < A_a < A_a^{**}$)  relative wage changes are:
\[
\frac{d \ln \left(\frac{w_C}{w_S}\right)}{d \ln A_a}
>
\frac{d \ln \left(\frac{w_F}{w_S}\right)}{d \ln A_a}
>
\frac{d \ln p_fY_f}{d \ln A_a} =E_a=\frac{p_aY_a}{p_fY_f}
> 0 >
\frac{d \ln \left(\frac{w_F}{w_C}\right)}{d \ln A_a}.
\]
In particular, relative wage of complementary- to substitutable-labor is completely
determined by the productivity of the AI technology and relative shares of labor in AI production, and is independent of relative labor supply $\overline{L}^C/\overline{L}^S$:
		\[
		\frac{w_C}{w_S}=(A_{a} \Phi)^{1/\alpha^C_{a}}.
		\]
\end{proposition}

Proposition \ref{prop_wagegap} shows that a rise in $A_a$ increases aggregate output with an elasticity equal to AI's expenditure share: $\frac{d \ln p_fY_f}{d \ln A_a} = E_a=\frac{p_aY_a}{p_fY_f}$. This is a direct analog of Hulten's theorem \citep{hulten1978growth}, which states that the effect of a sector's productivity gain on aggregate output is measured by that sector's share in GDP. Here, however, the expenditure share is itself endogenous and rising in $A_a$, so the marginal impact of AI productivity on the economy grows as AI becomes more widely adopted. The proposition shows that the productivity gain reshapes relative wages unequally across the three labor types.

As a result, there is a clear hierarchy of winners and losers during the transition:
\begin{itemize}
    \item \textit{The Winner ($\overline L^C$):} Complementary labor benefits the most. This is because $\overline L^C$ workers face a ``dual demand'': they are needed to use the AI in the final goods sector, and they are needed to \textit{build} the AI in the automation sector.
    
    \item \textit{The Baseline ($\overline L^F$):} Final goods specific labor acts as the benchmark. Since $w_F = \alpha^F_{f} Y_f / \overline L^F$, the wage of this group tracks aggregate output growth one-for-one ($\frac{d \ln w_F}{d \ln A_a} = \frac{d \ln Y_f}{d \ln A_a}$). They gain from the expanding economy but do not receive the extra premium that $\overline L^C$ enjoys.
    
    \item \textit{The Loser ($\overline L^S$):} Because they are perfect substitutes for AI, their wage is tied to the price of AI ($w_S = p_a$). As AI productivity rises, the efficiency-adjusted cost of the AI falls, dragging $w_S$ down both at the absolute level and relative to the other groups.
\end{itemize}

\begin{figure}[t]
\centering
\includegraphics[width=0.6\textwidth]{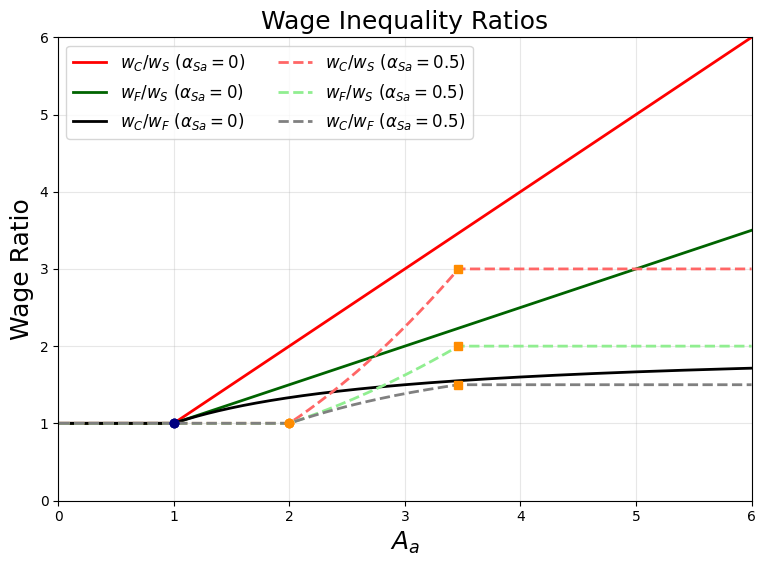}
\caption{AI Adoption and Wage Inequality}
\label{fig_wage_inequality}
\end{figure}

Figure~\ref{fig_wage_inequality} illustrates the wage inequality dynamics characterized in Proposition~\ref{prop_wagegap}, using the parameters from Figure \ref{fig_wages}. 
During the transition, complementary labor pulls ahead of all other groups. The ratio $w_C/w_S$ (red) rises most sharply, with elasticity $\frac{1}{\alpha^C_{a}}$ (i.e., $\frac{d \ln \left(\frac{w_C}{w_S}\right)}{d \ln A_a}=\frac{1}{\alpha^C_{a}}$. The ratio $w_F/w_S$ (green) increases at a slower rate, with elasticity $\frac{\Lambda}{1+\Lambda} \cdot \frac{1}{\alpha^C_{a}}$, where $\Lambda = \frac{\overline{L}^C}{\overline{L}^S} (A_{a}\Phi)^{1/\alpha^C_{a}}=\frac{w_C\overline{L}^C}{w_S\overline{L}^S} $ equals the ratio of complementary labor's equilibrium income share to that of substitutable labor. The ratio $w_C/w_F$ (black) rises at the slowest rate (in this example), with elasticity $\frac{1}{1+\Lambda} \cdot \frac{1}{\alpha^C_{a}}$. As shown in the Appendix, the ranking between the changes in $w_F/w_S$ and $w_C/w_F$ depends on whether $\Lambda$ exceeds unity: when $\Lambda > 1$, final goods specific labor gains faster relative to substitutable labor than complementary labor gains relative to final goods specific labor. Lastly, when $\alpha^S_{a}=0$ (solid lines), $w_C/w_S$ and $w_F/w_S$ grow without bound as the economy remains in endless transition, while $w_C/w_F$ rises toward the finite asymptote $\frac{\overline{L}^F}{\overline{L}^C}\frac{\alpha^S_{f}+\alpha^C_{f}}{\alpha^F_{f}}$. When $\alpha^S_{a}=0.5$ (dashed lines), the ratios stabilize upon reaching the post-AI regime, locking in a permanent redistribution of income.

The labor-supply invariance of the wage ratio in Proposition~\ref{prop_wagegap}---that $w_C/w_S$ is independent of the supplies $\overline{L}^C,\overline{L}^S$---is explained as follows.

The substitutable wage equals the AI price, $w_S = p_a$, and so the wage ratio is pinned down by the AI cost function alone, $w_C/w_S = w_C/p_a = (A_a\Phi)^{1/\alpha^C_a}$. When either group's size changes, the economy reallocates labor into or out of AI production---shifting AI adoption and the sectoral split of labor---so that this ratio is restored. Aggregate output rises or falls with the sizes of the labor forces, but the wage ratio is unchanged.

This labor-supply-invariance result delivers a stark policy implication.  Traditional supply-side interventions aimed at narrowing wage gaps by converting substitutable labor into complementary labor that one might expect to change wages, in this case do not affect inequality in wages between the two groups.  Effectively, although it decreases the wage of the complementary workers, it also decreases the price of automation in parallel.  
Such policies
are still valuable, since they increase the fraction of the population receiving the higher wage, but they do not change the wage-gap itself.  

Note also that the labor-supply invariance holds specifically during the transition, $A^{*}<A_a<A^{**}$. Because both thresholds in \eqref{threshold} scale with $(\overline{L}^{S}/\overline{L}^{C})^{\alpha^{C}_{a}}$, a change in relative supply also shifts $A^{*}$ and $A^{**}$: raising $\overline{L}^{S}/\overline{L}^{C}$ delays onset and completion, and a large enough change moves a fixed $A_a$ out of the regime---into the pre-AI or post-phase, where $w_C/w_S$ is again supply-dependent. The labor-supply factor $(\overline{L}^{S}/\overline{L}^{C})^{\alpha^{C}_{a}}$ cancels in the ratio $A^{**}/A^{*}$. Relative supply therefore relocates the invariance window without changing its width, which is pinned down by technology alone.\footnote{In Section~\ref{transsec} in the Online Supplementary Appendix, we provide an additional analysis on the width of the transition phase as a function of sectoral share parameters.}

\begin{figure}[t]
\centering
\includegraphics[width=0.9\textwidth]{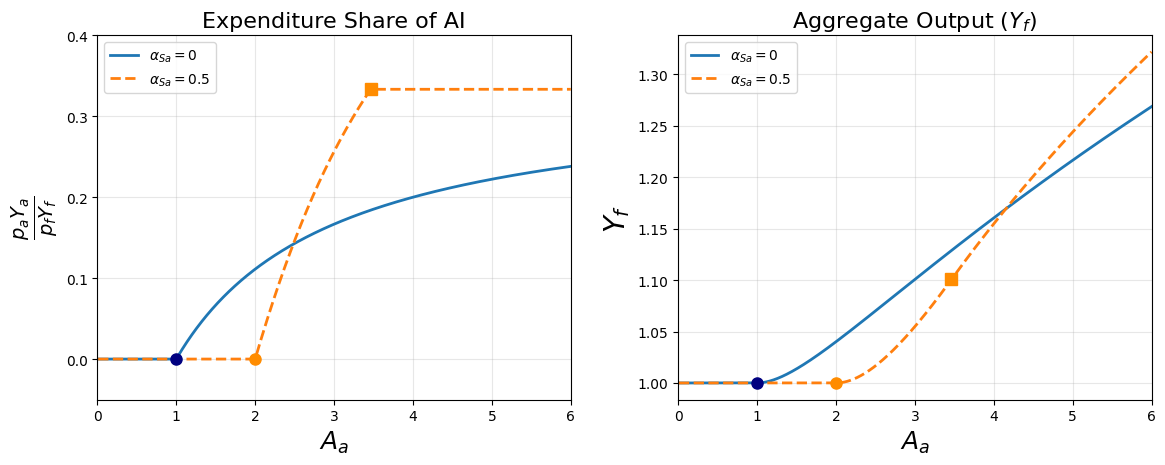}
\caption{The AI Expenditure Share in Final Goods Sector and the Aggregate Output}
\label{fig_ai_share_output}
\vspace{0.2cm}
\end{figure}

Figure~\ref{fig_ai_share_output} illustrates the aggregate dynamics of AI adoption  for the parameters from Figures \ref{fig_wages} and \ref{fig_wage_inequality}. 
The left panel shows that the AI expenditure share in final goods production rises monotonically during the transition. When $\alpha^S_{a}=0$ (solid blue), this share continues to grow towards $\alpha^S_{f}=1/3$ as the economy remains in endless transition; while when $\alpha^S_{a}=0.5$ (dashed orange), the share reaches $\alpha^S_{f}=1/3$ and stays constant at that level upon reaching the post-AI phase.

The right panel reveals that the output curves cross. Early in the transition, the $\alpha^S_{a}=0$ economy (solid blue) produces higher output because AI adoption begins sooner ($A_a^*=1$ versus $A_a^*=2$). However, as AI productivity increases, the $\alpha^S_{a}=0.5$ economy (dashed orange) eventually overtakes the blue line. This crossover reflects two reinforcing mechanisms. First, when $\alpha^S_{a}>0$, substitutable labor can be redeployed to AI production, allowing the economy to exploit higher productivity gains, rather than leaving these workers competing with increasingly efficient AI in the final goods sector. Second, the transition is faster: the economy reaches the post-AI regime at $A_a=A^{**}$, where output grows at the maximal rate with respect to $A_a$. In contrast, when $\alpha^S_{a}=0$, substitutable labor cannot contribute to AI production, limiting the aggregate gains from automation. Beyond the crossover point (around $A_a \approx 4$ in this example), the gap widens as AI productivity rises further---underscoring that the long-run benefits of AI depend critically on whether displaced workers can participate in building the technology that replaces them.

\section{Pareto-Improving AI: Tax Redistribution}

\label{SEC:INCOME_TAX}

Given that substitutable  labor not only suffers relative wage losses, but also absolute wage losses as the economy grows with the increased AI productivity, it is clear that the resulting inequality can become extreme.
Since overall GDP is increasing, it is also clear that with some redistribution, all forms of labor can 
be made better off as a result of the increased productivity of AI.\footnote{For an example where a major technological disruption caused a call for redistribution and ended up having important implications for politics, see the discussion by \cite{anelli2026} of the Free-Silver Populism movement in the 1890s in the U.S.}

We abstract away from final goods-specific labor by setting $\overline{L}^F = 0$ and $\alpha^F_{f} = 0$, because, as shown in Proposition~\ref{prop:equilibrium_wages}, their wage is proportional to aggregate output: $w_F = \alpha^F_{f} Y_f / \overline{L}^F$. Since aggregate output rises with AI productivity, these workers' benefit is proportional to gains from AI adoption and therefore do not require redistributive support. The more policy-relevant tension is between the group whose income grows faster than output (complementary labor) and the group whose income declines during the transition (substitutable labor).\footnote{As will become clear, one could include the third group with similar results but in three dimensions and an additional taxation term.  Working with two groups renders the analysis simpler in terms of figures and focuses attention on the groups who gain and lose relative to GDP growth, and so between-which transfers are most natural.} 

Thus, we simplify the model as follows:
\begin{equation*}
Y_a = A_a (L^S_{a})^{\alpha^S_{a}} (L^C_{a})^{\alpha^C_{a}}, \qquad \alpha^S_{a} + \alpha^C_{a} = 1,
\end{equation*}
\begin{equation*}
Y_f = A_f [L^S_{f} + X_a]^{\alpha^S_{f}} (L^C_{f})^{\alpha^C_{f}}, \qquad \alpha^S_{f} + \alpha^C_{f} = 1.
\end{equation*}

The analysis is tractable and important.
It is tractable for two reasons.  

One is that in a competitive market with inelastic labor supply, taxes and subsidies on wages are non-distortionary.   Thus it does not matter whether
taxes and subsidies are implemented lump sum or as a function of income or the wage.  
Thus, we can equivalently model any tax policy via taxes and subsidies on wages, which keeps notation and the analysis simple. 
We let
$\widetilde{w}_C = w_C(1-\tau_C)$ and $\widetilde{w}_S=w_S(1-\tau_S)$ be the after tax/subsidy wages, where $\tau_C,\tau_S$ are the tax rates (with negative rates being a subsidy).

The second reason is that the wages satisfy a simple equation:
\begin{equation*}
\widetilde{w}_C \overline{L}^C +\widetilde{w}_S \overline{L}^S = Y_f.
\end{equation*}

These two observations mean that a (balanced) tax policy is equivalent to choosing 
$\widetilde{w}_C ,\widetilde{w}_S $, and these pairs define a linear Pareto frontier whose labor supply weighted sums equal total GDP.

Given that AI adoption raises ${w}_C$ and reduces ${w}_S$, a Pareto-improving AI requires transfers from complementary to substitutable labor ($\tau_C>0$ and $\tau_S<0$). Note that, since taxes and subsidies are implemented via wages, and labor is supplied inelastically, effective wages to firms are unchanged and the equilibrium is as before, as are the thresholds for AI adoption ($A^*$) and full transition ($A^{**}$).

Throughout the remainder of this section, to analyze such tax policies we 
consider some starting point 
$A_a^0$ that a government views as a benchmark, 
such that $A^*\leq A_a^0 < A^{**}$, and then a productivity transition to some $A_a^1$, $A^{**} \geq A_a^1> A_a^0$.

We know from the results in Section \ref{SEC:TRANSITION} that untaxed/subsidized wages rise for $C$-workers and fall for $S$-workers.   
The Pareto improving after-tax/subsidy wages are those pictured by the wavy-marked subsegment of the new Pareto Frontier in Figure \ref{fig_paretotax}.

\begin{figure}[t]
\centering
\includegraphics[width=0.9\textwidth]{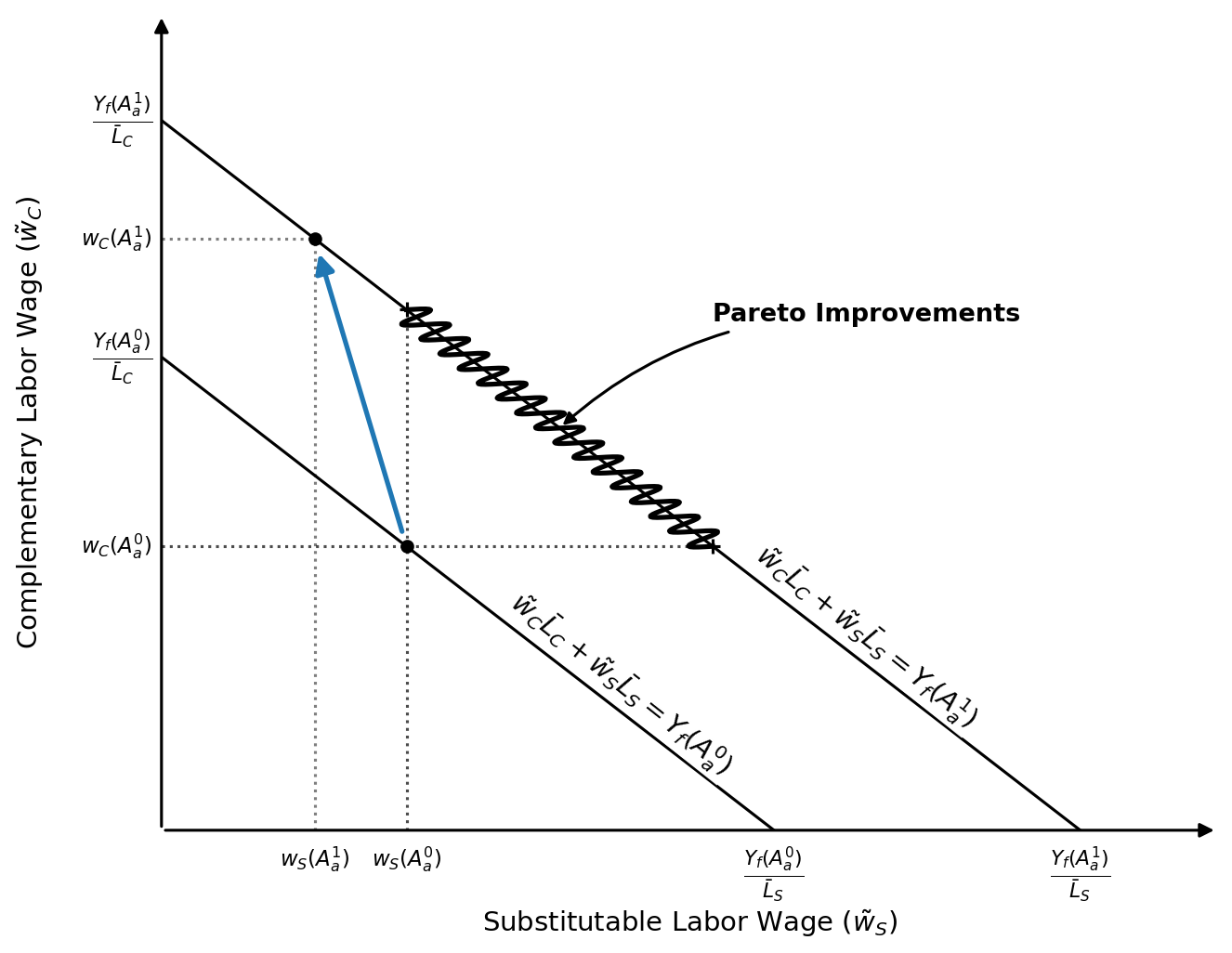}
\caption{Pareto Improving Tax Policies:  An economy starts at some productivity $A_a^0$, 
such that $A^*\leq A_a^0 < A^{**}$, and transitions to some $A^{**} \geq A_a^1> A_a^0$. Untaxed/subsidized wages rise for $C$-workers and fall for $S$-workers.   The Pareto improving after-tax/subsidy wages are those pictured by the wavy-marked subsegment of the new Pareto Frontier.}
\label{fig_paretotax}
\end{figure}

We can use the notion of Pareto improving tax policies to understand how relatively large taxes need to be to implement various targets. 
Of course, there are many possible target tax policies as there is a whole Pareto frontier.  

We first examine the minimal necessary tax policy that offers a Pareto improvement.   This is the tax that sets $\widetilde{w}_S(A_a^1)= w_S(A_a^0)$, so that the $S$-workers are just as well off after the transition {to $A_a^1$} as at the benchmark {at $A_a^0$}.  

The interesting result is that the tax rate that is needed to achieve this, 
$\tau_C (A_a^1)  $,
is {\sl single-peaked} in $A_a^1 \leq A^{**}$.
Initially the necessary tax rate increases, but then it eventually decreases.
The necessary tax rate starts out at zero, and it then increases as substitutable workers' wages fall and so taxes are needed to cover the shortfall,  and the necessary tax rate depends on the relative sizes of the labor pools.  Eventually, however, the decrease in the substitutable  wage slows as the wage approaches its post transition limit.  This, coupled with 
the fact that $w_C(A_a)$ increases faster than $w_S(A_a)$ falls (given that the total output $Y_f$ is increasing in $A_a$ for $A_a>A^*$ as shown in Proposition \ref{prop_wagegap}), means that the necessary tax rate eventually falls.  

Next, consider a tax target of splitting GDP so that labor splits GDP in some ratio $\gamma < {w}_C^1/{w}_S^1 $, so that some tax is required.
The target tax is to have
$\widetilde{w}_C^1/\widetilde{w}_S^1 =\gamma $. This policy requires a tax rate that is
{\sl increasing} in $A_a$. 
Intuitively, as AI adoption rises, the wage ratio between $C$-workers and $S$-workers widens but the target wage ratio is constant, so a larger share of $C$-workers' wages must be taxed to keep the after-tax ratio at $\gamma$.

These results are collected in the following proposition.

\begin{proposition}
\label{taxrate}
The minimal tax rate on complementary workers, $\tau_C>0$, required to hold substitutable workers at their benchmark wage (so that $\widetilde{w}_S^{1}=w_S^{0}$) is single-peaked in the productivity gain: it is zero at $A_a^{1}=A_a^{0}$, strictly increasing for $A_a^{0}<A_a^{1}<\widehat{A}_a$, and strictly decreasing for $A_a^{1}>\widehat{A}_a$, where
\[
\widehat{A}_a=A_a^{0}\,(\alpha^S_{f})^{-\alpha^C_{a}/\alpha^C_{f}} .
\]
If $\widehat{A}_a\geq A^{**}$ the required tax rate rises throughout the entire transition. By contrast, the tax rate required to hit some proportional split of GDP (so that $\widetilde{w}_C^{1}/\widetilde{w}_S^{1}=\gamma<w_C^{1}/w_S^{1}$) is strictly increasing in the productivity gain.
\end{proposition}

Since total GDP rises with $A_a$, the tax that holds $S$-workers at their benchmark starting wage ($\widetilde{w}_S^1 = w_S^0$) is always Pareto-improving: with $S$-workers' wages held fixed, the net gain goes to $C$-workers, so $\widetilde{w}_C^1 > w_C^0$.

For $\gamma$ sufficiently close to $ w_C^0/w_S^0$, the proportional policy (so that $\widetilde{w}_C^{1}/\widetilde{w}_S^{1}=\gamma<w_C^{1}/w_S^{1}$) is also Pareto-improving. To see this, for $\gamma = w_C^0/w_S^0$, the policy is Pareto-improving since the total wage bill increases when transitioning from $A_a^0$ to $A_a^1$ and hence both after-tax wages scale up with the tax. For $\gamma$ that deviates too substantially from this initial ratio, in either direction, the after tax wage of one of the two groups will be below its initial level.

\paragraph{Comparative Statics regarding Production Technologies and Labor Supplies.}

Beyond how Pareto improving tax rates depend on the productivity gain, we can also see how they change with the relative sizes of the labor force, as well as parameters of the production functions.

\begin{proposition}
\label{taxrate2}
Both the minimal tax rate on complementary workers required to hold substitutable workers at their benchmark wage and the tax rate required to hit some proportional split of GDP are increasing in $\frac{\overline{L}^S}{\overline{L}^C}$.

The minimal tax rate on complementary workers, $\tau_C>0$, required to hold substitutable workers at their benchmark wage is decreasing in the relative importance of substitutable labor in final goods production, $\alpha_f^S$; while the tax rate required to hit some proportional split of GDP is independent of $\alpha_f^S$.
\end{proposition}

The importance of substitutable labor in the AI-sector, $\alpha^S_{a}$, enters both targets in more complicated ways and so produces ambiguous comparative statics.  
In particular, note that $w_C/w_S = (A_a\Phi)^{1/\alpha^C_{a}}$, where $\alpha^C_{a} = 1 - \alpha^S_{a}$ and $\Phi = (\alpha^C_{a})^{\alpha^C_{a}}(\alpha^S_{a})^{\alpha^S_{a}}$. Whether a higher $\alpha^S_{a}$ widens or narrows the wage ratio depends on the productivity level: $\partial\ln(w_C/w_S)/\partial\alpha^S_{a} = \ln(A_a\alpha^S_{a})/(\alpha^C_{a})^2$, which is positive when $A_a\alpha^S_{a} > 1$ and negative when $A_a\alpha^S_{a} < 1$.\footnote{To see this, note that $\ln(w_C^0/w_S^0) = \frac{1}{\alpha^C_{a}}[\ln A_a^0 + \alpha^C_{a}\ln\alpha^C_{a} + \alpha^S_{a}\ln\alpha^S_{a}]$. Differentiating with respect to $\alpha^S_{a}$ gives $\partial\ln(w_C^0/w_S^0)/\partial\alpha^S_{a} = \ln(A_a^0\,\alpha^S_{a})/(\alpha^C_{a})^2$, which changes sign at $A_a^0\alpha^S_{a} = 1$.} For the $\gamma$-ratio target, the required tax increases with $w_C^1/w_S^1$, so it rises with $\alpha^S_{a}$ when $A_a^1\alpha^S_{a} > 1$ and falls when $A_a^1\alpha^S_{a} < 1$. For the wage target $\widetilde{w}_S = w_S^0$, $\alpha^S_{a}$ changes the wage ratio at both $A_a^0$ and $A_a^1$, affecting the deficit and the tax base simultaneously, and the net effect critically depends on the levels $A_0$ and $A_1$. 

Overall, our findings emphasize that when AI is produced by labor, redistribution policies are not static: optimal tax rates depend systematically on the productivity of AI, and must therefore evolve as the technology and the economy progress.

\section{Tax Policy with Monopoly AI Production}

\label{SEC:MONOPOLY}

We now consider what happens if AI production becomes monopolized. 
To keep the analysis uncluttered, we presume that the final goods sector remains competitive and maintain the production functions used in Section~\ref{SEC:INCOME_TAX}. 

Throughout this section, we consider starting at a level $A_a^{0}$ and transitioning to higher level
$A_a^{1}$ in region $A^{*}\le A_a^{0}<A_a^{1}\le A^{**}_{m}$.

First, note that without any government intervention, 
a monopolist AI producer restricts output and 
slows the transition to AI.   This in turn slows the decline of the substitutable wage, and also results in lower complementary wages and aggregate output compared to the competitive case.  The monopolist then captures some AI surplus as pure profit. This alters both the distribution of gains from AI and the instruments available for redistribution. 

In particular, a key difference from the competitive case is that total output is now divided among three claimants rather than two:
\begin{equation*}
  Y_f^{m} =  \widetilde{w}_S^{m}\,\overline{L}^S + \widetilde{w}_C^{m}\,\overline{L}^C + \widetilde{\Pi}^{m} ,
    \label{eq:budget_mono}
\end{equation*}
where $\widetilde{\Pi}^{m}$ is the monopolist's after-tax profit.

First, let us note that if the planner either puts in a price control that pegs the AI price to the competitive level, or puts in a full tax on monopoly profits, then the monopolist earns no rents and the problem simplifies to the competitive case.  There are many reasons to expect that full price controls and full monopoly taxes are not policy tools that would be chosen by a government, and so we also analyze what happens when a government instead employs partial taxes on monopoly profits that target worker welfare levels similar to what we analyzed in the competitive case.

The planner now has two tax instruments that can be used to finance a subsidy to substitutable workers: a tax on monopoly profit at rate
$\tau_\pi$ and a tax on complementary wages at rate $\tau_C$.

We model the AI producer as a monopolist, within an otherwise general equilibrium framework. This differs from a
partial-equilibrium formulation in which the monopolist takes $w_S^m$ and $w_C^m$
as fixed and faces AI demand only as a function of $p_a$. That formulation is not
well posed here: because AI and substitutable labor are perfect substitutes in
final-good production, wages adjust in equilibrium to the monopoly price.\footnote{If one examines demand for AI with fixed wages, it would be either 0 or full AI depending on whether $p_a$ was below or above $w_S$.  This artificial discontinuity would be because the monopolist did not anticipate wages changing to equilibrate to $p_a$.}    
Thus, the AI producer chooses its price and input levels while anticipating the
equilibrium response of wages, AI demand, and labor allocations induced by the
competitive final-good sector and labor-market clearing.

The monopolist's profit, $\Pi^m$,is
$$
\Pi^m=p_aY_a-w_S^mL^S_a-w_C^mL^C_a = p_a A_a(L^S_a)^{\alpha^S_a}(L^C_a)^{\alpha^C_a}
-w_S^mL^S_a-w_C^mL^C_a.
$$
The variables that the monopolist controls are
$p_a,L^S_a,L^C_a$,
and these are chosen 
to maximize after-tax profit
$(1-\tau_\pi)\Pi^m$, subject to the final-good sector's equilibrium reaction and market clearing conditions. 

In particular, unlike a partial equilibrium monopolist, here the monopolist anticipates the general equilibrium effects of its choices, as 
those in turn impact the demand for AI as a function of the price.   
In particular, the choice of labor inputs by the monopolist influences the market clearing wage $w_C$, which in turn influences the monopolist's profit.  Essentially, in this problem the monopolist has both monopoly and monopsony power, and by strategically choosing its inputs as well as the price of its output, it impacts the overall profits.  Intuitively, by using more substitutable labor itself, the monopolist can drive up the demand for AI for any given price level.

In Supplementary Appendix Section \ref{sec:monopoly_struc}, we consider an alternative monopoly
formulation in which the monopolist only chooses its price/output to maximize profit, chooses the
cost-minimizing labor mix given the resulting wages.  We compare this formulation with the GE-monopoly
formulation used here. The GE-monopoly formulation here has this extra monopsony-like aspect that the
firm chooses input quantities while internalizing their effect on equilibrium
wages. The alternative formulation in the appendix leads to a different allocation and 
lower profit, because it restricts the firm to cost-minimizing input bundles,
but the main qualitative implications remain similar.

During the transition, the final-good sector uses both AI and
substitutable labor, which implies $w_S^m=p_a$. 
Anticipating $w_S^m=p_a$, the
profit can be written as a choice over $L^S_a,L^C_a$, with
\[
\Pi^m
=
w_S^m\left(A_a(L^S_a)^{\alpha^S_a}(L^C_a)^{\alpha^C_a}-L^S_a\right)
-w_C^mL^C_a,
\]
where $w_S^m$ and $w_C^m$ are the endogenous functions of $(L^S_a,L^C_a)$ defined
by the final-good sector's reaction, with $w_S^m=p_a$ imposed along the interior
equilibrium locus. 

Importantly, in this reduced-form objective $w_S^m$ and $w_C^m$ are not fixed
wages; they are the endogenous equilibrium wage functions induced by the
final-good sector's reaction and labor-market clearing. The AI price is embedded in the inverse general-equilibrium relation \[ p_a=w_S^m(L_a^S,L_a^C). \]

In Appendix~\ref{app:monopoly},  we provide the explicit full monopoly mixed with
general-equilibrium formulation and derive the corresponding first-order
conditions. Lemma~\ref{lem:mon_alloc} in the Appendix shows that the monopolist restricts scale relative to competition: $L^{C,m}_a<L^{C,comp}_a$, $L^{S,m}_a<L^{S,comp}_a$ and $Y_a^{m}<Y_a^{comp}$.

For any $\tau_\pi < 1$, the after-tax objective is a positive scaling of $\Pi^{m}$ and the optimal allocation is the same as without taxation. When $\tau_\pi = 1$, the monopolist is indifferent across all production levels, and we assume it chooses the competitive allocation.

We carry out the same exercise as in the competitive case, considering an economy under monopoly at benchmark productivity $A_a^0$, and suppose AI productivity increases to $A_a^1 > A_a^0$. 
The total surplus from the technology transition from $A_0$ to $A_1$ decomposes as :
\begin{equation*}
    Y_f^{m,1} - Y_f^{m,0} = \underbrace{(w_C^{m,1} - w_C^{m,0})\overline{L}^C}_{\text{C-worker gain}} - \underbrace{(w_S^{m,0} - w_S^{m,1})\overline{L}^S}_{\text{S-worker deficit}} + \underbrace{(\Pi^{m,1} - \Pi^{m,0})}_{\text{profit gain}}.
    \label{eq:surplus_decomp}
\end{equation*}

The profit tax revenue partially (or fully) finances the transfer to substitutable workers, and any remaining deficit is covered by a tax on complementary wages. Pareto-improving AI for workers requires that no worker group is worse off after transfers at $A_a^1$ than at the benchmark $A_a^0$, and at least one being better-off (one of the first two inequalities being strict), 
and after-tax monopoly profit stays non-negative: \begin{equation*} 
\widetilde{w}_S^{m} \geq w_S^{m,0}, \qquad \widetilde{w}_C^{m} \geq w_C^{m,0}, \qquad \widetilde{\Pi}^{m} \geq 0. \label{eq:pareto_constraints_mono} 
\end{equation*}

The monopolist's output level does not admit a closed-form solution. We therefore combine analytical results with simulations to describe the monopoly equilibrium.

This analysis establishes four results.

\begin{itemize}
\item First, a tax on monopoly profits alone cannot finance Pareto-improving redistribution until AI productivity is sufficiently high: near the onset of adoption the monopolist's rents fall short of the substitutable-worker deficit, so the profit tax must be supplemented by a tax on complementary workers. 
\item Second, a complementary-worker tax alone is sufficient, just as under competition, because the total wage bill rises throughout the transition. Under monopoly this requires only that AI productivity, $A_a$, grows proportionally faster than the markup $\mu = p_a/c$---so that the monopolist passes part of each productivity gain into output rather than absorbing it as a higher markup. While we show analytically in Proposition \ref{prop:ctax_mono} that this is the sufficient condition, we verify through simulations that it holds throughout the transition. 
\item Third, removing the monopolist's market power, through a price cap or full profit taxation, restores the competitive first-best; when full taxation is infeasible, partial profit taxation leaves complementary workers better off than under unregulated monopoly but still worse off than under competition, so the output loss from monopoly is borne by the workforce. 
\item Fourth, for the alternative target of a fixed relative-wage split between the two labor types, the required tax is smaller than under competition, since monopoly compresses wage inequality.
\end{itemize}

Let us begin with the first result mentioned above.  Suppose the planner institutes a policy that targets maintaining substitutable workers at the benchmark wage of $\widetilde{w}_S = w_S^{m,0}$. 
Under the constraint $\widetilde{\Pi}^{m} \geq 0$ the planner may tax profit up to $\tau_\pi = 1$, so the profit tax alone ($\tau_C = 0$) achieves Pareto-improving AI if and only if\footnote{With $\tau_C=0$, complementary workers keep their untaxed wage $w_C^{m,1}$, so the remaining Pareto constraint $\widetilde{w}_C^{m}\geq w_C^{m,0}$ reduces to $w_C^{m,1}\geq w_C^{m,0}$; which is satisfied because C-worker wage is increasing in $A_a$ throughout the transition.}
\begin{equation}
    \Pi^{m,1} \geq (w_S^{m,0} - w_S^{m,1})\overline{L}^S.
\end{equation}
Figure~\ref{fig:mono_profit_deficit} shows that this fails near $A^*$---the monopolist production and its rents are too small to cover the deficit.  The profits are sufficient to cover the wage shortfall only once $A_a$ passes a crossing point (here around $A_a \approx 5.9$ for $\alpha^S_{a}=0$). Monopoly rents therefore accumulate slowly, and profit taxation suffices on its own only at high AI productivity.

\begin{figure}[t]
\centering
\includegraphics[width=0.52\textwidth]{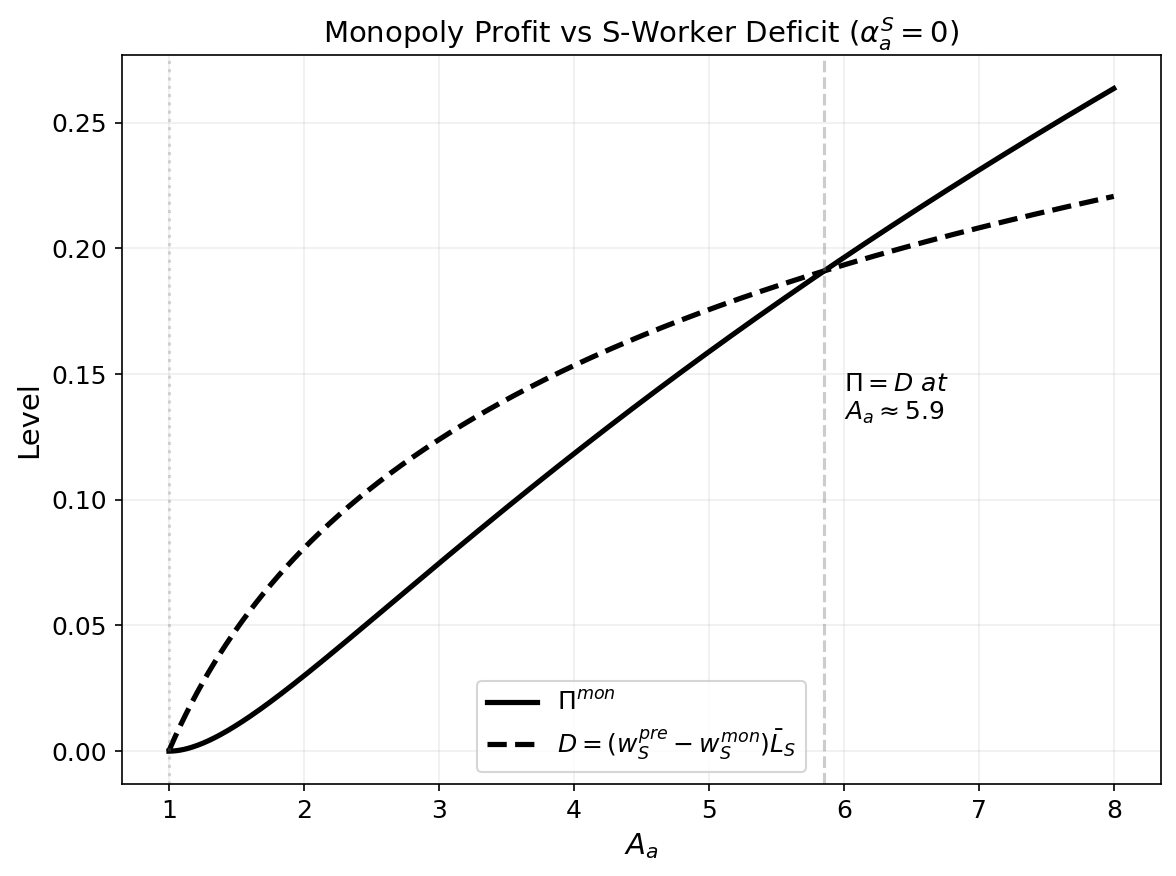}
\caption{monopoly Profit vs Substitutable-Worker Deficit}
\label{fig:mono_profit_deficit}

\vspace{0.1cm}
\begin{minipage}{0.95\textwidth}
\footnotesize
\textit{Notes:} Monopoly profit $\Pi^{m}$ (solid) vs.\ the substitutable-worker deficit $D = (w_S^{pre} - w_S^{m})\overline{L}^S$ (dashed); the vertical line marks the crossing $\Pi^{m} = D$. Parameters: $\alpha^S_{a} = 0$, $\alpha^S_{f} = \alpha^C_{f} = 0.5$, $\overline{L}^S = \overline{L}^C = 1$, $A_f = 1$.
\end{minipage}
\end{figure}

The intuition behind this is that near the onset of AI adoption, monopoly profit is second order in the increase in productivity relative to
adoption threshold $A^*$, while the substitutable-worker wage deficit is first order. Hence a profit tax alone cannot finance
Pareto-improving transfers arbitrarily close to the onset; it must be
supplemented by a tax on complementary workers. The simulations show that, as
AI productivity rises further, monopoly profit eventually overtakes the deficit
and profit taxation alone suffices.

Next, let us consider what happens when there is no tax on the monopolist
and only the complementary workers are taxed.  

A complementary-worker tax alone ($\tau_\pi = 0$) achieves Pareto-improving AI if and only if the complementary-worker gain covers the deficit,
\begin{equation}
\label{wagetaxonly}
    (w_C^{m,1} - w_C^{m,0})\overline{L}^C > (w_S^{m,0} - w_S^{m,1})\overline{L}^S,
\end{equation}
which, as in the competitive case, is equivalent to the total wage bill increasing. Figure~\ref{fig:mono_output_gains} confirms that the surplus $G_C^{m}-D^{m}$---smaller than under competition, since output is lower and part of the surplus is taken as profit---remains strictly positive throughout the transition, so the complementary-worker tax alone is always sufficient. As shown in the next result the total wage bill rises in $A_a$ whenever $A_a/\mu$ rises, where $\mu = p_a/c$ is the markup of the monopolist.

\begin{figure}[t]
\centering
\includegraphics[width=0.9\textwidth]{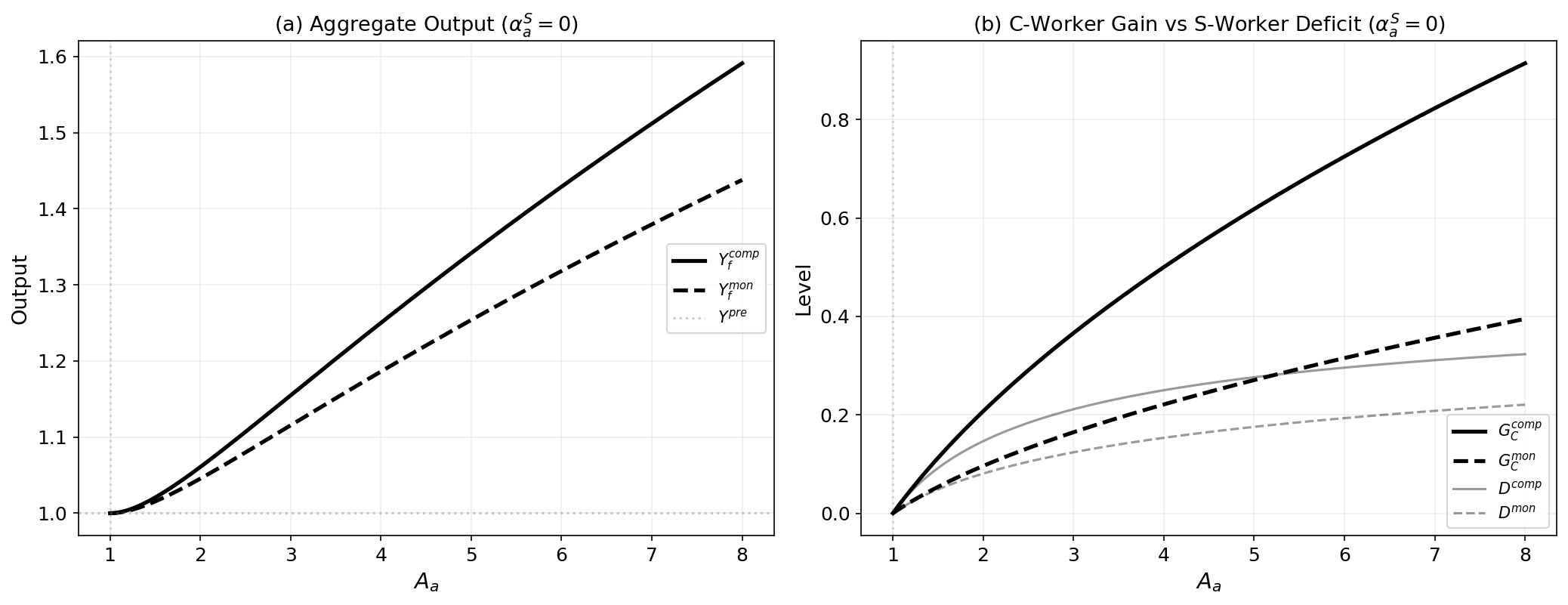}
\caption{Output Levels and Workers' Gains and Losses under Competition vs. Monopoly}
\label{fig:mono_output_gains}

\vspace{0.1cm}
\begin{minipage}{0.95\textwidth}
\footnotesize
\textit{Notes:} Panel~(a): aggregate output under competition (solid) and monopoly (dashed). Panel~(b): complementary-worker gain $G_C = (w_C - w_C^{pre})\overline{L}^C$ (thick) vs.\ substitutable-worker deficit $D = (w_S^{pre} - w_S)\overline{L}^S$ (thin) under both structures; $G_C > D$ at every $A_a > A^*$. Parameters: $\alpha^S_{f} = \alpha^C_{f} = 0.5$, $\alpha^S_{a} = 0$, $\overline{L}^S = \overline{L}^C = 1$, $A_f = 1$.
\end{minipage}
\end{figure}

Let
$\mu=p_a/c\ge 1$ denote the monopolist's markup of the AI price, $p_a$, over unit cost, $c$.

\begin{proposition}[Sufficiency of the complementary-worker tax under monopoly]
\label{prop:ctax_mono}
Consider a
planner targeting maintaining substitutable workers at wages $\widetilde{w}_S=w_S^{m,0}$.  A tax
on complementary wages alone ($\tau_\pi=0$, $\tau_C>0$) achieves Pareto-improving AI 
for workers if and
only if AI productivity grows proportionally faster than the markup,
\[
\frac{A_a^{1}}{A_a^{0}} > \frac{\mu^{1}}{\mu^{0}}.
\]

\end{proposition}

The proposition shows that a complementary-worker tax alone suffices for Pareto improvement for workers whenever the monopolist passes a positive fraction of each productivity gain into output rather than the markup. Whether the required tax rate is higher or lower than under competition then depends on how the markup moves with $A_a$. This can be explained as follows. The required rate equals the substitutable-worker deficit
being financed divided by the complementary-wage base being taxed. Monopoly unambiguously
shrinks the base---restricting AI production holds the complementary wage below its
competitive level, $w_C^{m}<w_C^{comp}$---and this alone raises the required rate; its
effect on the deficit, by contrast, is exactly what the markup governs. Since
$w_S^{m}=w_S^{comp}\,\mu^{\alpha^C_f/\alpha^C_a}$ at any given $A_a$, a roughly constant
markup scales the substitutable wage up at both $A_a^{0}$ and $A_a^{1}$ and leaves the
deficit no smaller, so the required rate exceeds its competitive counterpart by the factor
$\mu^{1/\alpha^C_a}$; a markup that rises steeply with $A_a$ instead lifts $w_S$ by more at
$A_a^{1}$ than at $A_a^{0}$, compressing the deficit enough that the required rate can fall
below the competitive level.

Proposition~\ref{prop:ctax_mono} shows that this condition is necessary and sufficient,
and Lemma~\ref{lem:markup_slower_than_productivity} in the Appendix proves that $\frac{A_a^{1}}{A_a^{0}}>\frac{\mu^{1}}{\mu^{0}}$ holds 
throughout the interior of the monopoly transition. The simulation in Figure~\ref{fig:mono_output_gains} illustrates this result.

Next let us consider the third result mentioned above.  The planner can remove the monopolist's pricing power directly, by capping the AI price at $p_a^{comp}$ or taxing profit at $100$ percent. Either restores the competitive allocation---and hence reduces the problem to that of Section~\ref{SEC:INCOME_TAX}, with the full surplus available---so, combined with the competitive complementary-worker tax, both are Pareto-improving for workers. 
However, when full taxation is infeasible and only $\tau_\pi < 1$ is available, the monopolist keeps $(1-\tau_\pi)\Pi^{m} > 0$ and continues to restrict output. Comparing, at $A_a^1$, (i)~monopoly with no profit tax, (ii)~monopoly with $\tau_\pi \in (0,1)$, and (iii)~restored competition---each protecting substitutable workers at $w_S^{m,0}$---Figure~\ref{fig:mono_regime_comparison} confirms the ordering $\widetilde{w}_C^{(iii)} > \widetilde{w}_C^{(ii)} > \widetilde{w}_C^{(i)}$, with the gap to competition widening in $A_a^1$ and only partly closed by the partial profit tax.
So, after-tax wages for complementary workers are best under competition, then under partial taxation of the monopolist, and worst with no taxation of the monopolist.

\begin{figure}[t]
\centering
\includegraphics[width=0.52\textwidth]{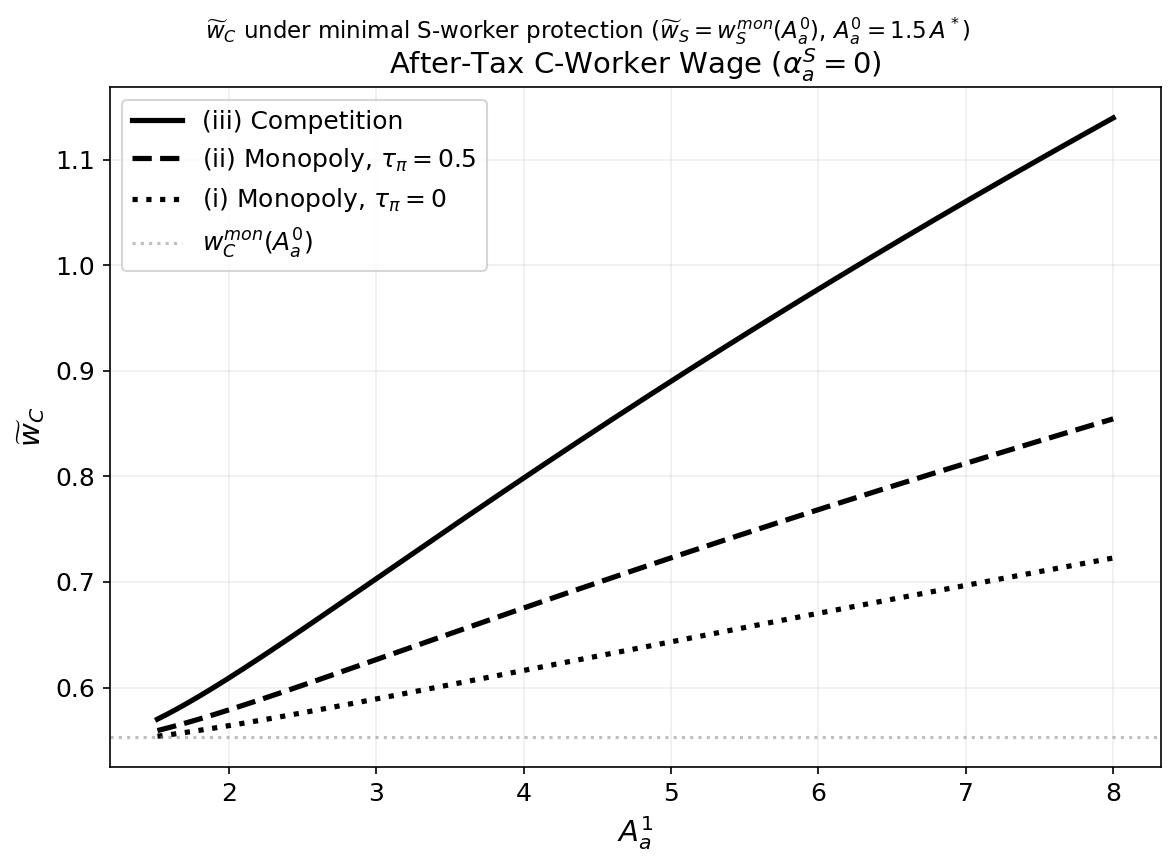}
\caption{After-Tax Complementary Wage under Minimal S-Worker Protection}
\label{fig:mono_regime_comparison}

\vspace{0.1cm}
\begin{minipage}{0.95\textwidth}
\footnotesize
\textit{Notes:} After-tax complementary wage $\widetilde{w}_C$ when the planner sets $\widetilde{w}_S = w_S^{m}(A_a^0)$: (iii) competition (solid), (ii) monopoly with $\tau_\pi = 0.5$ (dashed), (i) monopoly with $\tau_\pi = 0$ (dotted). The gray dotted line marks $w_C^{m}(A_a^0)$; benchmark $A_a^0 = 1.5\,A^*$. Parameters: $\alpha^S_{a} = 0$, $\alpha^S_{f} = \alpha^C_{f} = 0.5$, $\overline{L}^S = \overline{L}^C = 1$, $A_f = 1$.
\end{minipage}
\end{figure}

Finally, let us consider the fourth result mentioned above that concerns the alternative target of a fixed relative wage between the two types of labor.  

For the relative-wage target $\gamma$, the required tax $\tau_C$ (for $\tau_\pi=0$) is
\begin{equation*}
\tau_C = \frac{\overline{L}^S}{\gamma\overline{L}^C + 
\overline{L}^S}\left(1 - \frac{\gamma}{w_C^{m,1}/w_S^{m,1}}
\right),
\label{eq:tau_gamma_mono}
\end{equation*}
structurally identical to the competitive formula, with market structure entering only through the wage ratio. Since the onset threshold $A^*$ is common to both structures, lower AI adoption under monopoly compresses wage inequality, so the $\gamma$-tax is lower than under competition and falls further if monopoly profits are also taxed.

Overall, the analysis delivers three messages. First, optimal taxes depend not only on AI productivity but also on the market structure of AI production: the same productivity level can call for different instruments and rates under competition than under monopoly. Second, the instrument mix shifts along the transition path---a profit tax cannot stand alone near the onset, where rents are small relative to the deficit, and must be paired with a complementary-worker tax, while profit taxation alone suffices only at higher productivity. Third, moving toward competition raises output and the redistributable surplus, so that with redistribution in place \emph{all} workers are strictly better off than under monopoly: although the pre-tax wage $w_S^{m}$ is higher under monopoly, the after-tax wage attainable under competition is higher still, because the larger pie more than offsets the larger transfer needed to protect substitutable workers. The cost of monopoly therefore falls on the entire workforce. These messages bear directly on current proposals to tax AI rents, cap AI prices, or fund public wealth from AI profits: the right policy bundle is not fixed, but depends on where the economy stands in the transition and on how the AI sector is structured. Whether monopoly's potential advantages in spurring AI development outweigh the output loss documented here is left to future research.

\section{Discussion of the Model}
\label{SEC:ROBUSTNESS}

\subsection{Price Anchoring and AI-Sector Employment}
\label{sec:price_anchoring}

Our modeling of labor in AI production has essential implications for wages even if there is not a large absorption of labor by the AI sector.  
The vital aspect of labor being used in AI production is that it ties wages to the productivity of AI.  This has important and contrasting implications for substitutable and complementary labor, even if AI production is itself a small part of the total economy.  This is important to emphasize as the AI production sector may never absorb large fractions of the available labor forces, but may never the less determine their fates.  Let us examine this in more detail, given its importance.   

During the transition, final-good producers arbitrage between AI and substitutable labor, which has the important implication that $p_a=w_S$. This coupled with the fact that AI's price is its unit cost,
\[
p_a
=
\frac{w_S^{\alpha^S_a}w_C^{\alpha^C_a}}{A_a\Phi},
\qquad
\Phi=(\alpha^S_a)^{\alpha^S_a}(\alpha^C_a)^{\alpha^C_a},
\]
is what gives the conclusion that
\[
\frac{w_C}{w_S}=(A_a\Phi)^{1/\alpha^C_a},
\]
where recall that $\Phi = (\alpha^C_{a})^{\alpha^C_{a}} (\alpha^S_{a})^{\alpha^S_{a}}$.
Thus, the AI price anchors the relative wages, which ultimately boil down to AI's productivity and the relative roles of different sources of labor in the different sectors.  
This price channel holds regardless of how modest employment in the AI-sector is. 

This price-anchoring logic also explains our redistribution results. The relevant tax base is not the number of workers employed in AI production, but the income gains accruing to factors that determine the price of AI. As AI productivity rises, substitutable workers face wage losses because their wage is tied to competition with the AI price, which drops as it becomes more productive.  Other AI-producing factors gain because they become more valuable inputs into production. Aggregate gains exceed the losses of substitutable workers, so taxes on the incomes generated by AI production can finance Pareto-improving transfers.

The same logic applies beyond labor. If AI production uses scarce inputs with alternative uses---such as specialized capital, computing infrastructure, energy, land, data, or materials---then part of the gains from AI adoption accrues to the owners of those inputs. A tax on AI-producing factor income should therefore be interpreted broadly: it may fall on wages of AI-building workers, rents to specialized capital, or returns to other scarce inputs used in AI production. This is conceptually distinct from a tax on monopoly profit. As we showed in Section~\ref{SEC:MONOPOLY}, these tax bases need not move together: monopoly profits may be insufficient to compensate displaced workers near the onset of adoption, while factor-income taxation may be needed to supplement redistribution; later in the transition, as rents grow, profit taxation can become sufficient on its own.

\subsection{AI-in-AI}

We can easily extend the model to allow AI to be used in its own production. The AI sector now produces according to
\begin{equation}
	Y_a = A_a [L^S_{a} + X_{aa}]^{\alpha^S_{a}} (L^C_{a})^{\alpha^C_{a}},
\end{equation}
where $X_{aa}$ denotes AI output reinvested into the sector to substitute for substitutable labor, and the market clearing condition is $X_{aa} + X_{fa} = Y_a$. The threshold $A^*$ is unchanged.\footnote{At the onset, AI production is just beginning---$Y_a$ is zero---so there is no AI output available to use as an input in AI production. The recursive channel is therefore inactive at $A^*$, and the threshold is determined by the same condition as before.} 

Between $A^*$ and $A^{**}$ (the thresholds in the baseline model without AI used in AI production), aggregate outcomes---wages, prices, output, and the wage ratio---at any level of $A_a$ are identical to those in the baseline model regardless of how much AI is used within the AI sector itself. The recursive use of AI introduces multiplicity in the sectoral adoption intensities, but the equilibrium levels of wages, prices, and outputs are unique. The full derivation is provided in the Online Supplementary Appendix.

The critical divergence between the baseline model and the AI-in-AI specification arises beyond $A_a^{**}$ of the baseline model. In the baseline model, once all substitutable labor in the final goods sector is displaced, the Domar weight becomes fixed and further gains in $A_a$ raise output at a constant elasticity, while the wage ratio $w_C/w_S$ stabilizes. In the AI-in-AI specification, there is no finite $A_a^{**}$: even after the final goods sector has fully adopted AI, the AI sector continues to substitute its own output for labor. The wage ratio $w_C/w_S = (A_a\Phi)^{1/\alpha^C_a}$ expands indefinitely, and the Domar weight continues to grow. The self-improving nature of AI thus prevents the plateauing of AI-driven growth, but at the cost of persistently widening wage inequality.

\begin{figure}[t]
\centering
\includegraphics[width=1.0\textwidth]{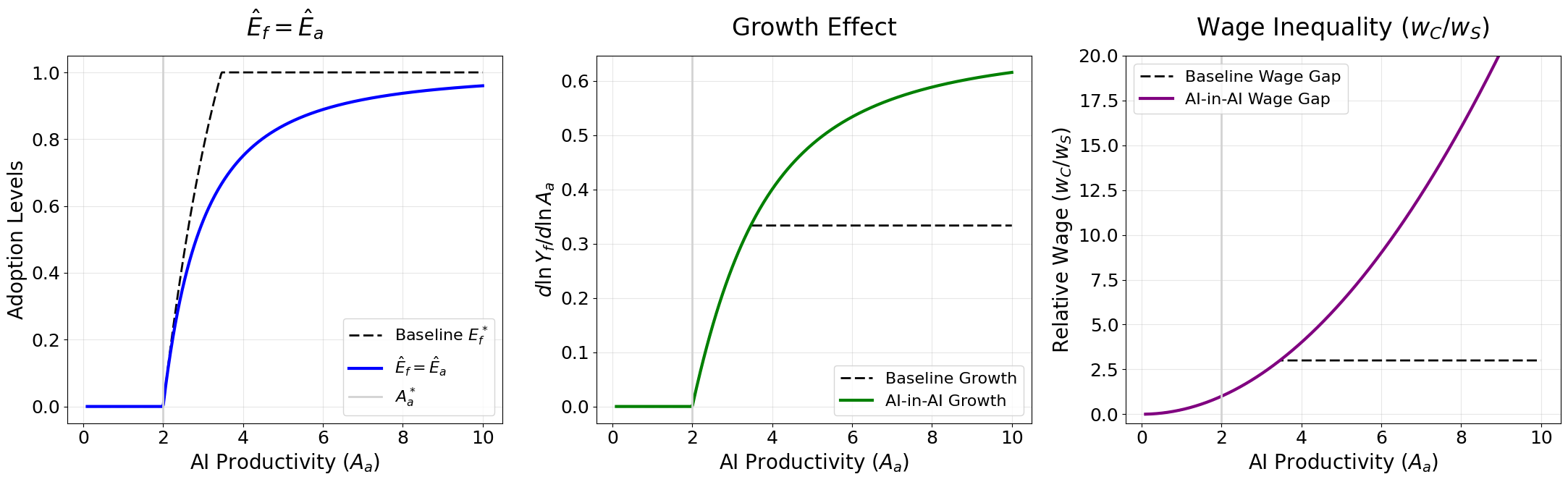}
\caption{AI-in-AI: Adoption, Growth, and Inequality}
\label{fig:ai_in_ai_sim}
\vspace{0.2cm}
\begin{minipage}{0.95\textwidth}
\footnotesize
\textit{Notes:} This figure illustrates equilibrium dynamics for a symmetric path where $\hat{E}_f = \hat{E}_a$. Parameters: $\overline{L}^S=\overline{L}^C=1$, $\alpha^S_{f}=\alpha^C_{f}=\alpha^F_{f}=1/3$, and $\alpha^S_{a}=\alpha^C_{a}=0.5$. The vertical gray line indicates the onset threshold $A_a^*=2$. Prior to the baseline full-displacement threshold ($A_a^{**} \approx 3.46$), both specifications yield identical outcomes. Beyond $A_a^{**}$, the AI-in-AI model exhibits continued growth and expanding inequality, while the baseline model stagnates.
\end{minipage}
\end{figure}

\subsection{Beyond Cobb-Douglas}
\label{sec:robustness_cd}

The invariance result---that the wage ratio $w_C/w_S$ depends only on AI-sector parameters and AI productivity---does not rely on the Cobb-Douglas specification. 

Under a constant-returns AI technology, the unit cost of AI is a function of input prices and productivity alone, $p_a = c(w_S, w_C; A_a)$. Combined with the equilibrium condition $p_a = w_S$ during the transition---which follows from the perfect substitutability of AI output and substitutable labor in the final goods sector, independent of the functional form of either technology---and the cost function giving $p_a = c = w_S$, this implicitly determines $w_C/w_S$ as a function of $A_a$ and AI-sector technology parameters only.

As a concrete illustration, replace the AI-sector production technology with the constant-elasticity-of-substitution (CES) production function
\begin{equation}
Y_a=A_a\Big[\alpha^S_{a}\,(L^S_{a})^{\rho}+\alpha^C_{a}\,(L^C_{a})^{\rho}\Big]^{1/\rho},\qquad \rho=\frac{\sigma-1}{\sigma},
\label{eq:ces_production}
\end{equation}
where $\sigma>0$ is the elasticity of substitution between substitutable and complementary labor in building AI ($\sigma\to1$ recovers Cobb-Douglas). 
Equilibrium conditions appearing in Appendix~\ref{app:ces} show that the wage ratio is
\begin{equation}
    \frac{w_C}{w_S} = \left(\frac{A_a^{1-\sigma} - (\alpha^S_a)^\sigma}{(\alpha^C_a)^\sigma}\right)^{\frac{1}{1-\sigma}}.
    \label{eq:wage_ratio_ces}
\end{equation}

The key properties of the Cobb-Douglas case survive. The ratio still depends only on AI-sector parameters and AI productivity, and it is still increasing in $A_a$ for every $\sigma$: a more productive AI sector raises the complementary wage relative to the substitutable wage. As $\sigma\to1$ the formula collapses to the Cobb-Douglas ratio $(A_a\Phi)^{1/\alpha^C_{a}}$. Since the redistribution results of Section~\ref{SEC:INCOME_TAX} and the monopoly extension use the AI technology only through this monotonicity, they carry over unchanged. What the elasticity changes is the span of productivities over which the transition runs. Replacing Cobb-Douglas with CES therefore leaves the direction of every effect unchanged.

Similarly, a CES structure in the final-good sector would not alter the qualitative implications. What matters is the substitutability of AI and substitutable labor and labor mobility within each class, not the functional form of the final good sector's production technology. \cite{kj} provide a formal treatment under general nested CES structures and study the labor-mobility dimension of a related problem, and the results are similar.

The Cobb-Douglas structure of final production, delivering the wage formulas \eqref{eq:wS_GDP}--\eqref{eq:wF_GDP}, was chosen for its tractability for the tax analysis in Section~\ref{SEC:INCOME_TAX}. The closed-form expressions for tax levels change with the functional form, but the qualitative results would hold under functional forms assigning a similar distinction for different labor types.

\section{Concluding Remarks}
\label{SEC:CONCLUSION}

We have developed a tractable model with multiple forms of labor and in which AI is produced by labor. We have used it to show how progress in AI productivity shifts labor allocation throughout the economy and changes relative wages and aggregate output.   A key feature---that AI is built by workers with productive uses elsewhere in the economy---generates equilibrium dynamics that are absent when automation is modeled simply as capital.  Labor reallocates endogenously across sectors as AI productivity rises, and the wage ratio between complementary and substitutable workers can be measured by AI-sector productivity, independent of relative labor supplies.

Three implications deserve emphasis. First, AI-driven growth and distributional outcomes do not move together. During the transition, aggregate output rises while substitutable workers experience absolute wage declines, and the gains from growth can be smaller than the wage losses borne by displaced workers. 
Second, the choice between absolute and relative welfare targets, as well as the pace of technology, matter fundamentally for Pareto-improving taxation. 
Third, market structure matters not only for the level of wages but also for the design of the redistribution instrument. Under monopoly, a profit tax is a strictly superior instrument in principle---it falls on pure rents and does not distort pre-tax income---but it cannot stand alone in the early transition: monopoly rents grow too slowly relative to the deficit substitutable workers face. Pareto-improving redistribution then requires combining a profit tax with a tax on workers producing AI, and the required combination shifts as AI productivity rises.

These results have implications for contemporary policy debates in which the question of how to share the gains from AI is becoming increasingly urgent. Proposals for public wealth funds, capital-based taxation, and broader participation in AI-driven growth all rely implicitly on assumptions about the structure of AI production and the available fiscal instruments. Our framework suggests that the answers depend systematically on the market structure of AI production---and that the same policy that achieves Pareto-improvement at one stage of the transition may be inadequate at another. Policy design therefore needs to be dynamic: tax regimes that work near the onset of AI adoption may need to evolve substantially as AI productivity grows and the composition of available rents shifts.

Several extensions deserve further investigation. 

The labor-produced AI assumption can be combined with frictions that we have abstracted from: costly retraining, skill acquisition, labor-market matching, and the dynamic incentives of monopolists facing entry threats. 

The empirical literature on AI and productivity has begun to document differential effects across worker categories that map into the role-based classification we use \citep{felten2021occupational,brynjolfsson2025canaries}.
Our three labor groups capture the essentials in plain terms: workers who are adversely affected by AI due to strong substitutability, workers who gain from AI adoption due to increased productivity, and workers who build AI and who gain the most from AI. 
Nonetheless, a natural next step is to combine our key new ingredient---that AI is itself produced by labor---with a model of AI substitution and complementarity at the task and skill level.  That would involve some loss of tractability in terms of policy analysis and broader insights, but would result in a model that could be fit more directly empirically.  
In particular, one thing that is missing from our model, and challenging to anticipate, is what new opportunities for labor will emerge from AI.   For instance, \cite{althoffReichardt2026} discuss how AI could make it possible for lower skilled workers equipped with AI to complete tasks that used to only be available for workers with higher skills.   In our model, that comes in the form of substitution (those high skilled workers would fall into our substitutable category, and those low skilled workers would fall into our complementary and final goods workers).  But it could also be that there are entirely new things that could be done with AI, resulting in new goods and new employment opportunities.   The rate of those could speed the transition, absorbing more workers.   

A broader lesson from our analysis is that the structure of AI production---who builds it, who is displaced by it, and how rents from it are captured---matters for the welfare consequences and depends on the pace of technological progress itself. The transition to widespread AI adoption is unlikely to be smooth or uniform across economies, and the policies that determine whether its gains are broadly shared need to reflect the production technologies of AI and the rest of the production economy. The framework we have developed provides a transparent benchmark for thinking about these questions, and we hope it stimulates further work that can refine, extend, and challenge the conclusions in the model here.

\printbibliography[title=References]

\newpage
\section{Appendix}

\label{appendix}

\subsection{Derivations for the Competitive Equilibrium}

A \textit{competitive equilibrium} is a set of prices $\{p_f, p_a, w_S, w_C, w_F\}$ and allocations and outputs $\{Y_f, Y_a, X_a, L^S_{f}, L^C_{f}, L^F_{f}, L^S_{a}, L^C_{a}\}$ such that:

\begin{itemize}
    \item[I.] The representative firm in the final good sector chooses inputs to maximize profit:
    \[
    \max_{L^F_{f}, L^S_{f}, L^C_{f}, X_a} \quad 
    A_{f} [L^S_{f} + X_a]^{\alpha^S_{f}} (L^C_{f})^{\alpha^C_{f}} (L^F_{f})^{\alpha^F_{f}}
    - \left( w_{F}L^F_{f} + w_{S}L^S_{f} + w_{C}L^C_{f} + p_a X_a \right).
    \]

    \item[II.] The representative firm in the AI sector chooses inputs to maximize profit:
    \[
    \max_{L^S_{a}, L^C_{a}} \quad 
    p_{a} \left[ A_{a} (L^S_{a})^{\alpha^S_{a}} (L^C_{a})^{\alpha^C_{a}} \right] 
    - \left( w_{S}L^S_{a} + w_{C}L^C_{a} \right).
    \]

    \item[III.] Representative households supply labor inelastically and maximize consumption $C^j$ subject to their budget constraint $C^j \leq w_j \overline L^j$ for $j \in \{S, C, F\}$.

    \item[IV.] Markets clear:
    \begin{align*}
        \text{final good:} & \quad Y_f = C^S + C^C + C^F. \\
        \text{AI output:} & \quad Y_a = X_a. \\
        \text{Labor Markets:} & \quad \overline L^S = L^S_{a} + L^S_{f}, \quad \overline L^C = L^C_{a} + L^C_{f}, \quad \overline L^F = L^F_{f}.
    \end{align*}
\end{itemize}

\subsubsection*{A. AI Sector}

The cost minimization problem for the representative AI producer is
\[
\underset{L^S_{a},L^C_{a}}{min}\ w_{S}L^S_{a}+w_{C}L^C_{a}
\quad\text{subject to}\quad
1=A_a(L^S_{a})^{\alpha^S_{a}}(L^C_{a})^{\alpha^C_{a}}.
\]

The Lagrangian function is:
\[
\mathcal{L=}w_{S}L^S_{a}+w_{C}L^C_{a}-\lambda_{a}\left(A_a(L^S_{a})^{\alpha^S_{a}}(L^C_{a})^{\alpha^C_{a}}-1\right).
\]

The first order conditions are:
\begin{itemize}
\item {\footnotesize{}$\frac{\partial\mathcal{L}}{\partial L^C_{a}}=w_{C}-\frac{\lambda_{a}^{*}\alpha^C_{a}A_a(L^{S*}_{a})^{\alpha^S_{a}}(L^{C*}_{a})^{\alpha^C_{a}}}{L^{C*}_{a}}=0$}{\footnotesize \par}
{\footnotesize{}$L^{C*}_{a}=\frac{\lambda_{a}^{*}\alpha^C_{a}A_a(L^{S*}_{a})^{\alpha^S_{a}}(L^{C*}_{a})^{\alpha^C_{a}}}{w_{C}}$}{\footnotesize \par}

\item {\footnotesize{}$\frac{\partial\mathcal{L}}{\partial L^S_{a}}=w_{S}-\frac{\lambda_{a}^{*}\alpha^S_{a}A_a(L^{S*}_{a})^{\alpha^S_{a}}(L^{C*}_{a})^{\alpha^C_{a}}}{L^{S*}_{a}}=0$}{\footnotesize \par}
{\footnotesize{}$L^{S*}_{a}=\frac{\lambda_{a}^{*}\alpha^S_{a}A_a(L^{S*}_{a})^{\alpha^S_{a}}(L^{C*}_{a})^{\alpha^C_{a}}}{w_{S}}$}{\footnotesize \par}

\item {\footnotesize{}$\frac{\partial\mathcal{L}}{\partial\lambda_{a}}=1-A_a(L^{S*}_{a})^{\alpha^S_{a}}(L^{C*}_{a})^{\alpha^C_{a}}=0$}{\footnotesize \par}
\end{itemize}

By plugging $L^{C*}_{a}$ and $L^{S*}_{a}$ into the last FOC above, we find:
{\footnotesize{}
\[
\frac{1}{\lambda_{a}^{*}}
=
A_a\left(\frac{\alpha^C_{a}}{w_{C}}\right)^{\alpha^C_{a}}\left(\frac{\alpha^S_{a}}{w_{S}}\right)^{\alpha^S_{a}}.
\]
}{\footnotesize \par}

The zero profit condition implies that
$\lambda_{a}^{*}=w_{C}L^{C*}_{a}+w_{S}L^{S*}_{a}=C_a(w_{S},w_{C})=p_{a}$.
Then, the price and factor demands can be written as follows:

\begin{equation}
p_{a}
=
\frac{1}{A_{a}(\alpha^C_{a})^{\alpha^C_{a}}(\alpha^S_{a})^{\alpha^S_{a}}}
\left(w_{S}\right)^{\alpha^S_{a}}\left(w_{C}\right)^{\alpha^C_{a}}.
\label{p_a-1-1}
\end{equation}

\[
L^S_{a}=\frac{\alpha^S_{a}p_{a}Y_{a}}{w_{S}}
\quad\text{and}\quad
L^C_{a}=\frac{\alpha^C_{a}p_{a}Y_{a}}{w_{C}}.
\]

\bigskip

\subsubsection*{B. Final-good sector}

Given the market clearing condition $Y_{a}=X_a$, we use the notation $Y_{a}$ instead of $X_a$.

{\footnotesize{}
\[
\begin{array}{c}
\underset{L^F_{f},L^S_{f},L^C_{f},Y_{a}}{min}
\ w_{F}L^F_{f}+w_{S}L^S_{f}+w_{C}L^C_{f}+p_{a}Y_{a}
\ \text{ subject to}\\
1=A_{f}(L^C_{f})^{\alpha^C_{f}}(L^F_{f})^{\alpha^F_{f}}[L^S_{f}+Y_{a}]^{\alpha^S_{f}}.
\end{array}
\]
}{\footnotesize \par}

The Lagrangian function is:
{\footnotesize{}
\[
\begin{array}{c}
\mathcal{L=}
w_{F}L^F_{f}+w_{S}L^S_{f}+w_{C}L^C_{f}+p_{a}Y_{a}\\
-\lambda_f\left(A_{f}(L^C_{f})^{\alpha^C_{f}}(L^F_{f})^{\alpha^F_{f}}[L^S_{f}+Y_{a}]^{\alpha^S_{f}}-1\right).
\end{array}
\]
}{\footnotesize \par}

The FOCs imply:
\[
A_{f}(L^{C*}_{f})^{\alpha^C_{f}}(L^{F*}_{f})^{\alpha^F_{f}}[L^{S*}_{f}+Y_{a}^{*}]^{\alpha^S_{f}}=1.
\]

By plugging the equation above into the other FOCs, we get:
\begin{itemize}
\item {\footnotesize{}$\frac{\partial\mathcal{L}}{\partial L^C_{f}}=w_{C}-\lambda_{f}^{*}\frac{\alpha^C_{f}}{L^{C*}_{f}}=0$}{\footnotesize \par}
{\footnotesize{}$L^{C*}_{f}=\lambda_{f}^{*}\frac{\alpha^C_{f}}{w_{C}}$}{\footnotesize \par}

\item {\footnotesize{}$\frac{\partial\mathcal{L}}{\partial L^F_{f}}=w_{F}-\lambda_{f}^{*}\frac{\alpha^F_{f}}{L^{F*}_{f}}=0$}{\footnotesize \par}
{\footnotesize{}$L^{F*}_{f}=\lambda_{f}^{*}\frac{\alpha^F_{f}}{w_{F}}$}{\footnotesize \par}

\item {\footnotesize{}$\frac{\partial\mathcal{L}}{\partial L^S_{f}}=\left(w_{S}-\lambda_{f}^{*}\frac{\alpha^S_{f}}{L^{S*}_{f}+Y_{a}^{*}}\right)\geq0$
and
$L^{S*}_{f}\left(w_{S}-\lambda_{f}^{*}\frac{\alpha^S_{f}}{L^{S*}_{f}+Y_{a}^{*}}\right)=0$}{\footnotesize \par}

{\footnotesize{}$\frac{\partial\mathcal{L}}{\partial Y_{a}}=\left(p_{a}-\lambda_{f}^{*}\frac{\alpha^S_{f}}{L^{S*}_{f}+Y_{a}^{*}}\right)\geq0$
and
$Y_{a}^{*}\left(p_{a}-\lambda_{f}^{*}\frac{\alpha^S_{f}}{L^{S*}_{f}+Y_{a}^{*}}\right)=0$}{\footnotesize \par}

Both $L^{S*}_{f}$ and $Y_{a}^{*}$ can not be zero, otherwise $Y_{f}=0$. Then:
\begin{itemize}
\item {\footnotesize{}Case 1. $w_{S}=\lambda_{f}^{*}\frac{\alpha^S_{f}}{L^{S*}_{f}+Y_{a}^{*}}$,
$L^{S*}_{f}=\lambda_{f}^{*}\frac{\alpha^S_{f}}{w_{S}}$; and $p_{a}>\lambda_{f}^{*}\frac{\alpha^S_{f}}{L^{S*}_{f}+Y_{a}^{*}}$,
$Y_{a}^{*}=0$ . In this case, $\frac{w_{S}}{p_{a}}<1$.}{\footnotesize \par}

\item {\footnotesize{}Case 2. $w_{S}>\lambda_{f}^{*}\frac{\alpha^S_{f}}{L^{S*}_{f}+Y_{a}^{*}}$,
$L^S_{f}=0$; and $p_{a}=\lambda_{f}^{*}\frac{\alpha^S_{f}}{L^{S*}_{f}+Y_{a}^{*}}$,
$Y_{a}^{*}=\lambda_{f}^{*}\frac{\alpha^S_{f}}{p_{a}}$. In this case, $\frac{w_{S}}{p_{a}}>1$.}{\footnotesize \par}

\item {\footnotesize{}Case 3. $w_{S}=\lambda_{f}^{*}\frac{\alpha^S_{f}}{L^{S*}_{f}+Y_{a}^{*}}$,
$L^{S*}_{f}=\lambda_{f}^{*}\frac{(1-e_{f})\alpha^S_{f}}{w_{S}}$;
$p_{a}=\lambda_{f}^{*}\frac{\alpha^S_{f}}{L^{S*}_{f}+Y_{a}^{*}}$,
$Y_{a}^{*}=\lambda_{f}^{*}\frac{e_{f}\alpha^S_{f}}{p_{a}}$ .
In this case, $\frac{w_{S}}{p_{a}}=1$.}{\footnotesize \par}
\end{itemize}
\end{itemize}

{\footnotesize{}
\[
\{L^{S*}_{f},Y_{a}^{*}\}=:\begin{cases}
\{\lambda_{f}^{*}\frac{\alpha^S_{f}}{w_{S}},0\} & \text{if }\frac{w_{S}}{p_{a}}<1\\
\{\lambda_{f}^{*}\frac{(1-e_{f})\alpha^S_{f}}{w_{S}},\lambda_{f}^{*}\frac{e_{f}\alpha^S_{f}}{p_{a}} \} & \text{if }\frac{w_{S}}{p_{a}}=1\\
\{0,\lambda_{f}^{*}\frac{\alpha^S_{f}}{p_{a}}\} & \text{if }\frac{w_{S}}{p_{a}}>1
\end{cases}
\]
}{\footnotesize \par}

The FOCs above together imply:
{\footnotesize{}
\[
\lambda_{f}^{*}=\begin{cases}
\frac{1}{A_{f}(\frac{\alpha^C_{f}}{w_{C}})^{\alpha^C_{f}}(\frac{\alpha^F_{f}}{w_{F}})^{\alpha^F_{f}}\left[\frac{\alpha^S_{f}}{w_{S}}\right]^{\alpha^S_{f}}} & \text{if }\frac{w_{S}}{p_{a}}<1\\[0.3em]
\frac{1}{A_{f}(\frac{\alpha^C_{f}}{w_{C}})^{\alpha^C_{f}}(\frac{\alpha^F_{f}}{w_{F}})^{\alpha^F_{f}}\left[\frac{\alpha^S_{f}}{w_{S}}\right]^{\alpha^S_{f}}} & \text{if }\frac{w_{S}}{p_{a}}=1\\[0.3em]
\frac{1}{A_{f}(\frac{\alpha^C_{f}}{w_{C}})^{\alpha^C_{f}}(\frac{\alpha^F_{f}}{w_{F}})^{\alpha^F_{f}}\left[\frac{\alpha^S_{f}}{p_{a}}\right]^{\alpha^S_{f}}} & \text{if }\frac{w_{S}}{p_{a}}>1
\end{cases}
\]
}{\footnotesize \par}

Under the zero profit conditions and the normalization of $p_{f}=1$, it follows that:
{\footnotesize{}
\[
\begin{cases}
A_{f}(\frac{\alpha^C_{f}}{w_{C}})^{\alpha^C_{f}}(\frac{\alpha^F_{f}}{w_{F}})^{\alpha^F_{f}}\left[\frac{\alpha^S_{f}}{w_{S}}\right]^{\alpha^S_{f}}=1 & \text{if }\frac{w_{S}}{p_{a}}<1\\
A_{f}(\frac{\alpha^C_{f}}{w_{C}})^{\alpha^C_{f}}(\frac{\alpha^F_{f}}{w_{F}})^{\alpha^F_{f}}\left[\frac{\alpha^S_{f}}{w_{S}}\right]^{\alpha^S_{f}}=1 & \text{if }\frac{w_{S}}{p_{a}}=1\\
A_{f}(\frac{\alpha^C_{f}}{w_{C}})^{\alpha^C_{f}}(\frac{\alpha^F_{f}}{w_{F}})^{\alpha^F_{f}}\left[\frac{\alpha^S_{f}}{p_a}\right]^{\alpha^S_{f}}=1 & \text{if }\frac{w_{S}}{p_{a}}>1
\end{cases}
\]
}{\footnotesize \par}

Lastly, by plugging $p_{a}$ into the equation above, it follows that:
{\footnotesize{}
\begin{equation}
\begin{cases}
A_{f}(\frac{\alpha^C_{f}}{w_{C}})^{\alpha^C_{f}}(\frac{\alpha^F_{f}}{w_{F}})^{\alpha^F_{f}}\left[\frac{\alpha^S_{f}}{w_{S}}\right]^{\alpha^S_{f}}=1
& \text{if }
A_{a}(\alpha^C_{a})^{\alpha^C_{a}}(\alpha^S_{a})^{\alpha^S_{a}}w_{S}^{1-\alpha^S_{a}}w_{C}^{-\alpha^C_{a}}<1\\[0.4em]
A_{f}(\frac{\alpha^C_{f}}{w_{C}})^{\alpha^C_{f}}(\frac{\alpha^F_{f}}{w_{F}})^{\alpha^F_{f}}\left[\frac{\alpha^S_{f}}{w_{S}}\right]^{\alpha^S_{f}}=1
& \text{if }
A_{a}(\alpha^C_{a})^{\alpha^C_{a}}(\alpha^S_{a})^{\alpha^S_{a}}w_{S}^{1-\alpha^S_{a}}w_{C}^{-\alpha^C_{a}}=1\\[0.4em]
A_{f}(\frac{\alpha^C_{f}}{w_{C}})^{\alpha^C_{f}}(\frac{\alpha^F_{f}}{w_{F}})^{\alpha^F_{f}}
\left[\frac{\alpha^S_{f}A_{a}(\alpha^C_{a})^{\alpha^C_{a}}(\alpha^S_{a})^{\alpha^S_{a}}}{\left(w_{S}\right)^{\alpha^S_{a}}\left(w_{C}\right)^{\alpha^C_{a}}}\right]^{\alpha^S_{f}}=1
& \text{if }
A_{a}(\alpha^C_{a})^{\alpha^C_{a}}(\alpha^S_{a})^{\alpha^S_{a}}w_{S}^{1-\alpha^S_{a}}w_{C}^{-\alpha^C_{a}}>1
\end{cases}
\label{w_S_and_w_C_and_exog_param}
\end{equation}
}{\footnotesize \par}

\bigskip

Then, the conditional factor demands are:
{\footnotesize{}
\begin{equation}
L^C_{f}=:\begin{cases}
\frac{\alpha^C_{f}Y_{f}}{w_{C}A_{f}(\frac{\alpha^C_{f}}{w_{C}})^{\alpha^C_{f}}(\frac{\alpha^F_{f}}{w_{F}})^{\alpha^F_{f}}\left[\frac{\alpha^S_{f}}{w_{S}}\right]^{\alpha^S_{f}}} & \text{if }\frac{w_{S}}{p_{a}}<1\\
\frac{\alpha^C_{f}Y_{f}}{w_{C}A_{f}(\frac{\alpha^C_{f}}{w_{C}})^{\alpha^C_{f}}(\frac{\alpha^F_{f}}{w_{F}})^{\alpha^F_{f}}\left[\frac{\alpha^S_{f}}{w_{S}}\right]^{\alpha^S_{f}}} & \text{if }\frac{w_{S}}{p_{a}}=1\\
\frac{\alpha^C_{f}Y_{f}}{w_{C}A_{f}(\frac{\alpha^C_{f}}{w_{C}})^{\alpha^C_{f}}(\frac{\alpha^F_{f}}{w_{F}})^{\alpha^F_{f}}\left[\frac{\alpha^S_{f}}{p_{a}}\right]^{\alpha^S_{f}}} & \text{if }\frac{w_{S}}{p_{a}}>1
\end{cases}
\label{Lc_f}
\end{equation}

\begin{equation}
L^F_{f}=:\begin{cases}
\frac{\alpha^F_{f}Y_{f}}{w_{F}A_{f}(\frac{\alpha^C_{f}}{w_{C}})^{\alpha^C_{f}}(\frac{\alpha^F_{f}}{w_{F}})^{\alpha^F_{f}}\left[\frac{\alpha^S_{f}}{w_{S}}\right]^{\alpha^S_{f}}} & \text{if }\frac{w_{S}}{p_{a}}<1\\
\frac{\alpha^F_{f}Y_{f}}{w_{F}A_{f}(\frac{\alpha^C_{f}}{w_{C}})^{\alpha^C_{f}}(\frac{\alpha^F_{f}}{w_{F}})^{\alpha^F_{f}}\left[\frac{\alpha^S_{f}}{w_{S}}\right]^{\alpha^S_{f}}} & \text{if }\frac{w_{S}}{p_{a}}=1\\
\frac{\alpha^F_{f}Y_{f}}{w_{F}A_{f}(\frac{\alpha^C_{f}}{w_{C}})^{\alpha^C_{f}}(\frac{\alpha^F_{f}}{w_{F}})^{\alpha^F_{f}}\left[\frac{\alpha^S_{f}}{p_{a}}\right]^{\alpha^S_{f}}} & \text{if }\frac{w_{S}}{p_{a}}>1
\end{cases}
\label{Lf_f}
\end{equation}

\begin{equation}
\{L^S_{f},Y_a\}=:\begin{cases}
\{\frac{\alpha^S_{f}Y_{f}}{w_{S}A_{f}(\frac{\alpha^C_{f}}{w_{C}})^{\alpha^C_{f}}(\frac{\alpha^F_{f}}{w_{F}})^{\alpha^F_{f}}\left[\frac{\alpha^S_{f}}{w_{S}}\right]^{\alpha^S_{f}}},0\}
& \text{if }\frac{w_{S}}{p_{a}}<1\\[0.4em]
\{\frac{(1-e_{f})\alpha^S_{f}Y_{f}}{w_{S}A_{f}(\frac{\alpha^C_{f}}{w_{C}})^{\alpha^C_{f}}(\frac{\alpha^F_{f}}{w_{F}})^{\alpha^F_{f}}\left[\frac{\alpha^S_{f}}{w_{S}}\right]^{\alpha^S_{f}}},\frac{e_{f}\alpha^S_{f}Y_{f}}{p_{a}A_{f}(\frac{\alpha^C_{f}}{w_{C}})^{\alpha^C_{f}}(\frac{\alpha^F_{f}}{w_{F}})^{\alpha^F_{f}}\left[\frac{\alpha^S_{f}}{w_{S}}\right]^{\alpha^S_{f}}}\}
& \text{if }\frac{w_{S}}{p_{a}}=1\\[0.4em]
\{0,\frac{\alpha^S_{f}Y_{f}}{p_{a}A_{f}(\frac{\alpha^C_{f}}{w_{C}})^{\alpha^C_{f}}(\frac{\alpha^F_{f}}{w_{F}})^{\alpha^F_{f}}\left[\frac{\alpha^S_{f}}{p_{a}}\right]^{\alpha^S_{f}}}\}
& \text{if }\frac{w_{S}}{p_{a}}>1
\end{cases}
\label{Ls_f_Ya}
\end{equation}

}{\footnotesize \par}

\vspace{0.5cm}

Combining the budget constraint and FOCs of the utility maximization leads to:

\[
p_fY_{f}=(\alpha^C_{f}+\alpha^S_{f}+\alpha^F_{f})p_fY_{f}=w_SL^S+w_CL^C+w_F L^F=C^S+C^C+C^F.
\]

By combining the market clearing conditions and factor demands, we get:
\begin{equation}
L^S=:\begin{cases}
\frac{p_{f}Y_{f}\alpha^S_{f}}{w_{S}} & \text{if }\frac{w_{S}}{p_{a}}<1\\
\frac{p_{a}Y_{a}\alpha^S_{a}+p_{f}Y_{f}(1-e_{f})\alpha^S_{f}}{w_{S}} & \text{if }\frac{w_{S}}{p_{a}}=1\\
\frac{p_{a}Y_{a}\alpha^S_{a}}{w_{S}} & \text{if }\frac{w_{S}}{p_{a}}>1
\end{cases}
\label{LS_market}
\end{equation}

\begin{equation}
L^C=:\begin{cases}
\frac{p_{f}Y_{f}\alpha^C_{f}}{w_{C}} & \text{if }\frac{w_{S}}{p_{a}}<1\\
\frac{p_{a}Y_{a}\alpha^C_{a}+p_{f}Y_{f}\alpha^C_{f}}{w_{C}} & \text{if }\frac{w_{S}}{p_{a}}=1\\
\frac{p_{a}Y_{a}\alpha^C_{a}+p_{f}Y_{f}\alpha^C_{f}}{w_{C}} & \text{if }\frac{w_{S}}{p_{a}}>1
\end{cases}
\label{LC_market}
\end{equation}

\begin{equation}
L^F
=:
\frac{p_f Y_f \alpha^F_{f}}{w_F}
\quad\Longrightarrow\quad
w_F=\alpha^F_{f}\frac{p_f Y_f}{L^F}=\alpha^F_{f}\frac{Y_f}{L^F}.
\label{wF_simple}
\end{equation}

Market clearing for goods, factor demands, and $p_{f}=1$ together imply:
\begin{equation}
p_{a}Y_a=e_{f}\alpha^S_{f}Y_{f}.
\label{paYa}
\end{equation}
where $e_{f}=1$ for $\frac{w_{S}}{p_{a}}>1$, $e_{f}=0$ for $\frac{w_{S}}{p_{a}}<1$, and $0\leq e_{f}\leq1$ for $\frac{w_{S}}{p_{a}}=1$.

Then, we can rewrite the market clearing for labor as follows:
\begin{equation}
L^S=:\begin{cases}
\frac{\alpha^S_{f}\left(w_{S}L^S+w_{C}L^C+w_F L^F\right)}{w_{S}} & \text{if }\frac{w_{S}}{p_{a}}<1\\
\frac{\left(e_{f}\alpha^S_{f}\alpha^S_{a}+(1-e_{f})\alpha^S_{f}\right)\left(w_{S}L^S+w_{C}L^C+w_F L^F\right)}{w_{S}} & \text{if }\frac{w_{S}}{p_{a}}=1\\
\frac{\left(\alpha^S_{f}\alpha^S_{a}\right)\left(w_{S}L^S+w_{C}L^C+w_F L^F\right)}{w_{S}} & \text{if }\frac{w_{S}}{p_{a}}>1
\end{cases}
\label{LS_rewrite}
\end{equation}

\begin{equation}
L^C=:\begin{cases}
\frac{\alpha^C_{f}\left(w_{S}L^S+w_{C}L^C+w_F L^F\right)}{w_{C}} & \text{if }\frac{w_{S}}{p_{a}}<1\\
\frac{\left(e_{f}\alpha^S_{f}\alpha^C_{a}+\alpha^C_{f}\right)\left(w_{S}L^S+w_{C}L^C+w_F L^F\right)}{w_{C}} & \text{if }\frac{w_{S}}{p_{a}}=1\\
\frac{\left(\alpha^S_{f}\alpha^C_{a}+\alpha^C_{f}\right)\left(w_{S}L^S+w_{C}L^C+w_F L^F\right)}{w_{C}} & \text{if }\frac{w_{S}}{p_{a}}>1
\end{cases}
\label{LC_rewrite}
\end{equation}

\begin{equation}
\frac{w_{C}}{w_{S}}=:\begin{cases}
\frac{L^S}{L^C}\frac{\alpha^C_{f}}{\alpha^S_{f}} & \text{if }\frac{w_{S}}{p_{a}}<1\\[0.3em]
\frac{L^S}{L^C}\frac{\left(e_{f}\alpha^S_{f}\alpha^C_{a}+\alpha^C_{f}\right)}{\left(e_{f}\alpha^S_{f}\alpha^S_{a}+(1-e_{f})\alpha^S_{f}\right)} & \text{if }\frac{w_{S}}{p_{a}}=1\\[0.3em]
\frac{L^S}{L^C}\frac{\alpha^S_{f}\alpha^C_{a}+\alpha^C_{f}}{\alpha^S_{f}\alpha^S_{a}} & \text{if }\frac{w_{S}}{p_{a}}>1
\end{cases}
\label{wC_over_wS_piecewise}
\end{equation}

There exist two threshold levels of productivity of the AI sector $A_a^{*}$ and $A_a^{**}$
such that there is no AI usage in final-good sector if $A_{a} \leq A_a^{*}$; the level of AI's share gradually increases in between $A_a^{*}$ and $A_a^{**}$; and AI replaces all substitutable labor employed in the final-good sector if $A_{a} \geq A_a^{**}$.
By using the LHS and RHS equalities at threshold levels in Equation \eqref{w_S_and_w_C_and_exog_param}, the threshold levels are as follows:
\begin{equation}
\begin{cases}
A_a^{*}=\frac{1}{(\alpha^C_{a})^{\alpha^C_{a}}(\alpha^S_{a})^{\alpha^S_{a}}}\left(\frac{L^S}{L^C}\frac{\alpha^C_{f}}{\alpha^S_{f}}\right)^{\alpha^C_{a}}\\[0.4em]
A_a^{**}=\frac{1}{(\alpha^C_{a})^{\alpha^C_{a}}(\alpha^S_{a})^{\alpha^S_{a}}}\left(\frac{L^S}{L^C}\frac{\alpha^S_{f}\alpha^C_{a}+\alpha^C_{f}}{\alpha^S_{f}\alpha^S_{a}}\right)^{\alpha^C_{a}}
\end{cases}
\label{Aa_star}
\end{equation}

\bigskip

Next, we derive the fraction $e_{f}$ from the interior condition $\frac{w_S}{p_a}=1$ (i.e., the substitution phase), together with \eqref{p_a-1-1}:
\[
\frac{p_a}{w_S}
=
\frac{1}{A_a(\alpha^C_{a})^{\alpha^C_{a}}(\alpha^S_{a})^{\alpha^S_{a}}}\left(\frac{w_C}{w_S}\right)^{\alpha^C_{a}}
=1
\quad\Longrightarrow\quad
\frac{w_C}{w_S}=\left(A_a(\alpha^C_{a})^{\alpha^C_{a}}(\alpha^S_{a})^{\alpha^S_{a}}\right)^{\frac{1}{\alpha^C_{a}}}.
\]

Thus, during the substitution phase ($A_a^{*} < A_{a} < A_a^{**}$), we have:
\begin{equation}
\frac{w_C}{w_S}
=
\left(A_a(\alpha^C_{a})^{\alpha^C_{a}}(\alpha^S_{a})^{\alpha^S_{a}}\right)^{\frac{1}{\alpha^C_{a}}}.
\label{wC_over_wS_subphase}
\end{equation}

Combining \eqref{wC_over_wS_subphase} with the middle line of \eqref{wC_over_wS_piecewise}, and solving for $e_f^{\ast}$, yields:
\begin{equation}
e_{f}=:\begin{cases}
0 & \text{if }A_{a} \leq A_a^{*}\\[0.6em]
\displaystyle
\frac{
\frac{L^C}{L^S}\left(A_{a}(\alpha^C_{a})^{\alpha^C_{a}}(\alpha^S_{a})^{\alpha^S_{a}}\right)^{\frac{1}{\alpha^C_{a}}}\alpha^S_{f}-\alpha^C_{f}
}{
\alpha^C_{a}\alpha^S_{f}\left(1+\frac{L^C}{L^S}\left(A_{a}(\alpha^C_{a})^{\alpha^C_{a}}(\alpha^S_{a})^{\alpha^S_{a}}\right)^{\frac{1}{\alpha^C_{a}}}\right)
}
& \text{if }A_a^{*}\leq A_{a}\leq A_a^{**}\\[0.8em]
1 & \text{if }A_{a} \geq A_a^{**}
\end{cases}
\label{ea}
\end{equation}

Then, during the substitution phase ($A_a^{*} < A_{a} < A_a^{**}$), the expenditure share of AI in value added is:
\[
e_{f}\alpha^S_{f}=\frac{p_aY_a}{p_fY_f}
=
\frac{
\alpha^S_{f}\frac{L^C}{L^S}\left(A_{a}(\alpha^C_{a})^{\alpha^C_{a}}(\alpha^S_{a})^{\alpha^S_{a}}\right)^{\frac{1}{\alpha^C_{a}}}-\alpha^C_{f}
}{
\alpha^C_{a}\left(1+\frac{L^C}{L^S}\left(A_{a}(\alpha^C_{a})^{\alpha^C_{a}}(\alpha^S_{a})^{\alpha^S_{a}}\right)^{\frac{1}{\alpha^C_{a}}}\right)
}.
\]

\bigskip
\subsection{Proof of Proposition~\ref{prop:equilibrium_wages}}

\noindent Write $E_a= \dfrac{p_aY_a}{p_fY_f}=e_f\,\alpha^S_{f}$ for the AI expenditure share, where the second equality is \eqref{paYa}. The three wage formulas follow by combining the labor-market-clearing conditions with \eqref{paYa}.

\medskip
\noindent\emph{Substitutable labor.} From the interior case ($w_S/p_a=1$) of \eqref{LS_market}, $w_S L^S=p_aY_a\,\alpha^S_{a}+p_fY_f(1-e_f)\alpha^S_{f}$. Substituting $p_aY_a=e_f\alpha^S_{f}Y_f$ from \eqref{paYa} and $p_f=1$, and using $1-\alpha^S_{a}=\alpha^C_{a}$,
\[
w_S L^S=\alpha^S_{f}Y_f\big[e_f\alpha^S_{a}+(1-e_f)\big]=\alpha^S_{f}Y_f\big[1-e_f(1-\alpha^S_{a})\big]=\alpha^S_{f}Y_f-e_f\alpha^S_{f}\alpha^C_{a}Y_f=Y_f\big(\alpha^S_{f}-E_a\alpha^C_{a}\big),
\]
so that
\[
w_S=\frac{Y_f\left(\alpha^S_{f}-E_a\alpha^C_{a}\right)}{L^S}.
\]

\medskip
\noindent\emph{Complementary labor.} From the interior case of \eqref{LC_market}, $w_C L^C=p_aY_a\,\alpha^C_{a}+p_fY_f\,\alpha^C_{f}=e_f\alpha^S_{f}\alpha^C_{a}Y_f+\alpha^C_{f}Y_f=Y_f\big(\alpha^C_{f}+E_a\alpha^C_{a}\big)$, so that
\[
w_C=\frac{Y_f\left(\alpha^C_{f}+E_a\alpha^C_{a}\right)}{L^C}.
\]

\medskip
\noindent\emph{Final-good-specific labor.} Equation~\eqref{wF_simple} gives directly
\[
w_F=\frac{\alpha^F_{f}Y_f}{L^F}.
\]

\medskip
\noindent Finally, $E_a=p_aY_a/(p_fY_f)=e_f\alpha^S_{f}$ equals $0$ in the pre-AI phase ($e_f=0$), $\alpha^S_{f}$ in the post-AI phase ($e_f=1$), and the interior value in \eqref{ea} during the transition. \eproof

\bigskip
\subsection{Derivation of Absolute Wages in the Competitive Economy}

Under perfect competition and constant returns to scale, the equilibrium is anchored by the following unit-cost identity:
\begin{equation}
    1 = \frac{1}{\Theta} (\min\{p_a,ws\})^{\alpha^S_{f}} (w_C)^{\alpha^C_{f}} (w_F)^{\alpha^F_{f}},
    \label{eq:unit_cost_anchor}
\end{equation}
where $p_a$ is the effective price of the substitutable input bundle, and we define the aggregate productivity constant:
\[ \Theta = A_f (\alpha^S_{f})^{\alpha^S_{f}} (\alpha^C_{f})^{\alpha^C_{f}} (\alpha^F_{f})^{\alpha^F_{f}}. \]

The effective price $\min\{p_a,ws\}$ equals $w_S$ before AI is adopted, $w_S=p_a$ during the transition, and $p_a$ once AI is the cheaper input during post-transition in a finite transition case.

\paragraph{1. Pre-AI Regime ($A_a \leq A_a^*$)}
In this regime, $Y_a=0$. All labor types are employed in the final-good sector ($L^S_{f}=L^S, L^C_{f}=L^C, L^F_{f}=L^F$). Using the market clearing conditions and \eqref{wF_simple}, wages equal marginal products:
\begin{align}
    w_S &= \alpha^S_{f} A_f (L^S)^{\alpha^S_{f}-1} (L^C)^{\alpha^C_{f}} (L^F)^{\alpha^F_{f}}, \label{eq:wS_pre_level} \\
    w_C &= \alpha^C_{f} A_f (L^S)^{\alpha^S_{f}} (L^C)^{\alpha^C_{f}-1} (L^F)^{\alpha^F_{f}}, 
    \label{eq:wC_pre_level} \\
    w_F &= \alpha^F_{f} A_f (L^S)^{\alpha^S_{f}} (L^C)^{\alpha^C_{f}} (L^F)^{\alpha^F_{f}-1}.
\end{align}

\paragraph{2. Transition Regime ($A_a^* < A_a < A_a^{**}$)}
During the transition, we employ $p_a = w_S$ and the relative wage ratio $\Omega(A_a) = w_C/w_S$ from \eqref{wC_over_wS_subphase}. To pin down the levels we also need the ratio of substitutable to final-specific wages. This ratio is determined by the equilibrium of Section~\ref{appendix}, not by the levels we are about to solve for: from the interior case of \eqref{LS_market} together with \eqref{paYa} (which give $w_S L^S=\alpha^S_{f}Y_f[1-e_f^{\ast}(1-\alpha^S_{a})]$) and from \eqref{wF_simple} ($w_F=\alpha^F_{f}Y_f/L^F$),
\[
\frac{w_S}{w_F}= K_S = \frac{\alpha^S_{f}[1 - e_f^*(1-\alpha^S_{a})]L^F}{\alpha^F_{f}L^S},
\qquad\text{equivalently}\qquad
w_F = \frac{1}{K_S} w_S,
\]
which depends on $A_a$ only through $e_f^{\ast}$ and on the parameters and endowments, but not on the wage levels. Substituting $w_C = \Omega(A_a)w_S$ and $w_F = \frac{1}{K_S}w_S$ into \eqref{eq:unit_cost_anchor} and using $\alpha^S_{f}+\alpha^C_{f}+\alpha^F_{f}=1$ collapses the exponent on $w_S$ to one, giving $w_S\,\Omega^{\alpha^C_{f}}K_S^{-\alpha^F_{f}}=\Theta$, and hence the levels:
\begin{align}
    w_S(A_a) &= \Theta \cdot \Omega(A_a)^{-\alpha^C_{f}} K_S^{\alpha^F_{f}}, \label{eq:wS_trans_final} \\
    w_C(A_a) &= \Theta \cdot \Omega(A_a)^{1-\alpha^C_{f}} K_S^{\alpha^F_{f}}. \label{eq:wC_trans_final}
\end{align}
At the onset threshold $A_a^*$, $e_f^*=0$ (so $K_S = \frac{\alpha^S_{f} L^F}{\alpha^F_{f} L^S}$) and $\Omega(A_a^*) = \frac{L^S \alpha^C_{f}}{L^C \alpha^S_{f}}$. Substituting these into \eqref{eq:wS_trans_final} and \eqref{eq:wC_trans_final} recovers \eqref{eq:wS_pre_level} and \eqref{eq:wC_pre_level} exactly, verifying continuity. Because $K_S$ depends on the endowments through $L^F/L^S$ (and through $e_f^{\ast}$), the transition wage \emph{levels} are invariant to $L^S$ and $L^C$ only in the two-type case $\alpha^F_{f}=0$, where $K_S^{\alpha^F_{f}}=1$ and the levels reduce to $w_S=\Theta\,\Omega^{-\alpha^C_{f}}$, $w_C=\Theta\,\Omega^{\alpha^S_{f}}$ used in Section~\ref{SEC:INCOME_TAX}; the wage \emph{ratio} $w_C/w_S=\Omega$ is invariant to labor supply for every $\alpha^F_{f}$.

\paragraph{3. Post-AI Regime ($A_a \geq A_a^{**}$)}

In the post-AI regime, substitutable labor $L^S$ is entirely reallocated to the AI sector as a production input ($L^S_{a} = L^S$), and the final-good sector utilizes only $X_a$. The aggregate production function is:
\begin{equation}
    Y_f = A_f \left( A_a (L^S)^{\alpha^S_{a}} (L^C_{a})^{\alpha^C_{a}} \right)^{\alpha^S_{f}} (L^C_{f})^{\alpha^C_{f}} (L^F)^{\alpha^F_{f}}
\end{equation}
Substituting the fixed equilibrium labor allocations $\frac{L^C_{a}}{L^C} = \frac{\alpha^S_{f}\alpha^C_{a}}{\alpha^C_{f} + \alpha^S_{f}\alpha^C_{a}}$ and $\frac{L^C_{f}}{L^C} = \frac{\alpha^C_{f}}{\alpha^C_{f} + \alpha^S_{f}\alpha^C_{a}}$, the wages are derived as the marginal products of the aggregate endowments:

\paragraph{I. Substitutable Labor Wage ($w_S$)}
This labor now acts as a "builder" of AI technology. Its wage is the marginal product of $L^S$ in the consolidated production chain:
\begin{equation}
    w_S = \frac{\partial Y_f}{\partial L^S} = (\alpha^S_{f} \alpha^S_{a}) \frac{Y_f}{L^S} \propto A_a^{\alpha^S_{f}}
\end{equation}

\paragraph{II. Complementary Labor Wage ($w_C$)}
This labor provides inputs to both the creation of AI ($L^C_{a}$) and the final-good production ($L^C_{f}$). Its wage reflects its total marginal contribution to $Y_f$:
\begin{equation}
    w_C = \frac{\partial Y_f}{\partial L^C} = (\alpha^C_{f} + \alpha^S_{f} \alpha^C_{a}) \frac{Y_f}{L^C} \propto A_a^{\alpha^S_{f}}
\end{equation}

\paragraph{III. Final-Good-Specific Labor Wage ($w_F$)}
This labor is a pure complement to the automated production process. Its wage tracks the expansion of aggregate output:
\begin{equation}
    w_F = \frac{\partial Y_f}{\partial L^F} = \alpha^F_{f} \frac{Y_f}{L^F} \propto A_a^{\alpha^S_{f}}
\end{equation}

We now derive the relative wages $\frac{w_F}{w_S}$ and $\frac{w_F}{w_C}$ in each regime.

\medskip
\noindent
\textbf{1. Pre-AI Regime ($A_a \leq A_a^*$)}

In this regime, automation is not adopted ($Y_a=0$). The economy operates as a standard Cobb-Douglas production function with labor allocations $L^S_{f}=L^S$ and $L^C_{f}=L^C$. The relative wages are determined purely by the ratio of output elasticities and endowments:
\[
\frac{w_F}{w_S} = \frac{\alpha^F_{f} Y_f / L^F}{\alpha^S_{f} Y_f / L^S} = \frac{\alpha^F_{f}}{\alpha^S_{f}}\frac{L^S}{L^F},
\quad \text{and} \quad
\frac{w_F}{w_C} = \frac{\alpha^F_{f} Y_f / L^F}{\alpha^C_{f} Y_f / L^C} = \frac{\alpha^F_{f}}{\alpha^C_{f}}\frac{L^C}{L^F}.
\]

\medskip
\noindent
\textbf{2. Transition Regime ($A_a^* < A_a < A_a^{**}$)}

Let $Z = L^S_{f} + Y_a$ be the aggregate of substitutable tasks. From the final-good firm's FOCs, wages are $w_F = \alpha^F_{f} \frac{Y_f}{L^F}$ and $w_S = p_a = \alpha^S_{f} \frac{Y_f}{Z}$.
Taking the ratio yields:
\begin{equation}
\frac{w_F}{w_S} = \frac{\alpha^F_{f}}{\alpha^S_{f}} \frac{Z}{L^F}.
\label{eq:ratio_F_S_step1}
\end{equation}
We need to express $Z$ in terms of the endowment $L^S$. Recall the labor market clearing condition for $L^S$:
\[
L^S = L^S_{f} + L^S_{a} = (1-e_f^{\ast})Z + L^S_{a}.
\]
In the automation sector, $p_a Y_a = w_S L^S_{a} / \alpha^S_{a}$. Since $p_a = w_S$ in the transition, this implies $L^S_{a} = \alpha^S_{a} Y_a = \alpha^S_{a} e_f^{\ast} Z$. Substituting this back into the labor clearing condition:
\[
L^S = (1-e_f^{\ast})Z + \alpha^S_{a} e_f^{\ast} Z = Z [ 1 - e_f^{\ast}(1-\alpha^S_{a}) ].
\]
Thus, $Z = \frac{L^S}{1 - e_f^{\ast}(1-\alpha^S_{a})}$. Substituting $Z$ into Eq. \eqref{eq:ratio_F_S_step1}:
\[
\frac{w_F}{w_S} = \frac{\alpha^F_{f}}{\alpha^S_{f}} \frac{L^S}{L^F} \frac{1}{1 - e_f^{\ast}(1-\alpha^S_{a})}.
\]

Next, consider $\frac{w_F}{w_C}$. From the FOCs, $w_C$ equates to the marginal product of complementary labor in the final-good sector: $w_C = \alpha^C_{f} \frac{Y_f}{L^C_{f}}$. The ratio is:
\begin{equation}
\frac{w_F}{w_C} = \frac{\alpha^F_{f}}{\alpha^C_{f}} \frac{L^C_{f}}{L^F}.
\label{eq:ratio_F_C_step1}
\end{equation}
We express $L^C_{f}$ in terms of the endowment $L^C$. The total expenditure on complementary labor is $w_C L^C = w_C L^C_{f} + w_C L^C_{a}$.
Using the cost shares: $w_C L^C_{f} = \alpha^C_{f} p_f Y_f$ and $w_C L^C_{a} = \alpha^C_{a} p_a Y_a = \alpha^C_{a} (e_f^{\ast} \alpha^S_{f} p_f Y_f)$.
Summing these gives total income to $L^C$:
\[
w_C L^C = p_f Y_f (\alpha^C_{f} + e_f^{\ast} \alpha^S_{f} \alpha^C_{a}).
\]
Using $w_C = \alpha^C_{f} p_f Y_f / L^C_{f}$, we substitute $w_C$ on the LHS:
\[
\frac{\alpha^C_{f} p_f Y_f}{L^C_{f}} L^C = p_f Y_f (\alpha^C_{f} + e_f^{\ast} \alpha^S_{f} \alpha^C_{a})
\implies
L^C_{f} = L^C \frac{\alpha^C_{f}}{\alpha^C_{f} + e_f^{\ast} \alpha^S_{f} \alpha^C_{a}}.
\]
Substituting $L^C_{f}$ into Eq. \eqref{eq:ratio_F_C_step1}:
\[
\frac{w_F}{w_C} = \frac{\alpha^F_{f}}{\alpha^C_{f}} \frac{1}{L^F} \left( L^C \frac{\alpha^C_{f}}{\alpha^C_{f} + e_f^{\ast} \alpha^S_{f} \alpha^C_{a}} \right)
= \frac{\alpha^F_{f}}{\alpha^C_{f} + e_f^{\ast} \alpha^S_{f} \alpha^C_{a}} \frac{L^C}{L^F}.
\]

\medskip
\noindent
\textbf{3. Post-AI Regime ($A_a \geq A_a^{**}$)}

In this regime, automation is fully adopted such that $e_f^{\ast}=1$ (or effectively the limit as $L^S_{f} \to 0$).
For $\frac{w_F}{w_S}$, substitute $e_f^{\ast}=1$ into the transition formula:
\[
\frac{w_F}{w_S} = \frac{\alpha^F_{f}}{\alpha^S_{f}} \frac{L^S}{L^F} \frac{1}{1 - (1-\alpha^S_{a})} = \frac{\alpha^F_{f}}{\alpha^S_{f} \alpha^S_{a}} \frac{L^S}{L^F}.
\]
For $\frac{w_F}{w_C}$, substitute $e_f^{\ast}=1$ into the transition formula:
\[
\frac{w_F}{w_C} = \frac{\alpha^F_{f}}{\alpha^C_{f} + \alpha^S_{f} \alpha^C_{a}} \frac{L^C}{L^F}.
\]

\subsection{Proof of Proposition~\ref{prop_wagegap}: Growing Wage Inequality}

\begin{proof}
Consider the substitution phase \(A_a^{*} < A_a < A_a^{**}\).

\medskip
\noindent\textbf{Part 1: Elasticity of \(w_C/w_S\).}\\
During the substitution phase we have:
\[
\frac{w_C}{w_S} = \left( A_a (\alpha^C_{a})^{\alpha^C_{a}} (\alpha^S_{a})^{\alpha^S_{a}} \right)^{\frac{1}{\alpha^C_{a}}}.
\]
Taking logs:
\[
\ln\left(\frac{w_C}{w_S}\right) = \frac{1}{\alpha^C_{a}} \ln A_a + \frac{1}{\alpha^C_{a}} \ln\left[(\alpha^C_{a})^{\alpha^C_{a}} (\alpha^S_{a})^{\alpha^S_{a}}\right].
\]
Differentiating with respect to \(\ln A_a\):
\[
\frac{d \ln (w_C/w_S)}{d \ln A_a} = \frac{1}{\alpha^C_{a}}.
\]
\medskip

\noindent\textbf{Part 2: Elasticity of \(w_F/w_S\).}\\
During the substitution phase:
\[
\frac{w_F}{w_S} = \frac{\alpha^F_{f}}{\alpha^S_{f}} \frac{L^S}{L^F} \frac{1}{1 - e_f^{\ast} \alpha^C_{a}},
\]
where \(e_f^{\ast}\) is given by Equation \eqref{ea}. Define:
\[
\Lambda(A_a) = \frac{L^C}{L^S} \left( A_a (\alpha^C_{a})^{\alpha^C_{a}} (\alpha^S_{a})^{\alpha^S_{a}} \right)^{\frac{1}{\alpha^C_{a}}}.
\]
Note that \(\frac{d \ln \Lambda}{d \ln A_a} = \frac{1}{\alpha^C_{a}}\). From Equation \eqref{ea}:
\[
e_f^{\ast} = \frac{\alpha^S_{f}\Lambda - \alpha^C_{f}}{\alpha^S_{f}\alpha^C_{a}(1+\Lambda)}.
\]

Now compute the elasticity. First, note that:
\[
\frac{w_F}{w_S} \propto \frac{1}{1 - e_f^{\ast} \alpha^C_{a}}.
\]
Compute \(1 - e_f^{\ast} \alpha^C_{a}\):
\begin{align*}
1 - e_f^{\ast} \alpha^C_{a} 
&= 1 - \frac{\alpha^S_{f}\Lambda - \alpha^C_{f}}{\alpha^S_{f}(1+\Lambda)} \\
&= \frac{\alpha^S_{f}(1+\Lambda) - (\alpha^S_{f}\Lambda - \alpha^C_{f})}{\alpha^S_{f}(1+\Lambda)} \\
&= \frac{\alpha^S_{f} + \alpha^C_{f}}{\alpha^S_{f}(1+\Lambda)}.
\end{align*}

Thus:
\[
\frac{w_F}{w_S} = \frac{\alpha^F_{f}}{\alpha^S_{f}} \frac{L^S}{L^F} \cdot \frac{\alpha^S_{f}(1+\Lambda)}{\alpha^S_{f} + \alpha^C_{f}} = \frac{\alpha^F_{f}}{\alpha^S_{f} + \alpha^C_{f}} \frac{L^S}{L^F} (1+\Lambda).
\]

By taking logs:
\[
\ln\left(\frac{w_F}{w_S}\right) = \ln\left(\frac{\alpha^F_{f}}{\alpha^S_{f} + \alpha^C_{f}} \frac{L^S}{L^F}\right) + \ln(1+\Lambda).
\]

Differentiate with respect to \(\ln A_a\):
\[
\frac{d \ln (w_F/w_S)}{d \ln A_a} = \frac{1}{1+\Lambda} \cdot \frac{d\Lambda}{d\ln A_a} = \frac{1}{1+\Lambda} \cdot \Lambda \cdot \frac{1}{\alpha^C_{a}} = \frac{1}{\alpha^C_{a}} \frac{\Lambda}{1+\Lambda}.
\]
Since \(\Lambda > 0\) in the substitution phase, this elasticity is strictly positive.

\medskip
\noindent\textbf{Part 3: Elasticity of \(w_F/w_C\).}\\
Using the chain rule:
\[
\frac{d \ln (w_F/w_C)}{d \ln A_a} = \frac{d \ln (w_F/w_S)}{d \ln A_a} - \frac{d \ln (w_C/w_S)}{d \ln A_a} = \frac{1}{\alpha^C_{a}} \frac{\Lambda}{1+\Lambda} - \frac{1}{\alpha^C_{a}} = -\frac{1}{\alpha^C_{a}} \frac{1}{1+\Lambda}.
\]
Since \(\Lambda > 0\), this elasticity is strictly negative.

\medskip

\noindent\textbf{Part 4: Elasticity of final-good sector output.}\\
In the substitution phase ($A_a^* < A_a < A_a^{**}$), by the envelope theorem, the elasticity of final-good sector output value with respect to automation productivity is:
\[
\frac{\textup{d}\ln(p_fY_f)}{\textup{d}\ln(A_{a})}=e_{f}\alpha^S_{f}
\]
In other words, the elasticity of GDP with respect to $A_a$
is equal to the Domar weight of the AI sector.
To see this:
\[
\frac{dY_f}{dA_a} = \frac{\partial Y_f}{\partial Y_a} \cdot \frac{\partial Y_a}{\partial A_a} = p_a \cdot \frac{\partial Y_a}{\partial A_a}
\]
\[
\frac{d \ln Y_f}{d \ln A_a} = \frac{A_a}{Y_f} \cdot \frac{dY_f}{dA_a} = \frac{A_a}{Y_f} \cdot \left( p_a \cdot \frac{\partial Y_a}{\partial A_a} \right)
\]
Next, by taking the partial derivative of $Y_a$ w.r.t. $A_a$:
\[
\frac{\partial Y_a}{\partial A_a} = \frac{Y_a}{A_a}
\]
Substituting this in:
\[
\frac{d \ln Y_f}{d \ln A_a} = \frac{A_a}{Y_f} \cdot \left( p_a \cdot \frac{Y_a}{A_a} \right) = \frac{p_a Y_a}{Y_f}
\]
Since $p_f = 1$, this is:
\[
\frac{d \ln Y_f}{d \ln A_a} = \frac{p_a Y_a}{p_f Y_f} = e_f \alpha^S_{f}
\]
where $e_f \alpha^S_{f} = \frac{p_a Y_a}{p_f Y_f}$ is the Domar weight of the AI sector.

\medskip

\noindent\textbf{Part 5: Comparison.}

First, compare \(\frac{d \ln (w_F/w_S)}{d \ln A_a}\) and \(\frac{d \ln Y_f}{d \ln A_a}\):

\[
\frac{d \ln (w_F/w_S)}{d \ln A_a} - \frac{d \ln Y_f}{d \ln A_a} = 
\frac{1}{\alpha^C_{a}(1+\Lambda)} \left[ \Lambda - (\alpha^S_{f}\Lambda - \alpha^C_{f}) \right]. 
\]

Then, $\frac{d \ln (w_F/w_S)}{d \ln A_a} - \frac{d \ln Y_f}{d \ln A_a}>0$ holds for $\alpha^S_{f}<1$ and  $\Lambda>0$. Next, since \(e_f^{\ast} > 0\) in the substitution phase, we have \(\frac{d \ln Y_f}{d \ln A_a} = e_f^{\ast} \alpha^S_{f} > 0\) (the Domar weight derived in Part~4). Finally, from Part 1, \(\frac{d \ln (w_C/w_S)}{d \ln A_a} = \frac{1}{\alpha^C_{a}} > 0\). Since \(\frac{1}{\alpha^C_{a}} > \frac{\Lambda}{1+\Lambda} \cdot \frac{1}{\alpha^C_{a}}\) (because \(0 < \frac{\Lambda}{1+\Lambda} < 1\)), we have:
\[
\frac{d \ln (w_C/w_S)}{d \ln A_a} > \frac{d \ln (w_F/w_S)}{d \ln A_a}.
\]

Putting it all together:
\[
\frac{d \ln (w_C/w_S)}{d \ln A_a} > \frac{d \ln (w_F/w_S)}{d \ln A_a} > \frac{d \ln Y_f}{d \ln A_a} > 0 > \frac{d \ln (w_F/w_C)}{d \ln A_a}.
\]

This completes the proof of all statements on wage inequality.
\end{proof}

\subsection{Proof of Proposition~\ref{taxrate}}

\noindent Throughout we use the two-type transition wages of Section~\ref{SEC:INCOME_TAX} (the $\alpha^F_{f}=0$ specialization of the absolute-wage derivation in Appendix~\ref{appendix}),
\[
w_S=\Theta\,\Omega^{-\alpha^C_{f}},
\qquad
w_C=\Theta\,\Omega^{\alpha^S_{f}},
\qquad
\Omega=\frac{w_C}{w_S}=\left(A_a\Phi\right)^{1/\alpha^C_{a}},
\]
where $\Phi=(\alpha^C_{a})^{\alpha^C_{a}}(\alpha^S_{a})^{\alpha^S_{a}}$, $\Theta=A_f(\alpha^S_{f})^{\alpha^S_{f}}(\alpha^C_{f})^{\alpha^C_{f}}$, and $\alpha^S_{f}+\alpha^C_{f}=1$. Because wage taxes and subsidies are non-distortionary under inelastic labor supply, the pre-tax wages---and hence $\Omega$, which is strictly increasing in $A_a$---are unaffected by the policy.

\medskip
\noindent To see the first claim, balance the budget, $(w_C^{1}-\widetilde{w}_C^{1})\overline{L}^C=(\widetilde{w}_S^{1}-w_S^{1})\overline{L}^S$, and set $\widetilde{w}_S^{1}=w_S^{0}$. The required rate is
\[
\tau_C=\frac{w_C^{1}-\widetilde{w}_C^{1}}{w_C^{1}}=\frac{(w_S^{0}-w_S^{1})\,\overline{L}^S}{w_C^{1}\,\overline{L}^C}.
\]
Writing $\Omega_0=\Omega(A_a^{0})$ and $\Omega_1=\Omega(A_a^{1})$ and substituting the wage levels above, $\Theta$ cancels and
\[
\tau_C=\frac{\overline{L}^S}{\overline{L}^C}\,\frac{\Omega_0^{-\alpha^C_{f}}-\Omega_1^{-\alpha^C_{f}}}{\Omega_1^{\alpha^S_{f}}}.
\]
Let $x= A_a^{1}/A_a^{0}\geq 1$ be the gross productivity gain, so that $\Omega_1=\Omega_0\,x^{1/\alpha^C_{a}}$. Using $\alpha^S_{f}+\alpha^C_{f}=1$,
\[
\tau_C=\frac{\overline{L}^S}{\overline{L}^C}\,\Omega_0^{-1}\left[x^{-\alpha^S_{f}/\alpha^C_{a}}-x^{-1/\alpha^C_{a}}\right]
=\frac{\overline{L}^S}{\overline{L}^C}\,\Omega_0^{-1}\,g(x),
\qquad
g(x)= x^{-1/\alpha^C_{a}}\!\left(x^{\alpha^C_{f}/\alpha^C_{a}}-1\right).
\]
Hence $\tau_C=0$ at $x=1$ (i.e.\ at $A_a^{1}=A_a^{0}$) and $\tau_C>0$ for $x>1$. Differentiating $g$,
\[
g'(x)=\frac{1}{\alpha^C_{a}}\,x^{-\frac{1}{\alpha^C_{a}}-1}\left[\,1-\alpha^S_{f}\,x^{\alpha^C_{f}/\alpha^C_{a}}\,\right],
\]
and since the prefactor is strictly positive, $g'(x)$ shares the sign of $1-\alpha^S_{f}\,x^{\alpha^C_{f}/\alpha^C_{a}}$. Therefore $g'(x)>0$ for $x<x^{*}$ and $g'(x)<0$ for $x>x^{*}$, where
\[
x^{*}=(\alpha^S_{f})^{-\alpha^C_{a}/\alpha^C_{f}} .
\]
Since $0<\alpha^S_{f}<1$ we have $x^{*}>1$, and at the benchmark $g'(1)=\alpha^C_{f}/\alpha^C_{a}>0$. Thus $\tau_C$ rises from zero at $A_a^{0}$, attains a unique maximum at $\widehat{A}_a=A_a^{0}\,x^{*}=A_a^{0}(\alpha^S_{f})^{-\alpha^C_{a}/\alpha^C_{f}}$, and falls thereafter; if $\widehat{A}_a\geq A^{**}$ the maximizer lies beyond the transition, so $\tau_C$ is increasing on all of $[A_a^{0},A^{**}]$. This establishes the first claim.

The intuition is that $\tau_C$ is the substitutable-worker deficit relative to the complementary-worker base, $\tau_C=(w_S^{0}-w_S^{1})\overline{L}^S/(w_C^{1}\overline{L}^C)$. Near the benchmark the deficit grows from zero and so outpaces the (already positive) base, pushing $\tau_C$ up; only once $w_S^{1}$ flattens toward its post-transition level does the base overtake the deficit and $\tau_C$ decline.

\medskip
\noindent To see the second claim, note that the balanced-budget tax delivering $\widetilde{w}_C^{1}/\widetilde{w}_S^{1}=\gamma$ solves
\[
\frac{\widetilde{w}_C^{1}}{\widetilde{w}_S^{1}}=\gamma=\frac{w_C^{1}(1-\tau_C)}{w_S^{1}(1-\tau_S)}
\qquad\text{together with}\qquad
\tau_C\,w_C^{1}\overline{L}^C+\tau_S\,w_S^{1}\overline{L}^S=0 .
\]
Eliminating $\tau_S$ yields
\[
\tau_C=\frac{\overline{L}^S}{\gamma\,\overline{L}^C+\overline{L}^S}\left(1-\frac{\gamma}{\Omega_1}\right),
\qquad
\Omega_1=\frac{w_C^{1}}{w_S^{1}}=\left(A_a^{1}\Phi\right)^{1/\alpha^C_{a}} .
\]
The prefactor is a positive constant, and $\Omega_1$ is strictly increasing in $A_a^{1}$, so $\gamma/\Omega_1$ is strictly decreasing and $\tau_C$ is strictly increasing in the productivity gain (and positive exactly when $\gamma<\Omega_1$). This establishes the second claim.

\medskip
\noindent Finally, take $\gamma=\Omega_0=w_C^{0}/w_S^{0}$. The pre-transition wages lie on the frontier $w_S^{0}\overline{L}^S+w_C^{0}\overline{L}^C=Y_f^{0}$ along the ray $w_C/w_S=\gamma$, and the after-tax wages lie on the new frontier $\widetilde{w}_S^{1}\overline{L}^S+\widetilde{w}_C^{1}\overline{L}^C=Y_f^{1}$ along the same ray. Since output rises during the transition, $Y_f^{1}>Y_f^{0}$, the after-tax point lies strictly farther out along the ray, so $\widetilde{w}_S^{1}>w_S^{0}$ and $\widetilde{w}_C^{1}>w_C^{0}$ and the policy is Pareto-improving; by continuity this persists for $\gamma$ in a neighborhood of $\Omega_0$. For $\gamma$ large enough the ray is steep enough that $\widetilde{w}_C^{1}<w_C^{0}$, and the policy is no longer Pareto-improving.\eproof

\subsection{Proof of Proposition~\ref{taxrate2}}

Throughout we use the two-type transition wages,
\[
w_S = \Theta\,\Omega^{-\alpha^C_{f}}, \qquad
w_C = \Theta\,\Omega^{\alpha^S_{f}}, \qquad
\Omega \equiv \frac{w_C}{w_S} = (A_a\Phi)^{1/\alpha^C_{a}},
\]
where $\Phi = (\alpha^C_{a})^{\alpha^C_{a}}(\alpha^S_{a})^{\alpha^S_{a}}$, $\Theta = A_f(\alpha^S_{f})^{\alpha^S_{f}}(\alpha^C_{f})^{\alpha^C_{f}}$, and $\alpha^S_{f}+\alpha^C_{f}=1$. Write $\Omega_0=\Omega(A_a^{0})$ and $\Omega_1=\Omega(A_a^{1})$. From the proof of Proposition~\ref{taxrate}, the two required tax rates are
\[
\tau_C^{\,w} = \frac{(w_S^{0}-w_S^{1})\,\overline{L}^S}{w_C^{1}\,\overline{L}^C},
\qquad
\tau_C^{\,\gamma} = \frac{\overline{L}^S}{\gamma\overline{L}^C+\overline{L}^S}\left(1-\frac{\gamma}{\Omega_1}\right),
\]
the latter positive precisely when $\gamma<\Omega_1$. By the labor-supply invariance of Proposition~\ref{prop_wagegap}, the pre-tax wage levels $w_S^{0},w_S^{1},w_C^{1}$ and the ratio $\Omega$ do not depend on $\overline{L}^S,\overline{L}^C$.

\medskip
\noindent\textit{(i) Both targets are increasing in $\overline{L}^S/\overline{L}^C$.}
Let $r=\overline{L}^S/\overline{L}^C$. For the wage target, $\tau_C^{\,w}=r\,(w_S^{0}-w_S^{1})/w_C^{1}$, where $(w_S^{0}-w_S^{1})/w_C^{1}>0$ is independent of $r$ by invariance; hence $\tau_C^{\,w}$ is strictly increasing in $r$. For the $\gamma$ target, write $\dfrac{\overline{L}^S}{\gamma\overline{L}^C+\overline{L}^S}=\dfrac{r}{\gamma+r}$, so
\[
\tau_C^{\,\gamma}=\frac{r}{\gamma+r}\left(1-\frac{\gamma}{\Omega_1}\right),
\qquad
\frac{\partial}{\partial r}\!\left(\frac{r}{\gamma+r}\right)=\frac{\gamma}{(\gamma+r)^{2}}>0 .
\]
Since $1-\gamma/\Omega_1>0$ is independent of $r$, $\tau_C^{\,\gamma}$ is strictly increasing in $r$.

\medskip
\noindent\textit{(ii) Dependence on $\alpha^S_{f}$.}
Substituting the wage levels into $\tau_C^{\,w}$, the factor $\Theta$ cancels, and using $\alpha^S_{f}+\alpha^C_{f}=1$,
\[
\tau_C^{\,w}
=\frac{\overline{L}^S}{\overline{L}^C}\cdot
\frac{\Omega_0^{-\alpha^C_{f}}-\Omega_1^{-\alpha^C_{f}}}{\Omega_1^{\alpha^S_{f}}}
=\frac{\overline{L}^S}{\overline{L}^C}\left[
\frac{1}{(A_a^{0}\Phi)^{1/\alpha^C_{a}}}\left(\frac{A_a^{0}}{A_a^{1}}\right)^{\alpha^S_{f}/\alpha^C_{a}}
-\frac{1}{(A_a^{1}\Phi)^{1/\alpha^C_{a}}}\right].
\]
Since $\Phi$ and $\alpha^C_{a}=1-\alpha^S_{a}$ depend only on AI-sector parameters, the sole $\alpha^S_{f}$-dependent term is $\big(A_a^{0}/A_a^{1}\big)^{\alpha^S_{f}/\alpha^C_{a}}$. As $A_a^{0}/A_a^{1}<1$ and the exponent $\alpha^S_{f}/\alpha^C_{a}$ increases in $\alpha^S_{f}$, this term is strictly decreasing in $\alpha^S_{f}$; the remaining term is constant in $\alpha^S_{f}$. Hence $\tau_C^{\,w}$ is strictly decreasing in $\alpha^S_{f}$. For the $\gamma$ target, $\tau_C^{\,\gamma}$ depends on the technology only through $\Omega_1=(A_a^{1}\Phi)^{1/\alpha^C_{a}}$, an AI-sector object, so $\alpha^S_{f}$ does not appear and $\tau_C^{\,\gamma}$ is independent of $\alpha^S_{f}$.

\medskip
\noindent\textit{(iii) Dependence on $\alpha^S_{a}$.}
Since $\ln\Omega=\frac{1}{\alpha^C_{a}}\big[\ln A_a+\alpha^C_{a}\ln\alpha^C_{a}+\alpha^S_{a}\ln\alpha^S_{a}\big]$ with $\alpha^C_{a}=1-\alpha^S_{a}$, differentiating (the $\alpha^C_{a}\ln\alpha^C_{a}$ terms cancel) gives
\[
\frac{\partial\ln(w_C/w_S)}{\partial\alpha^S_{a}}=\frac{\ln(A_a\alpha^S_{a})}{(\alpha^C_{a})^{2}},
\]
positive when $A_a\alpha^S_{a}>1$ and negative when $A_a\alpha^S_{a}<1$. For the $\gamma$ target, $\tau_C^{\,\gamma}$ is increasing in $\Omega_1$ and $\partial\Omega_1/\partial\alpha^S_{a}=\Omega_1\ln(A_a^{1}\alpha^S_{a})/(\alpha^C_{a})^{2}$, so $\partial\tau_C^{\,\gamma}/\partial\alpha^S_{a}$ shares the sign of $\ln(A_a^{1}\alpha^S_{a})$. For the wage target, writing $\tau_C^{\,w}=\frac{\overline{L}^S}{\overline{L}^C}\big[\Omega_0^{-\alpha^C_{f}}\Omega_1^{-\alpha^S_{f}}-\Omega_1^{-1}\big]$, $\alpha^S_{a}$ enters through both $\Omega_0$ and $\Omega_1$, moving the deficit and the tax base simultaneously, so the net effect has no definite sign and depends on $A_a^{0}$ and $A_a^{1}$. \qed

\subsection{Derivations for Section~\ref{SEC:MONOPOLY}: Monopoly AI Production}
\label{app:monopoly}

Throughout this subsection we work with the two-type economy of Section~\ref{SEC:INCOME_TAX} ($\alpha^F_{f}=0$, $\overline{L}^F=0$):
\[
Y_a = A_a (L^S_{a})^{\alpha^S_{a}} (L^C_{a})^{\alpha^C_{a}}, \quad \alpha^S_{a}+\alpha^C_{a}=1,
\qquad
Y_f = A_f\big[L^S_{f}+X_a\big]^{\alpha^S_{f}} (L^C_{f})^{\alpha^C_{f}}, \quad \alpha^S_{f}+\alpha^C_{f}=1 .
\]
The final-good sector is perfectly competitive under both market structures; only the AI sector differs. We write $\Phi = (\alpha^S_{a})^{\alpha^S_{a}}(\alpha^C_{a})^{\alpha^C_{a}}$ and $\Theta = A_f (\alpha^S_{f})^{\alpha^S_{f}}(\alpha^C_{f})^{\alpha^C_{f}}$, and use superscripts $mon$ and $comp$ to denote the two regimes evaluated at a common productivity $A_a$. Part~A fixes the common structure; Part~B isolates the monopolist's allocation; Part~C proves the no-tax comparison (Proposition~\ref{prop:mon_comp_no_tax}) and the thresholds; Part~D records the wage-bill machinery; Part~E assembles the proof of Proposition~\ref{prop:ctax_mono}; Part~F derives the redistribution instruments and the regime ordering.

\subsubsection*{The monopoly-GE formulation}

As explained in Section~\ref{SEC:MONOPOLY}, we model the AI producer as a monopolist fully understanding the 
general-equilibrium setting within which they are embedded, rather than as a partial-equilibrium monopolist
facing fixed-wages. The monopolist chooses the AI price and the
inputs used in AI production, anticipating the equilibrium response of the
competitive final-good sector and labor markets.

The derivation below characterizes the non-degenerate interior transition branch
with \(0<\alpha_a^S<1\). The limiting case \(\alpha_a^S=0\), used in some
figures, is handled separately in the numerical implementation.

Given prices \((p_a,w_S^m,w_C^m)\), the competitive final-good sector chooses AI
use and final-good labor demands, \((X_a,L_f^S,L_f^C)\), to maximize
\[
A_f[L_f^S+X_a]^{\alpha_f^S}(L_f^C)^{\alpha_f^C}
-
p_aX_a
-
w_S^mL_f^S
-
w_C^mL_f^C .
\]
Let
\[
Z=L_f^S+X_a.
\]
During the interior transition, both AI and substitutable labor are used in
final-good production. The final-good sector's first-order conditions are
therefore
\[
p_a
=
\alpha_f^S A_f Z^{\alpha_f^S-1}(L_f^C)^{\alpha_f^C},
\]
\[
w_S^m
=
\alpha_f^S A_f Z^{\alpha_f^S-1}(L_f^C)^{\alpha_f^C},
\]
and
\[
w_C^m
=
\alpha_f^C A_f Z^{\alpha_f^S}(L_f^C)^{\alpha_f^C-1}.
\]
Since AI and substitutable labor enter the same final-good bundle, the first two
conditions imply
\[
p_a=w_S^m
\]
at any interior equilibrium.

The monopolist produces AI according to
\[
Y_a=A_a(L_a^S)^{\alpha_a^S}(L_a^C)^{\alpha_a^C}.
\]
Equilibrium requires AI-market clearing and labor-market clearing:
\[
X_a=Y_a,
\qquad
L_f^S=\overline L^S-L_a^S,
\qquad
L_f^C=\overline L^C-L_a^C.
\]
Thus, after imposing market clearing, we can write
\[
Z
=
\overline L^S-L_a^S
+
A_a(L_a^S)^{\alpha_a^S}(L_a^C)^{\alpha_a^C}.
\]
Hence the final-good sector's equilibrium reaction is
\begin{equation}
\begin{aligned}
p_a=w_S^m
&=
\alpha_f^S A_f
\left[
\overline L^S-L_a^S
+
A_a(L_a^S)^{\alpha_a^S}(L_a^C)^{\alpha_a^C}
\right]^{\alpha_f^S-1}
(\overline L^C-L_a^C)^{\alpha_f^C},
\\
w_C^m
&=
\alpha_f^C A_f
\left[
\overline L^S-L_a^S
+
A_a(L_a^S)^{\alpha_a^S}(L_a^C)^{\alpha_a^C}
\right]^{\alpha_f^S}
(\overline L^C-L_a^C)^{\alpha_f^C-1}.
\end{aligned}
\label{eq:wage_levels_mono}
\end{equation}

The monopolist's price-setting profit is
\[
\Pi^m
=
p_aY_a-w_S^mL_a^S-w_C^mL_a^C.
\]
Using the interior final-good reaction \(w_S^m=p_a\), this can be written as
\[
\Pi^m
=
w_S^m(Y_a-L_a^S)-w_C^mL_a^C.
\]
Therefore, along the interior general-equilibrium residual-demand locus, the
monopolist's problem can be written in reduced form as
\[
\max_{L_a^S,L_a^C}
\Pi^m(L_a^S,L_a^C)
=
w_S^m(L_a^S,L_a^C)
\left[
A_a(L_a^S)^{\alpha_a^S}(L_a^C)^{\alpha_a^C}
-
L_a^S
\right]
-
w_C^m(L_a^S,L_a^C)L_a^C,
\]
where \(w_S^m(L_a^S,L_a^C)\) and \(w_C^m(L_a^S,L_a^C)\) are the endogenous wage
functions defined in \eqref{eq:wage_levels_mono}. The price choice is embedded
in the inverse-demand relation
\[
p_a=w_S^m(L_a^S,L_a^C).
\]
Thus wages are not held fixed when the monopolist changes its input choices:
\[
\frac{\partial \Pi^m}{\partial L_a^C}
=
w_S^m\frac{\partial Y_a}{\partial L_a^C}
+
\frac{\partial w_S^m}{\partial L_a^C}(Y_a-L_a^S)
-
w_C^m
-
L_a^C\frac{\partial w_C^m}{\partial L_a^C}.
\]

To close the argument, consider price deviations by the AI producer.
There are two cases: price deviations that keep the economy inside the
transition region, and price cuts that move the economy to the boundary
\(L_f^S=0\).

First consider any price deviation that keeps the economy in the interior
transition, so that \(L_f^S>0\). In this region, final-good producers use both
AI and substitutable labor. Hence the equilibrium price of AI must satisfy
\[
p_a=w_S^m.
\]
Therefore, every interior price deviation corresponds to some feasible interior
input pair \((L_a^S,L_a^C)\), with the induced price given by the inverse
general-equilibrium relation
\[
p_a=w_S^m(L_a^S,L_a^C).
\]
Thus the monopolist's price choice is already embedded in the reduced-form
problem
\[
\max_{L_a^S,L_a^C}
\Pi^m(L_a^S,L_a^C).
\]
Let
\[
(L_a^{S,m},L_a^{C,m})
\]
denote the monopoly-GE solution in the interior transition region, and let
\[
p_a^m=w_S^m(L_a^{S,m},L_a^{C,m}).
\]
Then, by construction,
\[
\Pi^m(L_a^{S,m},L_a^{C,m})
\geq
\Pi^m(L_a^S,L_a^C)
\]
for every other feasible input pair that induces an interior-transition
equilibrium. Hence no price increase or price cut that leaves the economy in
the interior transition region can raise monopoly profit. Such a price deviation
is simply another point on the same general-equilibrium residual-demand locus,
and the reduced-form maximization already optimizes over that locus.

It remains to consider price cuts large enough to move the economy to the
boundary \(L_f^S=0\). Fix \(A_a\) and suppose that the GE-monopoly solution is an
interior-transition allocation,
\[
(L_a^{S,m},L_a^{C,m}),
\qquad
0<L_a^{S,m}<\overline L^S,
\qquad
0<L_a^{C,m}<\overline L^C,
\qquad
L_f^{S,m}>0.
\]
Let the corresponding interior profit be
\[
\Pi^{m,\mathrm{int}}
=
\Pi^m(L_a^{S,m},L_a^{C,m}).
\]

At the boundary \(L_f^S=0\), all substitutable labor is used in AI production:
\[
L_a^{S,B}=\overline L^S.
\]
For any feasible \(L_a^C\in(0,\overline L^C)\), boundary AI output is
\[
Y_a^B(L_a^C)
=
A_a(\overline L^S)^{\alpha_a^S}(L_a^C)^{\alpha_a^C},
\]
final-good complementary labor is
\[
L_f^{C,B}(L_a^C)
=
\overline L^C-L_a^C,
\]
and the effective substitutable input in final-good production is
\[
Z^B(L_a^C)=Y_a^B(L_a^C).
\]
The threshold boundary price is the final-good marginal value of the effective
substitutable input evaluated at the boundary:
\[
w_S^{B,*}(L_a^C)=p_a^{B,*}(L_a^C)
=
\alpha_f^S A_f
\left[Z^B(L_a^C)\right]^{\alpha_f^S-1}
\left[L_f^{C,B}(L_a^C)\right]^{\alpha_f^C}.
\]
At this threshold, final-good producers are just willing to purchase the
boundary AI quantity while using no substitutable labor in final-good production.
The corresponding complementary-labor wage is
\[
w_C^{B,*}(L_a^C)
=
\alpha_f^C A_f
\left[Z^B(L_a^C)\right]^{\alpha_f^S}
\left[L_f^{C,B}(L_a^C)\right]^{\alpha_f^C-1}.
\]
Since \(p_a^{B,*}\) is the highest price consistent with the boundary allocation,
any lower price that keeps \(L_f^S=0\) weakly lowers revenue and cannot generate
higher profit for the same boundary quantity. Therefore, it is enough to compare
the interior profit with the best boundary profit evaluated at the threshold
price:
\[
\Pi^{m,B,*}
=
\sup_{0<L_a^C<\overline L^C}
\left\{
p_a^{B,*}(L_a^C)Y_a^B(L_a^C)
-
p_a^{B,*}(L_a^C)\overline L^S
-
w_C^{B,*}(L_a^C)L_a^C
\right\}.
\]
Equivalently,
\[
\Pi^{m,B,*}
=
\sup_{0<L_a^C<\overline L^C}
\left\{
p_a^{B,*}(L_a^C)
\left[
Y_a^B(L_a^C)-\overline L^S
\right]
-
w_C^{B,*}(L_a^C)L_a^C
\right\}.
\]

Because the monopoly-GE solution at this \(A_a\) is assumed to be an
interior-transition solution, the best boundary profit cannot exceed the
interior monopoly-GE profit:
\[
\Pi^{m,B,*}\leq\Pi^{m,\mathrm{int}}.
\]
If the inequality failed, a boundary allocation evaluated at its threshold price
would yield strictly higher profit, so the GE-monopoly solution would not be
interior. Since lowering the price below \(p_a^{B,*}\) while keeping
\(L_f^S=0\) cannot improve on \(\Pi^{m,B,*}\), no price cut that moves the
economy to the boundary is profitable.

Combining the interior maximization with the boundary comparison, the
monopoly-GE allocation is profit maximizing against both types of deviations:
(i) price changes that keep the economy in the interior transition region, and
(ii) price cuts that move the economy to the boundary \(L_f^S=0\).

Lastly, we show that the interior monopoly-GEsolution is unique on the
non-degenerate transition branch. To see this, use the GE condition
\(MP_S^a=1\), which implies
\[
L_a^S=\alpha_a^S Y_a,
\qquad
Y_a=\kappa L_a^C,
\qquad
\kappa=\left[A_a(\alpha_a^S)^{\alpha_a^S}\right]^{1/\alpha_a^C}.
\]
Let \(L=L_a^C\) and define
\[
\chi=\alpha_a^C\kappa,
\qquad
Z=\overline L^S+\chi L,
\qquad
L_f^C=\overline L^C-L.
\]
Using
\[
w_S=\alpha_f^S A_f Z^{\alpha_f^S-1}(L_f^C)^{\alpha_f^C},
\qquad
w_C=\alpha_f^C A_f Z^{\alpha_f^S}(L_f^C)^{\alpha_f^C-1},
\]
the monopoly-GE profit along the interior branch becomes
\[
\Pi^m(L)
=
A_f L
\left[
\alpha_f^S \chi(\overline L^C-L)
-
\alpha_f^C(\overline L^S+\chi L)
\right]
(\overline L^S+\chi L)^{-\alpha_f^C}
(\overline L^C-L)^{-\alpha_f^S}.
\]
Equivalently, define the un-normalized allocation ratio
\[
\xi=\frac{Z}{L_f^C}
=
\frac{\overline L^S+\chi L}{\overline L^C-L}.
\]
Then profit is proportional to
\[
\Pi^m(\xi)
\propto
\frac{
(\xi-\underline \xi)(\overline \xi-\xi)
}{
\xi^{\alpha_f^C}(\xi+\chi)
},
\]
where
\[
\underline \xi=\frac{\overline L^S}{\overline L^C},
\qquad
\overline \xi=\frac{\alpha_f^S}{\alpha_f^C}\chi.
\]
Thus profit is zero at the two endpoints
\[
\xi=\underline \xi
\qquad\text{and}\qquad
\xi=\overline \xi,
\]
and is positive between them.

The first-order condition for an interior maximum is
\[
\frac{1}{\xi-\underline \xi}
-
\frac{1}{\overline \xi-\xi}
-
\frac{\alpha_f^C}{\xi}
-
\frac{1}{\xi+\chi}
=
0.
\]
Let
\[
t=\frac{\xi-\underline \xi}{\overline \xi-\xi}>0.
\]
After substituting this transformation, the first-order condition becomes a
cubic equation of the form
\[
t^3
+
a t^2
-
b t
-
c
=
0,
\]
where
\[
a=1-\alpha_f^S\left(1-\frac{\underline \xi}{\overline \xi}\right)>0,
\]
\[
b=
\left[
\frac{\underline \xi}{\overline \xi}
+
\alpha_f^S
-
\alpha_f^S\frac{\underline \xi}{\overline \xi}
\right]a
>0,
\]
and
\[
c=
\frac{\underline \xi}{\overline \xi}a>0.
\]
The signs of the coefficients are therefore
\[
+,\quad +,\quad -,\quad -.
\]
By Descartes' rule of signs, this cubic has at most one positive real root,
because its coefficient sequence has exactly one sign change. Since profit is
zero at the two endpoints and positive between them, an interior maximizer
exists, so the first-order condition has at least one positive root. Therefore,
the cubic has exactly one positive root. Hence there is exactly one interior
stationary point, and this stationary point is the unique interior maximizer on
the non-degenerate GE-monopoly transition branch.

In the Supplementary Appendix, we also study the alternative monopoly
formulation that we mentioned in the main text. In that formulation, the AI producer has market power in the AI
output market, but for each level of AI output that it sells, it chooses the cost-minimizing
input bundle. The distinction is that the monopoly-GE formulation that we have analyzed above includes
monopsony-like features where the monopolist understands that its use of substitutable workers affects the final good demand for AI.  Above, the firm chooses directly over feasible input pairs
\((L_a^S,L_a^C)\), internalizing the induced equilibrium wages and allocations.
The cost-minimization formulation instead restricts the firm to input bundles
that satisfy the cost-minimizing input-mix condition,
\[
\frac{MP_S^a}{MP_C^a}
=
\frac{w_S}{w_C},
\]
or equivalently,
\[
\frac{\alpha_a^S}{\alpha_a^C}
\frac{L_a^C}{L_a^S}
=
\frac{w_S}{w_C}.
\]
Thus the monopoly-GE problem has a weakly larger choice set than the
cost-minimization monopoly problem. Consequently, when both problems are solved
over the same transition domain, the monopoly-GE profit is weakly higher than
the monopoly profit under the cost-minimization formulation.

\subsubsection*{A. Pricing, unit cost, and the markup}

From the final-good sector's equilibrium reaction, any interior monopoly
allocation satisfies
\[
p_a=w_S^m,
\]
and the monopoly wages are the corresponding marginal products in the
competitive final-good sector:
\[
w_S^m
=
\alpha_f^S A_f Z^{\alpha_f^S-1}(L_f^C)^{\alpha_f^C},
\qquad
w_C^m
=
\alpha_f^C A_f Z^{\alpha_f^S}(L_f^C)^{\alpha_f^C-1},
\]
where
\[
Z=\overline{L}^S-L_a^S+Y_a,
\qquad
L_f^C=\overline{L}^C-L_a^C.
\]

The Cobb-Douglas unit cost of producing one unit of AI is
\[
c(w_S^m,w_C^m;A_a)
=
\frac{(w_S^m)^{\alpha_a^S}(w_C^m)^{\alpha_a^C}}{A_a\Phi},
\]
where
\[
\Phi
=
(\alpha_a^S)^{\alpha_a^S}(\alpha_a^C)^{\alpha_a^C}.
\]
Equivalently,
\[
c(w_S^m,w_C^m;A_a)
=
\frac{1}{A_a}
\left(\frac{w_S^m}{\alpha_a^S}\right)^{\alpha_a^S}
\left(\frac{w_C^m}{\alpha_a^C}\right)^{\alpha_a^C}.
\]

Define the AI markup by
\[
\mu
=
\frac{p_a}{c(w_S^m,w_C^m;A_a)}.
\]
Using the interior condition \(p_a=w_S^m\), we obtain
\[
\mu
=
\frac{w_S^m}{c(w_S^m,w_C^m;A_a)}
=
\frac{A_a\Phi w_S^m}
{(w_S^m)^{\alpha_a^S}(w_C^m)^{\alpha_a^C}}.
\]
Since \(\alpha_a^S+\alpha_a^C=1\), this becomes
\[
\mu
=
A_a\Phi
\left(\frac{w_S^m}{w_C^m}\right)^{\alpha_a^C}.
\]
Therefore,
\[
\left(\frac{w_C^m}{w_S^m}\right)^{\alpha_a^C}
=
\frac{A_a\Phi}{\mu},
\]
and hence, during the interior monopoly transition,
\begin{equation}
\frac{w_C^m}{w_S^m}
=
\left(\frac{A_a\Phi}{\mu}\right)^{1/\alpha_a^C},
\qquad
\mu=\frac{p_a}{c}=\frac{w_S^m}{c}\geq 1.
\label{eq:wage_ratio_markup}
\end{equation}
Here \(\mu\) is the markup of the AI price over unit cost. Under competition,
\(\mu=1\); under monopoly, the output restriction implies \(\mu>1\).

The planner has two instruments: a tax \(\tau_\pi\) on profit and a tax
\(\tau_C\) on complementary wages, financing the subsidy to substitutable
workers. The monopolist maximizes \((1-\tau_\pi)\Pi^m\). For any
\(\tau_\pi<1\), this is a positive scaling of \(\Pi^m\) and yields the untaxed
allocation, while at \(\tau_\pi=1\) it is indifferent across output levels and,
by our tie-breaking convention, chooses the competitive allocation.

\subsubsection*{B. The monopolist's allocation}

\begin{lemma}\label{lem:mon_alloc}
At any interior transition allocation the monopolist preserves the competitive input mix within the AI sector,
\[
L^S_{a}=\alpha^S_{a}Y_a,
\qquad
Y_a=\kappa\,L^C_{a},
\qquad
\kappa=\big[A_a(\alpha^S_{a})^{\alpha^S_{a}}\big]^{1/\alpha^C_{a}},
\]
and restricts scale relative to competition: $L^{C,m}_a<L^{C,comp}_a$, $L^{S,m}_a<L^{S,comp}_a$ and $Y_a^{m}<Y_a^{comp}$.
\end{lemma}

\begin{proof}
\emph{Input mix.} A competitive AI firm takes $w_S,w_C$ as given, so $\partial\Pi/\partial L^S_{a}=w_S(\partial Y_a/\partial L^S_{a}-1)=0$ gives $\partial Y_a/\partial L^S_{a}=\alpha^S_{a}Y_a/L^S_{a}=1$, i.e.\ $L^S_{a}=\alpha^S_{a}Y_a$. The monopolist instead internalizes that a change in $L^S_{a}$ moves both wages through $Z$, which shifts by $\partial Y_a/\partial L^S_{a}-1$:
\[
\frac{\partial w_S}{\partial L^S_{a}} = w_S(\alpha^S_{f}-1)\frac{\partial Y_a/\partial L^S_{a}-1}{Z},
\qquad
\frac{\partial w_C}{\partial L^S_{a}} = w_C\,\alpha^S_{f}\frac{\partial Y_a/\partial L^S_{a}-1}{Z}.
\]
Substituting into $\partial\Pi/\partial L^S_{a}=w_S(\partial Y_a/\partial L^S_{a}-1)+\tfrac{\partial w_S}{\partial L^S_{a}}(Y_a-L^S_{a})-\tfrac{\partial w_C}{\partial L^S_{a}}L^C_{a}$, factoring out $(\partial Y_a/\partial L^S_{a}-1)$, and using $Z=\overline{L}^S+(Y_a-L^S_{a})$ and \eqref{eq:profit_mono},
\[
\frac{\partial\Pi}{\partial L^S_{a}}=\Big(\frac{\partial Y_a}{\partial L^S_{a}}-1\Big)\cdot \frac{w_S\overline{L}^S+\alpha^S_{f}\,\Pi}{Z}.
\]

At the monopolist's optimum profit is nonnegative, $\Pi\geq 0$, since not operating always guarantees $\Pi=0$; the numerator $w_S\overline{L}^S+\alpha^S_{f}\Pi$ is then strictly positive (as $w_S\overline{L}^S>0$), so the second factor is strictly positive. The first-order
condition with respect to \(L_a^S\) therefore implies
\[
\frac{\partial Y_a}{\partial L_a^S}=1.
\]
Using the Cobb-Douglas AI production function,
\[
Y_a=A_a(L_a^S)^{\alpha_a^S}(L_a^C)^{\alpha_a^C},
\]
we have
\[
\frac{\partial Y_a}{\partial L_a^S}
=
\alpha_a^S\frac{Y_a}{L_a^S}.
\]
Thus
\[
\alpha_a^S\frac{Y_a}{L_a^S}=1,
\]
or equivalently
\[
L_a^S=\alpha_a^S Y_a.
\]
This is the same input-mix condition that arises under competition: although
the monopolist restricts the scale of AI production, it uses the same relative
mix of substitutable and complementary labor along the interior input-mix ray.

Substituting \(L_a^S=\alpha_a^S Y_a\) into the AI production function gives
\[
Y_a
=
A_a(\alpha_a^S Y_a)^{\alpha_a^S}(L_a^C)^{\alpha_a^C}.
\]
Since \(\alpha_a^S+\alpha_a^C=1\), this can be rearranged as
\[
Y_a^{\alpha_a^C}
=
A_a(\alpha_a^S)^{\alpha_a^S}(L_a^C)^{\alpha_a^C}.
\]
Taking both sides to the power \(1/\alpha_a^C\) yields
\[
Y_a=\kappa L_a^C,
\qquad
\kappa
=
\left[A_a(\alpha_a^S)^{\alpha_a^S}\right]^{1/\alpha_a^C}.
\]
Thus \(\kappa\) measures how AI output increases when the monopolist expands
production while maintaining the optimal input mix \(L_a^S=\alpha_a^S Y_a\).
It is therefore the marginal product of complementary labor along this
input-mix ray, not the partial marginal product of \(L_a^C\) holding \(L_a^S\)
fixed. The latter is
\[
\frac{\partial Y_a}{\partial L_a^C}
=
\alpha_a^C\frac{Y_a}{L_a^C}
=
\alpha_a^C\kappa.
\]

\emph{Scale.} We now characterize the scale of AI production after imposing the
input-mix condition derived above. Since
\[
L_a^S=\alpha_a^S Y_a
\qquad\text{and}\qquad
Y_a=\kappa L_a^C,
\]
we have
\[
L_a^S=\alpha_a^S\kappa L_a^C.
\]
Thus the effective substitutable input used in final-good production is
\[
Z
=
\overline L^S-L_a^S+Y_a
=
\overline L^S-\alpha_a^S\kappa L_a^C+\kappa L_a^C
=
\overline L^S+\alpha_a^C\kappa L_a^C,
\]
where we use \(\alpha_a^S+\alpha_a^C=1\). Similarly,
\[
L_f^C=\overline L^C-L_a^C.
\]
The net amount of AI output that exceeds the substitutable labor used in AI
production is
\[
Y_a-L_a^S
=
Y_a-\alpha_a^S Y_a
=
\alpha_a^C Y_a
=
\alpha_a^C\kappa L_a^C.
\]
Therefore, after imposing the optimal input mix, the monopolist's profit depends
only on the scale variable \(L_a^C\):
\[
\Pi(L_a^C)
=
w_S(L_a^C)(Y_a-L_a^S)-w_C(L_a^C)L_a^C.
\]
Using \(Y_a-L_a^S=\alpha_a^C\kappa L_a^C\), this becomes
\[
\Pi(L_a^C)
=
w_S(L_a^C)\alpha_a^C\kappa L_a^C
-
w_C(L_a^C)L_a^C,
\]
or equivalently
\[
\Pi(L_a^C)
=
L_a^C
\left[
\alpha_a^C\kappa w_S(L_a^C)-w_C(L_a^C)
\right].
\]
Define
\[
T(L_a^C)
\equiv
\alpha_a^C\kappa w_S(L_a^C)-w_C(L_a^C).
\]
Then
\[
\Pi(L_a^C)=L_a^C T(L_a^C).
\]

Differentiating \(\Pi(L_a^C)=L_a^C T(L_a^C)\) gives
\[
\Pi'(L_a^C)
=
T(L_a^C)+L_a^C T'(L_a^C).
\]
Equivalently, writing this out in terms of wages,
\[
\frac{\partial \Pi}{\partial L_a^C}
=
\underbrace{\left(\alpha_a^C\kappa w_S-w_C\right)}_{\text{competitive margin}}
+
\underbrace{
\alpha_a^C\kappa L_a^C\frac{\partial w_S}{\partial L_a^C}
-
L_a^C\frac{\partial w_C}{\partial L_a^C}
}_{\text{monopoly terms}}.
\]
The first term is exactly the competitive AI firm's scale margin. The second
term appears only because the monopolist internalizes the effect of changing
AI output on equilibrium wages.

The competitive scale \(L_a^{C,comp}\) satisfies the competitive AI-sector
first-order condition
\[
p_a\frac{\partial Y_a}{\partial L_a^C}=w_C.
\]
At the input-mix ray,
\[
\frac{\partial Y_a}{\partial L_a^C}
=
\alpha_a^C\frac{Y_a}{L_a^C}
=
\alpha_a^C\kappa.
\]
Using \(p_a=w_S\), the competitive condition becomes
\[
\alpha_a^C\kappa w_S=w_C.
\]
Therefore
\[
T(L_a^{C,comp})=0,
\]
and hence
\[
\Pi(L_a^{C,comp})
=
L_a^{C,comp}T(L_a^{C,comp})
=
0.
\]

We next compare the monopoly and competitive scales. Along the input-mix ray,
raising \(L_a^C\) increases the effective substitutable bundle,
\[
Z=\overline L^S+\alpha_a^C\kappa L_a^C,
\]
and lowers complementary labor remaining in final-good production,
\[
L_f^C=\overline L^C-L_a^C.
\]
Using the final-sector wage equations,
\[
w_S
=
\alpha_f^S A_f Z^{\alpha_f^S-1}(L_f^C)^{\alpha_f^C},
\]
and
\[
w_C
=
\alpha_f^C A_f Z^{\alpha_f^S}(L_f^C)^{\alpha_f^C-1},
\]
we obtain
\[
\frac{\partial w_S}{\partial L_a^C}<0
\qquad\text{and}\qquad
\frac{\partial w_C}{\partial L_a^C}>0.
\]
The first inequality holds because \(\alpha_f^S-1<0\) and \(L_f^C\) falls as
\(L_a^C\) rises. The second holds because \(Z\) rises and \(L_f^C\) falls, both
of which raise the marginal product of complementary labor in the final-good
sector.

Therefore
\[
T'(L_a^C)
=
\alpha_a^C\kappa\frac{\partial w_S}{\partial L_a^C}
-
\frac{\partial w_C}{\partial L_a^C}
<0.
\]
Since \(T(L_a^{C,comp})=0\) and \(T\) is strictly decreasing, we have
\[
T(L_a^C)>0
\quad\text{for}\quad
L_a^C<L_a^{C,comp},
\]
and
\[
T(L_a^C)<0
\quad\text{for}\quad
L_a^C>L_a^{C,comp}.
\]
Hence a small reduction below the competitive scale generates strictly positive
profit:
\[
\Pi(L_a^C)=L_a^C T(L_a^C)>0
\qquad
\text{for } L_a^C<L_a^{C,comp}\text{ sufficiently close to }L_a^{C,comp}.
\]

Moreover, for any \(L_a^C\geq L_a^{C,comp}\), we have \(T(L_a^C)\leq0\) and
\(T'(L_a^C)<0\). Therefore
\[
\Pi'(L_a^C)
=
T(L_a^C)+L_a^C T'(L_a^C)
<0
\qquad
\text{for all }L_a^C\geq L_a^{C,comp}.
\]
Thus profit is strictly decreasing at and above the competitive scale, so the
monopolist never chooses a scale weakly above the competitive allocation.

It remains to establish existence and interiority. On the half-open interval
\([0,\overline L^C)\), the profit function is continuous and satisfies
\[
\Pi(0)=0.
\]
As shown above, profit is strictly positive at some scale slightly below
\(L_a^{C,comp}\). In addition,
\[
\Pi(L_a^C)\to -\infty
\qquad\text{as}\qquad
L_a^C\uparrow \overline L^C,
\]
because
\[
L_f^C=\overline L^C-L_a^C\to0,
\]
which implies
\[
w_C
=
\alpha_f^C A_f Z^{\alpha_f^S}(L_f^C)^{\alpha_f^C-1}
\to\infty.
\]
The divergence of \(w_C\) drives the cost term \(-w_C L_a^C\) to \(-\infty\).
Thus the maximizer is confined to some compact subinterval
\[
[0,b]\subset [0,\overline L^C),
\qquad b<\overline L^C,
\]
on which a maximum is attained. Since the maximum value is strictly positive
while \(\Pi(0)=0\), the maximizer cannot be zero. Since profit is strictly
decreasing for all \(L_a^C\geq L_a^{C,comp}\), the maximizer cannot lie weakly
above the competitive scale. Therefore
\[
0<L_a^{C,m}<L_a^{C,comp}.
\]

Finally, along the input-mix ray,
\[
Y_a=\kappa L_a^C
\qquad\text{and}\qquad
L_a^S=\alpha_a^S Y_a.
\]
Hence the lower monopoly scale immediately implies
\[
Y_a^m<Y_a^{comp}
\qquad\text{and}\qquad
L_a^{S,m}<L_a^{S,comp}.
\]
\end{proof}

\subsubsection*{C. Output, wages, and adoption thresholds}

\begin{proposition}
\label{prop:mon_comp_no_tax}
For any AI productivity $A^{*}<A_a\leq A^{**}_{comp}$, the untaxed monopoly equilibrium, relative to the competitive benchmark, satisfies
\[
\textup{(i)}\ \ Y_f^{m} < Y_f^{comp},
\qquad
\textup{(ii)}\ \ w_C^{m} < w_C^{comp},
\qquad
\textup{(iii)}\ \ w_S^{m} > w_S^{comp}.
\]

Equivalently, the markup compresses the wage ratio: $w_C^{m}/w_S^{m}=(A_a\Phi/\mu)^{1/\alpha^C_{a}}<(A_a\Phi)^{1/\alpha^C_{a}}=w_C^{comp}/w_S^{comp}$.
\end{proposition}

\begin{proof}
By Lemma~\ref{lem:mon_alloc}, the monopolist preserves the same AI-sector
input mix as in competition but chooses a strictly smaller scale:
\[
L_a^{C,m}<L_a^{C,comp},
\qquad
Y_a^m<Y_a^{comp},
\qquad
L_a^{S,m}<L_a^{S,comp}.
\]
Along this common input-mix ray,
\[
Z=\overline L^S-L_a^S+Y_a
=
\overline L^S+\alpha_a^C\kappa L_a^C,
\qquad
L_f^C=\overline L^C-L_a^C.
\]
Hence the lower monopoly scale implies
\[
Z^m<Z^{comp}
\qquad\text{and}\qquad
L_f^{C,m}>L_f^{C,comp}.
\]

The output comparison follows directly from the fact that the competitive
allocation maximizes final output subject to the technologies and resource
constraints. Since the monopoly allocation is feasible but differs from the
competitive allocation, strict concavity of the final-good production function
along the relevant interior allocation set implies
\[
Y_f^m<Y_f^{comp}.
\]
This proves (i).

We next compare wages. In the final-good sector,
\[
w_S
=
\alpha_f^S A_f Z^{\alpha_f^S-1}(L_f^C)^{\alpha_f^C}.
\]
Because \(\alpha_f^S+\alpha_f^C=1\), we have
\(\alpha_f^S-1=-\alpha_f^C<0\). Thus \(w_S\) is decreasing in \(Z\) and
increasing in \(L_f^C\). Since monopoly has lower \(Z\) and higher \(L_f^C\),
both effects raise the substitutable wage relative to competition:
\[
w_S^m>w_S^{comp}.
\]
This proves (iii).

Similarly,
\[
w_C
=
\alpha_f^C A_f Z^{\alpha_f^S}(L_f^C)^{\alpha_f^C-1}.
\]
Here \(w_C\) is increasing in \(Z\) and decreasing in \(L_f^C\), since
\(\alpha_f^C-1=-\alpha_f^S<0\). The monopoly allocation has lower \(Z\) and
higher \(L_f^C\), so both effects lower the complementary wage:
\[
w_C^m<w_C^{comp}.
\]
This proves (ii).

Finally, the wage-ratio comparison follows from the markup formula
\eqref{eq:wage_ratio_markup}. Under monopoly,
\[
\frac{w_C^m}{w_S^m}
=
\left(\frac{A_a\Phi}{\mu}\right)^{1/\alpha_a^C},
\]
whereas under competition price equals unit cost, so \(\mu=1\) and
\[
\frac{w_C^{comp}}{w_S^{comp}}
=
(A_a\Phi)^{1/\alpha_a^C}.
\]
Since the monopolist restricts output relative to competition, it charges a
strict markup over unit cost, \(\mu>1\). Therefore,
\[
\frac{w_C^m}{w_S^m}
=
\left(\frac{A_a\Phi}{\mu}\right)^{1/\alpha_a^C}
<
(A_a\Phi)^{1/\alpha_a^C}
=
\frac{w_C^{comp}}{w_S^{comp}}.
\]
Thus monopoly compresses the complementary-to-substitutable wage ratio relative
to competition.
\end{proof}

\noindent\textbf{Thresholds.} The onset threshold is unchanged, $A^{*}_{m}=A^{*}_{comp}=A^{*}$: a monopolist with zero output has $Y_a-L^S_{a}=0$ and $L^C_{a}=0$, so the monopoly terms in the first-order condition vanish and it faces the same entry margin as a competitive firm. The completion threshold is strictly higher, $A^{**}_{m}>A^{**}_{comp}$: since $L^S_{a}=\alpha^S_{a}\kappa L^C_{a}$ and $L^{C,m}_a<L^{C,comp}_a$ throughout the competitive transition, at $A_a=A^{**}_{comp}$ the competitive economy reaches $L^{S,comp}_a=\overline{L}^S$ while the monopoly economy still has $L^{S,m}_a<\overline{L}^S$, so full displacement requires strictly higher productivity.

\subsubsection*{D. The allocation ratio and the wage bill}

\begin{lemma}\label{lem:wagebill}
Define the allocation ratio $R = (Z/\overline{L}^S)/(L^C_{f}/\overline{L}^C)$. Then $R=1$ in the pre-AI economy and $R>1$ for every $A_a>A^{*}$ at which AI is active. Wages and the total wage bill satisfy
\begin{equation}
\frac{w_S}{w_S^{pre}}=R^{\alpha^S_{f}-1},
\qquad
\frac{w_C}{w_C^{pre}}=R^{\alpha^S_{f}},
\qquad
W\equiv w_S\overline{L}^S+w_C\overline{L}^C=Y^{pre}f(R),
\label{eq:wages_in_R}
\end{equation}
with $f(R)=\alpha^S_{f}R^{\alpha^S_{f}-1}+\alpha^C_{f}R^{\alpha^S_{f}}$ strictly increasing for $R>1$. Under monopoly,
\begin{equation}
R^{m}=\frac{\alpha^S_{f}}{\alpha^C_{f}}\,\frac{\overline{L}^C}{\overline{L}^S}\,\Big(\frac{A_a\Phi}{\mu}\Big)^{1/\alpha^C_{a}},
\label{eq:R_m}
\end{equation}
so $R^{m}$ is strictly increasing in $A_a$ if and only if $A_a/\mu$ is increasing.
\end{lemma}

\begin{proof}
Before AI adoption, \(Y_a=0\), \(L_a^S=L_a^C=0\), and hence
\[
Z=\overline L^S,
\qquad
L_f^C=\overline L^C.
\]
Therefore \(R=1\). Once AI is active, the input-mix result implies
\[
Y_a-L_a^S=\alpha_a^C Y_a>0,
\]
so
\[
Z=\overline L^S-L_a^S+Y_a
=
\overline L^S+\alpha_a^C Y_a
>
\overline L^S.
\]
Moreover, \(L_a^C>0\), so
\[
L_f^C=\overline L^C-L_a^C<\overline L^C.
\]
Thus
\[
R
=
\frac{Z/\overline L^S}{L_f^C/\overline L^C}
>
1.
\]

We next express wages in terms of \(R\). By definition,
\[
\frac{Z}{\overline L^S}
=
R\frac{L_f^C}{\overline L^C}.
\]
Using the final-good-sector wage equations and dividing by the pre-AI wages,
\[
w_S^{pre}
=
\alpha_f^S A_f(\overline L^S)^{\alpha_f^S-1}
(\overline L^C)^{\alpha_f^C},
\qquad
w_C^{pre}
=
\alpha_f^C A_f(\overline L^S)^{\alpha_f^S}
(\overline L^C)^{\alpha_f^C-1},
\]
we obtain
\[
\frac{w_S}{w_S^{pre}}
=
\left(\frac{Z}{\overline L^S}\right)^{\alpha_f^S-1}
\left(\frac{L_f^C}{\overline L^C}\right)^{\alpha_f^C}
=
R^{\alpha_f^S-1},
\]
where the powers of \(L_f^C/\overline L^C\) cancel because
\(\alpha_f^S+\alpha_f^C=1\). Similarly,
\[
\frac{w_C}{w_C^{pre}}
=
\left(\frac{Z}{\overline L^S}\right)^{\alpha_f^S}
\left(\frac{L_f^C}{\overline L^C}\right)^{\alpha_f^C-1}
=
R^{\alpha_f^S}.
\]
This proves the wage expressions in \eqref{eq:wages_in_R}.

The total wage bill is therefore
\[
W
=
w_S^{pre}\overline L^S R^{\alpha_f^S-1}
+
w_C^{pre}\overline L^C R^{\alpha_f^S}.
\]
Since, in the pre-AI economy,
\[
w_S^{pre}\overline L^S=\alpha_f^S Y^{pre},
\qquad
w_C^{pre}\overline L^C=\alpha_f^C Y^{pre},
\]
we get
\[
W
=
Y^{pre}
\left[
\alpha_f^S R^{\alpha_f^S-1}
+
\alpha_f^C R^{\alpha_f^S}
\right]
=
Y^{pre}f(R),
\]
where
\[
f(R)
=
\alpha_f^S R^{\alpha_f^S-1}
+
\alpha_f^C R^{\alpha_f^S}.
\]
Differentiating,
\[
f'(R)
=
\alpha_f^S\alpha_f^C R^{\alpha_f^S-2}(R-1),
\]
so
\[
f'(R)>0
\qquad
\text{for all }R>1.
\]

Finally, we derive the monopoly expression for \(R^m\). From the final-good-sector
wage equations,
\[
\frac{w_C^m}{w_S^m}
=
\frac{\alpha_f^C}{\alpha_f^S}
\frac{Z^m}{L_f^{C,m}}.
\]
Using
\[
R^m
=
\frac{Z^m/\overline L^S}{L_f^{C,m}/\overline L^C},
\]
we have
\[
\frac{Z^m}{L_f^{C,m}}
=
R^m\frac{\overline L^S}{\overline L^C}.
\]
Hence
\[
\frac{w_C^m}{w_S^m}
=
\frac{\alpha_f^C}{\alpha_f^S}
R^m
\frac{\overline L^S}{\overline L^C}.
\]
Solving for \(R^m\),
\[
R^m
=
\frac{\alpha_f^S}{\alpha_f^C}
\frac{\overline L^C}{\overline L^S}
\frac{w_C^m}{w_S^m}.
\]
Using \eqref{eq:wage_ratio_markup},
\[
\frac{w_C^m}{w_S^m}
=
\left(\frac{A_a\Phi}{\mu}\right)^{1/\alpha_a^C},
\]
we obtain
\[
R^m
=
\frac{\alpha_f^S}{\alpha_f^C}
\frac{\overline L^C}{\overline L^S}
\left(\frac{A_a\Phi}{\mu}\right)^{1/\alpha_a^C}.
\]
This is \eqref{eq:R_m}. Thus monopoly affects the allocation ratio through the
same productivity term as competition, but discounted by the markup \(\mu\).
\end{proof}

\subsubsection*{E. Proof of Proposition~\ref{prop:ctax_mono}}

\begin{proof}
With \(\tau_\pi=0\), after-tax monopoly profit is simply
\(\widetilde{\Pi}^m=\Pi^m\geq 0\), since the monopolist can always choose not to
produce and earn zero profit. Thus the profit constraint is automatically
satisfied.

Now consider the policy that protects substitutable workers exactly:
\[
\widetilde w_S^1=w_S^{m,0}.
\]
Relative to the monopoly wage at \(A_a^1\), this requires a subsidy payment
\[
(w_S^{m,0}-w_S^{m,1})\overline L^S
\]
to substitutable workers. With \(\tau_\pi=0\), this transfer is financed entirely
by complementary workers. Hence the after-tax income of complementary workers is
given by the balanced-budget condition
\[
\widetilde w_C^1\overline L^C
=
w_C^{m,1}\overline L^C
-
(w_S^{m,0}-w_S^{m,1})\overline L^S.
\]

Since substitutable workers are held exactly at their benchmark wage, the policy
is Pareto-improving for workers if and only if complementary workers are strictly
better off:
\[
\widetilde w_C^1>w_C^{m,0}.
\]
Using the previous display, this condition is equivalent to
\[
w_C^{m,1}\overline L^C
-
(w_S^{m,0}-w_S^{m,1})\overline L^S
>
w_C^{m,0}\overline L^C,
\]
or
\[
(w_C^{m,1}-w_C^{m,0})\overline L^C
>
(w_S^{m,0}-w_S^{m,1})\overline L^S .
\]
This is condition~\eqref{wagetaxonly}.

Rearranging the same inequality gives
\[
w_C^{m,1}\overline L^C+w_S^{m,1}\overline L^S
>
w_C^{m,0}\overline L^C+w_S^{m,0}\overline L^S.
\]
Thus the wage-tax-only policy is Pareto-improving if and only if the total wage
bill rises:
\[
W^1>W^0.
\]

By Lemma~\ref{lem:wagebill},
\[
W=Y^{pre}f(R),
\]
where \(Y^{pre}\) is fixed across the comparison and \(f\) is strictly increasing
for \(R>1\). Since both allocations are in the interior transition, 
\(R^{m,0},R^{m,1}>1\). Therefore
\[
W^1>W^0
\quad\Longleftrightarrow\quad
R^{m,1}>R^{m,0}.
\]

Finally, from \eqref{eq:R_m},
\[
R^m
=
\frac{\alpha_f^S}{\alpha_f^C}
\frac{\overline L^C}{\overline L^S}
\left(\frac{A_a\Phi}{\mu}\right)^{1/\alpha_a^C}.
\]
The terms
\[
\frac{\alpha_f^S}{\alpha_f^C},
\qquad
\frac{\overline L^C}{\overline L^S},
\qquad
\Phi
\]
are fixed across the comparison. Since \(1/\alpha_a^C>0\), we have
\[
R^{m,1}>R^{m,0}
\quad\Longleftrightarrow\quad
\frac{A_a^1}{\mu^1}>\frac{A_a^0}{\mu^0}.
\]
Equivalently,
\[
\frac{A_a^1}{A_a^0}>\frac{\mu^1}{\mu^0}.
\]
Thus the wage-tax-only policy is Pareto-improving if and only if AI productivity
grows faster than the monopoly markup. This gives the stated necessary and
sufficient condition.
\end{proof}

\begin{lemma}[Markup growth is slower than productivity growth]
\label{lem:markup_slower_than_productivity}
Along the interior monopoly transition, $A_a/\mu$ is strictly increasing in $A_a$.
Equivalently, for any two productivity levels $A_a^1>A_a^0$ in the interior monopoly
transition,
\[
\frac{A_a^1}{\mu^1}>\frac{A_a^0}{\mu^0},
\qquad\text{or}\qquad
\frac{A_a^1}{A_a^0}>\frac{\mu^1}{\mu^0}.
\]

\end{lemma}

\begin{proof}
Recall from Lemma~\ref{lem:mon_alloc} that, at any interior monopoly
transition allocation, the monopolist preserves the competitive input mix inside
the AI sector:
\[
L_a^S=\alpha_a^S Y_a,
\qquad
Y_a=\kappa L_a^C,
\qquad
\kappa=\left[A_a(\alpha_a^S)^{\alpha_a^S}\right]^{1/\alpha_a^C}.
\]
Thus
\[
Y_a-L_a^S=\alpha_a^C Y_a=\alpha_a^C\kappa L_a^C,
\]
and
\[
Z=\overline L^S+\alpha_a^C\kappa L_a^C,
\qquad
L_f^C=\overline L^C-L_a^C.
\]
Since \(\kappa\) is strictly increasing in \(A_a\), it is enough to show that the
monopoly allocation ratio \(R^m\) is strictly increasing in \(\kappa\).

We use the allocation ratio
\[
R
=
\frac{Z/\overline L^S}{L_f^C/\overline L^C}.
\]
Substituting
\[
Z=\overline L^S+\alpha_a^C\kappa L_a^C,
\qquad
L_f^C=\overline L^C-L_a^C,
\]
and solving for \(L_a^C\) gives
\[
L_a^C(R,\kappa)
=
\frac{(R-1)\overline L^S\overline L^C}
{R\overline L^S+\alpha_a^C\kappa\overline L^C}.
\]
The final-good-sector wage equations imply
\[
\frac{w_C}{w_S}
=
\frac{\alpha_f^C}{\alpha_f^S}\frac{Z}{L_f^C}
=
\frac{\alpha_f^C}{\alpha_f^S}
R\frac{\overline L^S}{\overline L^C},
\]
and, by Lemma~\ref{lem:wagebill},
\[
w_S=w_S^{pre}R^{\alpha_f^S-1}.
\]

Using
\[
\Pi
=
w_S(Y_a-L_a^S)-w_C L_a^C,
\]
and the input-mix relation
\[
Y_a-L_a^S=\alpha_a^C\kappa L_a^C,
\]
profit can be written as
\[
\Pi
=
L_a^C(\alpha_a^C\kappa w_S-w_C).
\]
Now substitute the expressions derived above. First,
\[
L_a^C(R,\kappa)
=
\frac{(R-1)\overline L^S\overline L^C}
{R\overline L^S+\alpha_a^C\kappa\overline L^C}.
\]
Second,
\[
w_S=w_S^{pre}R^{\alpha_f^S-1}.
\]
Third,
\[
\alpha_a^C\kappa w_S-w_C
=
w_S
\left[
\alpha_a^C\kappa
-
\frac{w_C}{w_S}
\right]
=
w_S
\left[
\alpha_a^C\kappa
-
\frac{\alpha_f^C}{\alpha_f^S}
R\frac{\overline L^S}{\overline L^C}
\right].
\]
Therefore
\[
\Pi(R,\kappa)
=
w_S^{pre}\overline L^S\overline L^C
\,
R^{\alpha_f^S-1}
(R-1)
\frac{
\alpha_a^C\kappa
-
\frac{\alpha_f^C}{\alpha_f^S}
R\frac{\overline L^S}{\overline L^C}
}{
R\overline L^S+\alpha_a^C\kappa\overline L^C
}.
\]
Since \(w_S^{pre}\overline L^S\overline L^C>0\) is independent of \(R\) and
\(\kappa\), maximizing profit is equivalent to maximizing
\[
R^{\alpha_f^S-1}
(R-1)
\frac{
\alpha_a^C\kappa
-
\frac{\alpha_f^C}{\alpha_f^S}
R\frac{\overline L^S}{\overline L^C}
}{
R\overline L^S+\alpha_a^C\kappa\overline L^C
}.
\]

We restrict attention to the interior positive-profit domain. Since
\(L_a^C>0\), positive profit is equivalent to
\[
\alpha_a^C\kappa w_S-w_C>0.
\]
Using the wage-ratio expression above, this becomes
\[
\alpha_a^C\kappa
-
\frac{\alpha_f^C}{\alpha_f^S}
R\frac{\overline L^S}{\overline L^C}
>0.
\]

On the interior positive-profit domain, \(\Pi(R,\kappa)>0\), so maximizing
\(\Pi\) is equivalent to maximizing \(\log \Pi\). The positive-profit condition is
\[
\alpha_a^C\kappa
-
\frac{\alpha_f^C}{\alpha_f^S}
R\frac{\overline L^S}{\overline L^C}
>0.
\]

For a fixed \(R\), this condition becomes easier to satisfy when \(\kappa\)
increases, and for a fixed \(\kappa\), it becomes easier to satisfy when \(R\)
decreases. Therefore, in the comparison below, if \(R^{m,1}<R^{m,0}\), then
\(R^{m,0}\) is feasible at \(\kappa^1\) because \(\kappa^1>\kappa^0\), and
\(R^{m,1}\) is feasible at \(\kappa^0\) because \(R^{m,1}<R^{m,0}\). Hence both
candidate values can be evaluated at both productivity levels.

Dropping only additive constants that are
independent of \(R\), define
\[
\ell(R,\kappa)
=
\log(R-1)
+
(\alpha_f^S-1)\log R
+
\log\left[
\alpha_a^C\kappa
-
\frac{\alpha_f^C}{\alpha_f^S}
R\frac{\overline L^S}{\overline L^C}
\right]
-
\log\left[
R\overline L^S+\alpha_a^C\kappa\overline L^C
\right].
\]
This objective has strict increasing differences in \((R,\kappa)\). Indeed,
\[
\frac{\partial^2 \ell}{\partial R\,\partial \kappa}
=
\frac{
\alpha_a^C
\frac{\alpha_f^C}{\alpha_f^S}
\frac{\overline L^S}{\overline L^C}
}{
\left[
\alpha_a^C\kappa
-
\frac{\alpha_f^C}{\alpha_f^S}
R\frac{\overline L^S}{\overline L^C}
\right]^2
}
+
\frac{
\alpha_a^C\overline L^S\overline L^C
}{
\left[
R\overline L^S+\alpha_a^C\kappa\overline L^C
\right]^2
}
>0.
\]
Thus a higher \(\kappa\) raises the marginal return to choosing a higher \(R\).

Now take two productivity levels \(A_a^1>A_a^0\). Since
\[
\kappa
=
\left[A_a(\alpha_a^S)^{\alpha_a^S}\right]^{1/\alpha_a^C},
\]
we have
\[
\kappa^1>\kappa^0.
\]
Let \(R^{m,1}\) and \(R^{m,0}\) denote the corresponding interior monopoly
maximizers.

We first show that \(R^{m,1}\geq R^{m,0}\). Suppose, toward a contradiction, that
\[
R^{m,1}<R^{m,0}.
\]
Since \(R^{m,0}\) is optimal at \(\kappa^0\),
\[
\ell(R^{m,0},\kappa^0)
\geq
\ell(R^{m,1},\kappa^0).
\]
Since \(R^{m,1}\) is optimal at \(\kappa^1\),
\[
\ell(R^{m,1},\kappa^1)
\geq
\ell(R^{m,0},\kappa^1).
\]
Adding these two optimality inequalities and rearranging gives
\[
\ell(R^{m,1},\kappa^1)-\ell(R^{m,1},\kappa^0)
\geq
\ell(R^{m,0},\kappa^1)-\ell(R^{m,0},\kappa^0).
\]
However, strict increasing differences mean that the gain from raising
\(\kappa\) is larger at higher values of \(R\). Therefore, if
\(R^{m,1}<R^{m,0}\) and \(\kappa^1>\kappa^0\), then
\[
\ell(R^{m,1},\kappa^1)-\ell(R^{m,1},\kappa^0)
<
\ell(R^{m,0},\kappa^1)-\ell(R^{m,0},\kappa^0).
\]
This contradicts the previous inequality. Hence
\[
R^{m,1}\geq R^{m,0}.
\]

The inequality is strict at an interior optimum. Suppose instead that
\[
R^{m,1}=R^{m,0}.
\]
Since \(R^{m,0}\) is an interior maximizer at \(\kappa^0\), its first-order
condition is
\[
\frac{\partial \ell(R^{m,0},\kappa^0)}{\partial R}=0.
\]
Strict increasing differences imply that increasing \(\kappa\) strictly raises
the marginal payoff to \(R\). Hence
\[
\frac{\partial \ell(R^{m,0},\kappa^1)}{\partial R}
>
\frac{\partial \ell(R^{m,0},\kappa^0)}{\partial R}
=
0.
\]
Therefore the same \(R^{m,0}\) cannot also be an interior maximizer at
\(\kappa^1\). Hence
\[
R^{m,1}>R^{m,0}.
\]

Finally, from \eqref{eq:R_m},
\[
R^m
=
\frac{\alpha_f^S}{\alpha_f^C}
\frac{\overline L^C}{\overline L^S}
\left(
\frac{A_a\Phi}{\mu}
\right)^{1/\alpha_a^C}.
\]
All terms outside \(A_a/\mu\) are fixed across the comparison, and
\(1/\alpha_a^C>0\). Therefore
\[
R^{m,1}>R^{m,0}
\quad\Longleftrightarrow\quad
\frac{A_a^1}{\mu^1}>
\frac{A_a^0}{\mu^0}.
\]
Equivalently,
\[
\frac{A_a^1}{A_a^0}>
\frac{\mu^1}{\mu^0}.
\]
Thus, along the interior monopoly transition, the monopoly allocation ratio
\(R^m\) rises with AI productivity if and only if AI productivity grows faster
than the monopoly markup.
\end{proof}

\subsubsection*{F. Redistribution instruments and the regime comparison}

\noindent\textbf{Wage-ratio \((\gamma)\) target.} Consider first the wage-tax-only version of the policy. To bring the after-tax ratio to
\[
\frac{\widetilde w_C^{m,1}}{\widetilde w_S^{m,1}}=\gamma,
\qquad
\gamma<\frac{w_C^{m,1}}{w_S^{m,1}},
\]
the balanced-budget wage tax solves
\[
\tau_C w_C^{m,1}\overline L^C+\tau_S w_S^{m,1}\overline L^S=0.
\]
Using
\[
\widetilde w_C^{m,1}=(1-\tau_C)w_C^{m,1},
\qquad
\widetilde w_S^{m,1}=(1-\tau_S)w_S^{m,1},
\]
the target condition becomes
\[
\frac{(1-\tau_C)w_C^{m,1}}{(1-\tau_S)w_S^{m,1}}
=
\gamma.
\]
Solving the target condition jointly with the balanced-budget condition gives
\begin{equation}
\tau_C
=
\frac{\overline L^S}{\gamma\overline L^C+\overline L^S}
\left(
1-\frac{\gamma}{w_C^{m,1}/w_S^{m,1}}
\right).
\label{eq:tau_gamma_mono_app}
\end{equation}
This is structurally identical to the competitive formula, with market structure
entering only through the wage ratio. Since the monopoly wage ratio is compressed
relative to the competitive wage ratio,
\[
\frac{w_C^{m,1}}{w_S^{m,1}}
<
\frac{w_C^{comp,1}}{w_S^{comp,1}},
\]
the required wage tax on complementary workers is lower under monopoly for any
fixed feasible target \(\gamma\). Available profit-tax revenue can further
reduce the required wage tax if it is used to finance the subsidy to
substitutable workers.

\medskip

\noindent\textbf{Floor-wage target.} Consider first the wage-tax-only case,
\(\tau_\pi=0\). To hold substitutable workers at their benchmark monopoly wage,
\[
\widetilde w_S^{m,1}=w_S^{m,0},
\]
the required wage tax on complementary workers is
\[
\tau_C^m
=
\frac{(w_S^{m,0}-w_S^{m,1})\overline L^S}
{w_C^{m,1}\overline L^C}.
\]
To compare this tax to the competitive benchmark, let
\[
\Omega=\frac{w_C}{w_S}.
\]
The unit-cost identity implies
\[
w_S=\Theta\Omega^{-\alpha_f^C},
\qquad
w_C=\Theta\Omega^{\alpha_f^S},
\]
where \(\Theta\) is common across market structures for a fixed productivity
level. Since monopoly introduces the markup wedge
\[
\Omega^m=\Omega^{comp}\mu^{-1/\alpha_a^C},
\]
we obtain
\[
w_S^m=w_S^{comp}\mu^{\alpha_f^C/\alpha_a^C},
\qquad
w_C^m=w_C^{comp}\mu^{-\alpha_f^S/\alpha_a^C}.
\]
Therefore
\[
\frac{\tau_C^m}{\tau_C^{comp}}
=
(\mu^1)^{\alpha_f^S/\alpha_a^C}
\cdot
\frac{
(\mu^0)^{\alpha_f^C/\alpha_a^C}w_S^{comp,0}
-
(\mu^1)^{\alpha_f^C/\alpha_a^C}w_S^{comp,1}
}{
w_S^{comp,0}-w_S^{comp,1}
}.
\]
The tax base shrinks unambiguously under monopoly because
\[
w_C^{m,1}<w_C^{comp,1}.
\]
The behavior of the deficit depends on the markup. With a constant markup,
\(\mu^0=\mu^1=\mu\), the substitutable-wage schedule is scaled up uniformly:
\[
w_S^m=w_S^{comp}\mu^{\alpha_f^C/\alpha_a^C}.
\]
Thus the deficit is also scaled up by
\(\mu^{\alpha_f^C/\alpha_a^C}\), and the tax-rate ratio collapses to
\[
\frac{\tau_C^m}{\tau_C^{comp}}
=
\mu^{1/\alpha_a^C}
>
1.
\]
Hence the floor-wage tax is higher under monopoly when the markup is constant.
If the markup rises sufficiently fast with \(A_a\), the deficit can instead be
compressed, and the required wage tax can fall below the competitive level. The
comparison for this target therefore depends on how the markup changes with
productivity.

\medskip

\noindent\textbf{Profit-tax-only.} Now set \(\tau_C=0\) and allow a profit tax
\(\tau_\pi\leq 1\). Since complementary wages rise along the interior monopoly
transition, protecting substitutable workers is the binding worker constraint.
To hold
\[
\widetilde w_S^{m,1}=w_S^{m,0},
\]
the required transfer to substitutable workers is
\[
(w_S^{m,0}-w_S^{m,1})\overline L^S.
\]
A profit tax alone can finance this transfer if and only if
\[
\Pi^{m,1}
\geq
(w_S^{m,0}-w_S^{m,1})\overline L^S.
\]
Equivalently, there exists some \(\tau_\pi\leq 1\) such that
\[
\tau_\pi \Pi^{m,1}
=
(w_S^{m,0}-w_S^{m,1})\overline L^S.
\]
Thus profit taxation alone suffices exactly when monopoly profit covers the
substitutable-worker deficit. In the numerical illustrations, monopoly rents
accumulate gradually: near \(A^*\), the deficit exceeds profit and a wage tax on
complementary workers is still needed; beyond a crossing point, profit overtakes
the deficit and the profit tax alone suffices.

\medskip

\noindent\textbf{Price cap, full profit taxation, and partial profit taxation.}
Let regime \((i)\) be unregulated monopoly with no profit tax, regime \((ii)\)
be monopoly with a partial profit tax \(0<\tau_\pi<1\), and regime \((iii)\) be
restored competition, implemented either by a price cap at \(p_a^{comp}\) or,
under the tie-breaking assumption above, by full profit taxation
\(\tau_\pi=1\).

In regime \((iii)\), the markup is eliminated and the competitive allocation is
restored:
\[
Y_f=Y^{comp},
\qquad
\Pi=0.
\]
In regime \((ii)\), the monopolist retains after-tax profit
\[
\widetilde \Pi^{m,1}=(1-\tau_\pi)\Pi^{m,1}>0,
\]
and output remains restricted:
\[
Y^{m,1}<Y^{comp,1}.
\]
In regime \((i)\), the monopolist retains all monopoly profit:
\[
\widetilde \Pi^{m,1}=\Pi^{m,1}.
\]

Now impose the same floor-wage target in all three regimes:
\[
\widetilde w_S^1=w_S^{m,0}.
\]
The budget identity is
\[
Y^1
=
w_S^{m,0}\overline L^S
+
\widetilde w_C^1\overline L^C
+
\widetilde \Pi^1.
\]
Therefore, at productivity \(A_a^1\), the residual income of complementary
workers is
\[
\widetilde w_C^{(iii)}\overline L^C
=
Y^{comp,1}
-
w_S^{m,0}\overline L^S,
\]
\[
\widetilde w_C^{(ii)}\overline L^C
=
Y^{m,1}
-
w_S^{m,0}\overline L^S
-
(1-\tau_\pi)\Pi^{m,1},
\]
and
\[
\widetilde w_C^{(i)}\overline L^C
=
Y^{m,1}
-
w_S^{m,0}\overline L^S
-
\Pi^{m,1}.
\]
Differencing gives
\[
\widetilde w_C^{(iii)}-\widetilde w_C^{(ii)}
=
\frac{
(Y^{comp,1}-Y^{m,1})+(1-\tau_\pi)\Pi^{m,1}
}{\overline L^C}
>
0,
\]
and
\[
\widetilde w_C^{(ii)}-\widetilde w_C^{(i)}
=
\frac{\tau_\pi\Pi^{m,1}}{\overline L^C}
>
0.
\]
The first inequality uses
\[
Y^{comp,1}>Y^{m,1}
\]
from Proposition~\ref{prop:mon_comp_no_tax}, and the second uses
\(\Pi^{m,1}>0\). Hence
\[
\widetilde w_C^{(iii)}
>
\widetilde w_C^{(ii)}
>
\widetilde w_C^{(i)}.
\]
Thus, holding substitutable workers at the same benchmark wage, complementary
workers are best off when competition is restored, next best off under partial
profit taxation, and worst off under unregulated monopoly. The numerical
illustrations in Figure~\ref{fig:mono_regime_comparison} show that the gap to
competition widens with \(A_a^1\) and is only partly closed by partial profit
taxation.

\newpage

\section{Online Supplementary Appendix}

\label{Supplementary}

\subsection{Additional Results for the Competitive Transition}
\label{app:supp_transition}

This subsection collects the supplementary results that accompany the analysis of the competitive economy transition (Propositions~\ref{prop_wagegap} and the wage characterization in the main text): heatmaps of the equilibrium across the parameter space, the comparative statics of the AI expenditure share, and the labor-allocation patterns.

\begin{figure}[H]
\centering
\includegraphics[width=1\textwidth]{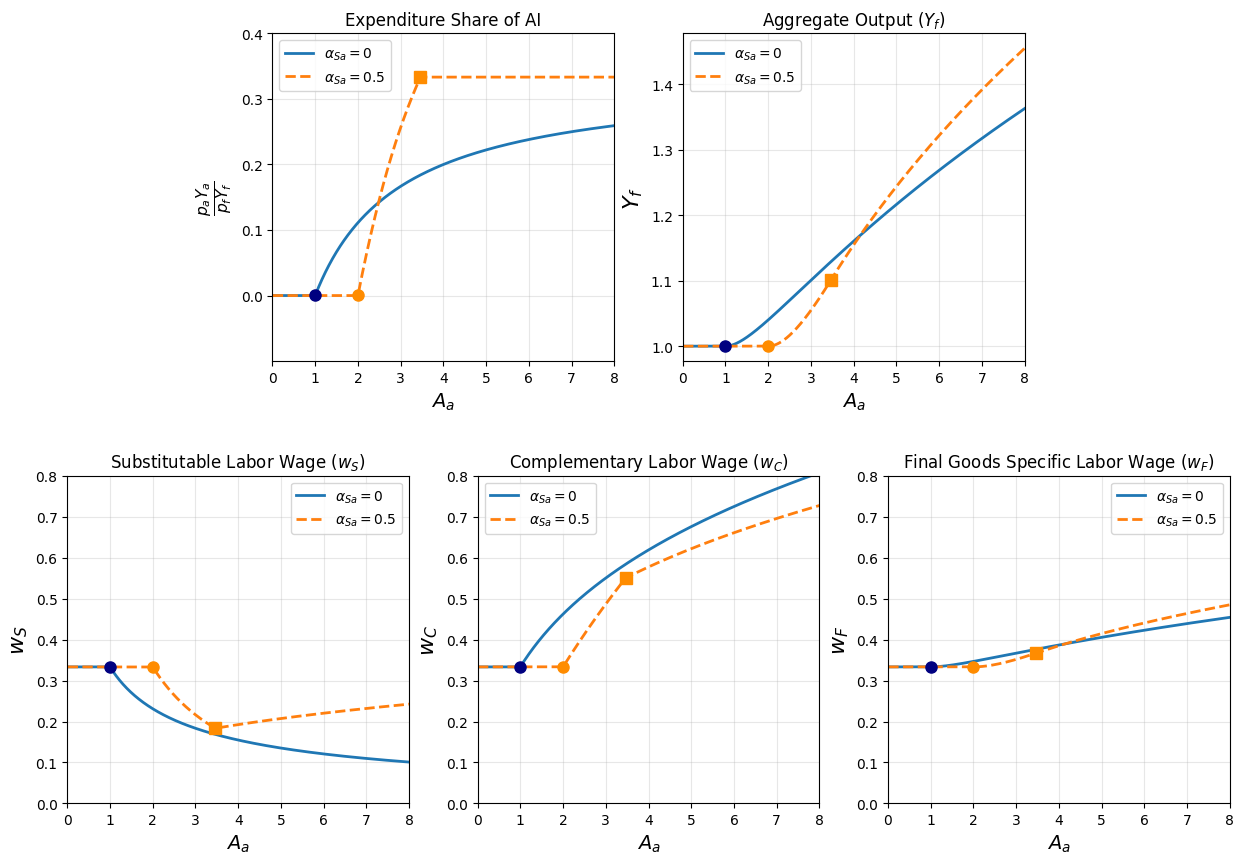}
\caption{AI Productivity and General Equilibrium Wage Dynamics}
\label{fig_wages_old}

\vspace{0.2cm}
 \begin{minipage}{0.95\textwidth}
 \footnotesize
 \textit{Notes:} This figure presents heatmaps showing how equilibrium outcomes vary across AI productivity ($A_a$, horizontal axis) and the substitutable labor share in AI production ($\alpha^S_{a}$, vertical axis). Panel 1 displays the expenditure share of AI in the final goods sector ($p_a Y_a / p_f Y_f$). Panel 2 shows aggregate output ($Y_f$). Panel 3 displays the wage for final-good-specific labor ($w_F$). Panels 4 and 5 display wages for substitutable ($w_S$) and complementary ($w_C$) labor, respectively. The solid white curve indicates the threshold $A_a^*$ separating the pre-AI and transition regimes; the dashed white curve indicates $A_a^{**}$ separating the transition and post-AI regimes. Parameters: $\alpha^S_{f}=\alpha^C_{f}=\alpha^F_{f}=1/3$, $\overline L^S=1$, $\overline L^C=\overline L^F=1$, and $A_f=1$.
 \end{minipage}
\end{figure}

Figure \ref{fig_wages_old} illustrates the equilibrium wage dynamics and AI adoption intensity across a given parameter space. The horizontal axis represents AI productivity $A_a$, while the vertical axis represents $\alpha^S_{a}$, the share of substitutable labor in AI production. The solid white curve marks the threshold $A_a^*$ at which AI adoption begins (transition from pre-AI to transition regime), while the dashed white curve marks $A_a^{**}$ at which substitutable labor is fully displaced from the final-good sector (transition to post-AI regime). For fixed parameters $\alpha^S_{f}=\alpha^C_{f}=\alpha^F_{f}=1/3$, $\overline L^S=1$, $\overline L^C=\overline L^F=1$, and $A_f=1$, all wages are equal in the pre-AI regime, and wages evolve as the final goods sector adopts AI. If $\alpha^S_{a}$ is low (bottom of each panel), the transition phase is prolonged because substitutable labor cannot be absorbed by the AI sector; in the limiting case $\alpha^S_{a} \to 0$, the economy remains indefinitely in the transition regime. As $\alpha^S_{a}$ increases, the gap between $A_a^*$ and $A_a^{**}$ narrows, depicting a faster transition.

\begin{figure}[H]
\centering
\includegraphics[width=0.8\textwidth]{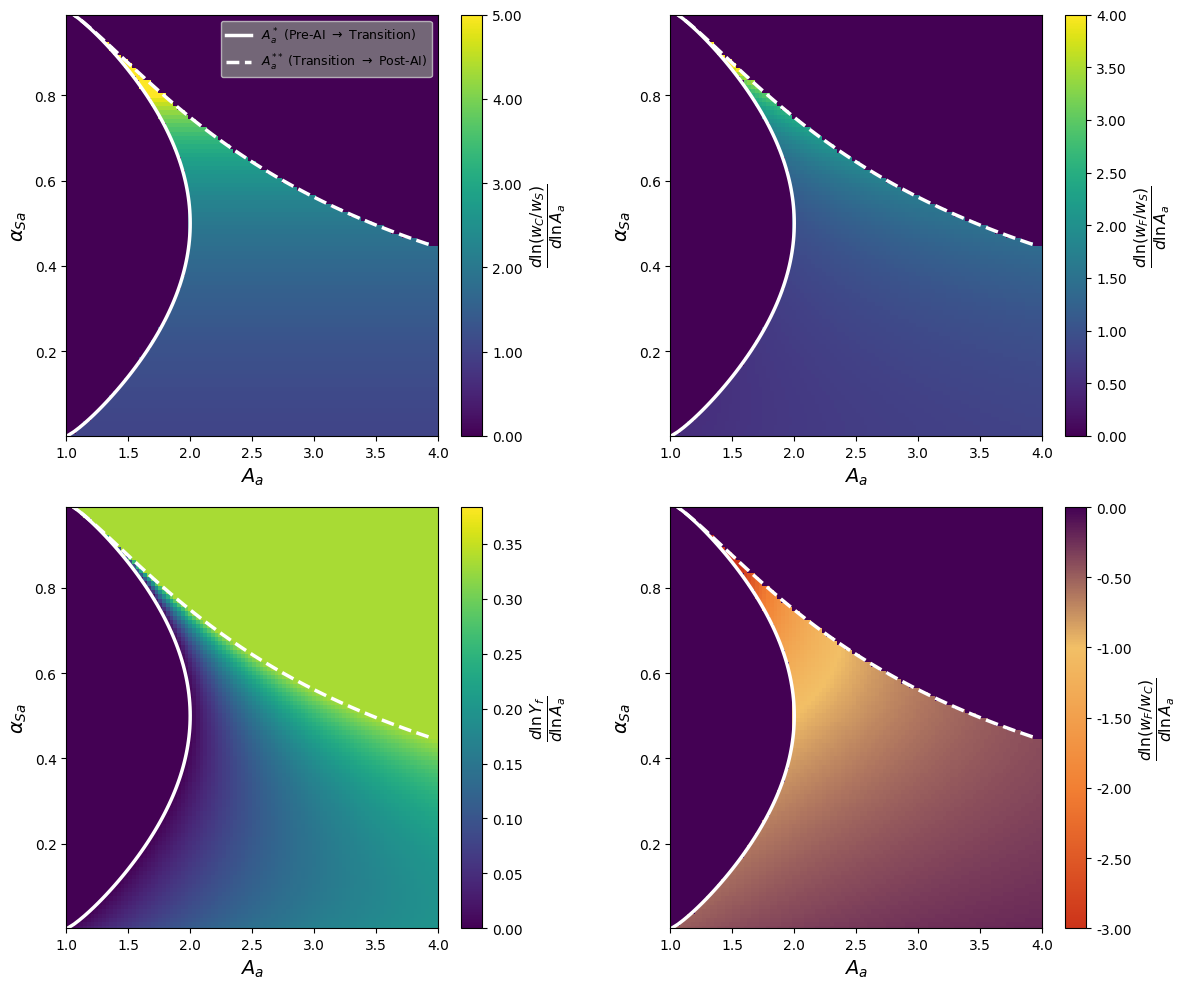}
\caption{Elasticities of Wage Inequality with Respect to AI Productivity}
\label{fig_wage_elasticities}

\vspace{0.2cm}
\begin{minipage}{0.95\textwidth}
\footnotesize
\textit{Notes:} This figure presents heatmaps of the elasticities of relative wages and aggregate output with respect to AI productivity ($A_a$). Panel 1 displays $\frac{d \ln(w_C/w_S)}{d \ln A_a} = \frac{1}{\alpha^C_{a}}$, which depends only on technology parameters during the transition. Panel 2 displays $\frac{d \ln(w_F/w_S)}{d \ln A_a}$. Panel 3 displays $\frac{d \ln Y_f}{d \ln A_a}$, the output elasticity equal to the AI expenditure share. Panel 4 displays $\frac{d \ln(w_F/w_C)}{d \ln A_a}$, which is negative during transition. In the pre-AI regime (left of solid white curve) and post-AI regime (right of dashed white curve), all wage ratio elasticities are zero; in the post-AI regime, output grows at rate $\alpha^S_{f}$. Parameters: $\alpha^S_{f}=\alpha^C_{f}=\alpha^F_{f}=1/3$, $\overline L^S=\overline L^C=\overline L^F=1$, and $A_f=1$.
\end{minipage}
\end{figure}

Figure~\ref{fig_wage_elasticities} visualizes the elasticities characterized in Proposition~\ref{prop_wagegap}. Panel 1 displays the elasticity $\frac{d \ln(w_C/w_S)}{d \ln A_a} = \frac{1}{\alpha^C_{a}}$, which is constant within the transition regime for a given $\alpha^S_{a}$ (and hence $\alpha^C_{a} = 1 - \alpha^S_{a}$). As $\alpha^S_{a}$ decreases (moving down on the vertical axis), $\alpha^C_{a}$ rises, and the elasticity decreases---indicating that wage inequality between complementary and substitutable labor responds less sharply to AI productivity improvements when AI production is more intensive in complementary labor. Panel 2 shows $\frac{d \ln(w_F/w_S)}{d \ln A_a}$, which is positive but strictly smaller than Panel 1, confirming the ordering in Proposition~\ref{prop_wagegap} part (i). Panel 3 displays the output elasticity $\frac{d \ln Y_f}{d \ln A_a} = \frac{p_a Y_a}{p_f Y_f}$, which increases with $A_a$ during the transition as AI captures a larger share of production. Panel 4 shows that $\frac{d \ln(w_F/w_C)}{d \ln A_a} < 0$ throughout the transition, reflecting that complementary labor gains faster than final-good-specific labor.

A striking pattern emerges in Panel 3: the output elasticity is highest in the post-AI regime (top-right region), where $\alpha^S_{a}$ is large and substitutable labor can be fully absorbed into AI production. When $\alpha^S_{a}$ is high, the economy transitions quickly to full AI adoption and achieves the largest output gains. Conversely, when $\alpha^S_{a}$ is low, substitutable labor remains trapped in the final-good sector with declining wages, and the economy cannot fully exploit AI's productivity potential.

\paragraph{Comparative statics of the AI expenditure share.}
The level of the AI expenditure share during the transition, $E_a = e_f^{\ast}\alpha^S_{f} = p_aY_a/(p_fY_f)$, is derived in Appendix~\ref{appendix} (the competitive-equilibrium derivations); using the interior value of $e_f^{\ast}$ in Equation~\eqref{ea} it equals
\begin{equation}
E_a=\frac{\alpha^S_{f}\,\frac{L^C}{L^S}\left(A_a(\alpha^C_{a})^{\alpha^C_{a}}(\alpha^S_{a})^{\alpha^S_{a}}\right)^{1/\alpha^C_{a}}-\alpha^C_{f}}{\alpha^C_{a}\left(1+\frac{L^C}{L^S}\left(A_a(\alpha^C_{a})^{\alpha^C_{a}}(\alpha^S_{a})^{\alpha^S_{a}}\right)^{1/\alpha^C_{a}}\right)} .
\label{eq:Ea_recall}
\end{equation}
Rather than restate this level, we record here its comparative statics in $A_a$, which underlie the claim in the main text that the AI expenditure share rises during the transition.

\begin{proposition}[Monotonicity and curvature of the AI expenditure share]
\label{mass}
During the substitution phase ($A_a^{*}<A_a<A_a^{**}$), the AI expenditure share $E_a=p_aY_a/(p_fY_f)$ in \eqref{eq:Ea_recall} is strictly increasing in $A_a$. Moreover, let
\[
\Lambda \;\equiv\; \frac{w_C \overline{L}^C}{w_S \overline{L}^S}\;=\;\frac{L^C}{L^S}\left(A_a(\alpha^C_{a})^{\alpha^C_{a}}(\alpha^S_{a})^{\alpha^S_{a}}\right)^{1/\alpha^C_{a}}
\]
denote the complementary-to-substitutable income-share ratio (the quantity $\Lambda$ of the main text). Then $E_a$ is convex in $A_a$ when $\Lambda<\frac{1-\alpha^C_{a}}{1+\alpha^C_{a}}$ and concave when $\Lambda>\frac{1-\alpha^C_{a}}{1+\alpha^C_{a}}$. Since $\Lambda$ is strictly increasing in $A_a$, the curvature changes sign at most once along the transition.
\end{proposition}

\begin{proof}
Write $\Lambda \equiv \frac{L^C}{L^S}\big(A_a(\alpha^C_{a})^{\alpha^C_{a}}(\alpha^S_{a})^{\alpha^S_{a}}\big)^{1/\alpha^C_{a}}$; by Proposition~\ref{prop_wagegap}, $w_C/w_S=(A_a(\alpha^C_{a})^{\alpha^C_{a}}(\alpha^S_{a})^{\alpha^S_{a}})^{1/\alpha^C_{a}}$, so $\Lambda=w_C\overline{L}^C/(w_S\overline{L}^S)$. In terms of $\Lambda$, \eqref{eq:Ea_recall} reads
\[
E_a=\frac{\alpha^S_{f}\Lambda-\alpha^C_{f}}{\alpha^C_{a}(1+\Lambda)} .
\]
Differentiating,
\[
\frac{dE_a}{d\Lambda}=\frac{\alpha^S_{f}+\alpha^C_{f}}{\alpha^C_{a}(1+\Lambda)^2}>0 ,
\qquad
\frac{d\Lambda}{dA_a}=\frac{1}{\alpha^C_{a}}\frac{\Lambda}{A_a}>0 ,
\]
so $dE_a/dA_a=(dE_a/d\Lambda)(d\Lambda/dA_a)>0$: the share is strictly increasing. Explicitly,
\[
\frac{dE_a}{dA_a}
=\frac{\frac{L^C}{L^S}\,(\alpha^S_{f}+\alpha^C_{f})}{(\alpha^C_{a})^2}
\cdot
\frac{\left((\alpha^C_{a})^{\alpha^C_{a}}(\alpha^S_{a})^{\alpha^S_{a}}\right)^{1/\alpha^C_{a}}\,A_a^{1/\alpha^C_{a}-1}}
{\left(1+\frac{L^C}{L^S}\left(A_a(\alpha^C_{a})^{\alpha^C_{a}}(\alpha^S_{a})^{\alpha^S_{a}}\right)^{1/\alpha^C_{a}}\right)^2}.
\]
For the curvature, differentiating once more gives
\[
\frac{d^2E_a}{dA_a^2}
= M\cdot\left[\left(\tfrac{1}{\alpha^C_{a}}-1\right)-\left(\tfrac{1}{\alpha^C_{a}}+1\right)\Lambda\right],
\]
where $M>0$.\footnote{$M=\dfrac{\frac{L^C}{L^S}\,(\alpha^S_{f}+\alpha^C_{f})}{(\alpha^C_{a})^2}
\left((\alpha^C_{a})^{\alpha^C_{a}}(\alpha^S_{a})^{\alpha^S_{a}}\right)^{1/\alpha^C_{a}}\dfrac{A_a^{1/\alpha^C_{a}-2}}
{\left(1+\frac{L^C}{L^S}\left(A_a(\alpha^C_{a})^{\alpha^C_{a}}(\alpha^S_{a})^{\alpha^S_{a}}\right)^{1/\alpha^C_{a}}\right)^3}>0.$}
The bracket is positive if and only if $\Lambda<(\tfrac{1}{\alpha^C_{a}}-1)/(\tfrac{1}{\alpha^C_{a}}+1)=\frac{1-\alpha^C_{a}}{1+\alpha^C_{a}}$, giving the stated convex/concave regions. Because $\Lambda$ is strictly increasing in $A_a$ and the threshold $\frac{1-\alpha^C_{a}}{1+\alpha^C_{a}}$ is constant, $\Lambda$ crosses it at most once, so the curvature changes sign at most once.
\end{proof}


Figure~\ref{fig_ai_share_output} illustrates these dynamics: the expenditure share of AI in value added ($E_a=e_f^{\ast}\alpha^S_{f}$) rises monotonically in $A_a$ but at a diminishing rate over the displayed range, consistent with Proposition~\ref{mass}.

\paragraph{Labor allocation.}
The next proposition records how the sectoral allocation of each labor type responds to AI productivity during the transition.

\begin{proposition}[Labor allocation]
\label{prop_labor_allocation}
In the substitution phase ($A_a^{*} < A_a < A_a^{**}$), the labor-allocation responses to AI productivity are
\[
\frac{\mathrm{d}\,(L^S_{f}/L^S)}{\mathrm{d} A_a} <0, \quad
\frac{\mathrm{d}\,(L^S_{a}/L^S)}{\mathrm{d} A_a} >0, \quad
\frac{\mathrm{d}\,(L^C_{f}/L^C)}{\mathrm{d} A_a} <0, \quad
\frac{\mathrm{d}\,(L^C_{a}/L^C)}{\mathrm{d} A_a} > 0, \quad
\frac{\mathrm{d}\,(L^F_{f}/L^F)}{\mathrm{d} A_a} =0 .
\]
\end{proposition}

These allocation patterns are driven by changes in $A_a$ through their effect on $e_f^{\ast}$ (the endogenous usage of AI). An increase in $e_f^{\ast}$ through higher $A_a$ leads to (i) an increase (decrease) in the share of substitutable labor employed in the AI (final-good) sector, and (ii) an increase (decrease) in the share of complementary labor employed in the AI (final-good) sector, while the employment of $L^F$ stays entirely in final-good production.

\begin{proof}
We use the automation-share representation. Using the equilibrium AI value-added share $e_f^{\ast}\alpha^S_{f}$, the equilibrium final-good output can be written as
\[
Y_{f}
=
\overline{A}\,
\left(L^S_{f}\right)^{\alpha^S_{f}-e_f^{\ast}\alpha^S_{f}}
\left((L^S_{a})^{\alpha^S_{a}}\right)^{e_f^{\ast}\alpha^S_{f}}
\left((L^C_{a})^{\alpha^C_{a}}\right)^{e_f^{\ast}\alpha^S_{f}}
\left(L^C_{f}\right)^{\alpha^C_{f}}
\left(L^F_{f}\right)^{\alpha^F_{f}},
\]
where $\overline{A}$ is a converted productivity constant that reproduces the same equilibrium $Y_f$ as the original technologies. Using this representation and $\alpha^S_{a}=1-\alpha^C_{a}$, the equilibrium labor allocations are
\[
\frac{L^S_{f}}{L^S}
=
\frac{(1-e_{f}^{\ast})\,\alpha^S_{f}}{\alpha^S_{f}-e_{f}^{\ast}\alpha^S_{f}\alpha^C_{a}},
\qquad
\frac{L^S_{a}}{L^S}
=
\frac{e_{f}^{\ast}\alpha^S_{f}(1-\alpha^C_{a})}{\alpha^S_{f}-e_{f}^{\ast}\alpha^S_{f}\alpha^C_{a}},
\]
\[
\frac{L^C_{a}}{L^C}
=
\frac{e_{f}^{\ast}\alpha^S_{f}\alpha^C_{a}}{\alpha^C_{f}+e_{f}^{\ast}\alpha^S_{f}\alpha^C_{a}},
\qquad
\frac{L^C_{f}}{L^C}
=
\frac{\alpha^C_{f}}{\alpha^C_{f}+e_{f}^{\ast}\alpha^S_{f}\alpha^C_{a}},
\]
and $L^F_{f}/L^F=1$, since $L^F$ is used only in final-good production. Each of the four nondegenerate ratios is a strictly monotone function of $e_f^{\ast}$: $L^S_{f}/L^S$ and $L^C_{f}/L^C$ are strictly decreasing in $e_f^{\ast}$, while $L^S_{a}/L^S$ and $L^C_{a}/L^C$ are strictly increasing in $e_f^{\ast}$. Because $e_f^{\ast}$ is strictly increasing in $A_a$ during the transition (Equation~\eqref{ea}), the stated signs follow, and $L^F_{f}/L^F$ is constant.
\end{proof}

\subsection{The Duration of the Transition Phase}
\label{transsec}

We quantify the persistence of the transition regime by examining the ratio of the full-substitution threshold ($A_a^{**}$) to the initial adoption threshold ($A_a^{*}$). This ratio measures the range of productivity growth over which the economy remains in the substitution phase---and thus the range over which the labor-supply invariance mechanism holds. Dividing the expressions in Equation~\eqref{threshold} yields a ratio that depends solely on technological parameters:
\begin{equation}
\label{ratio_eq}
\frac{A_a^{**}}{A_a^{*}} = \left( \frac{\alpha^S_{f}\alpha^C_{a}+\alpha^C_{f}}{\alpha^S_{a}\alpha^C_{f}} \right)^{\alpha^C_{a}}.
\end{equation}
Equation~\eqref{ratio_eq} reveals that the transition phase is not merely a temporary adjustment but can be structurally prolonged by the nature of AI production. Specifically, as the share of substitutable labor in the AI sector approaches zero ($\alpha^S_{a} \to 0$, implying $\alpha^C_{a} \to 1$), the ratio diverges to infinity. Similarly, a lower intensity of complementary tasks in the final-good sector ($\alpha^C_{f} \to 0$) magnifies the ratio, requiring a larger expansion of AI productivity to fully automate the sector.

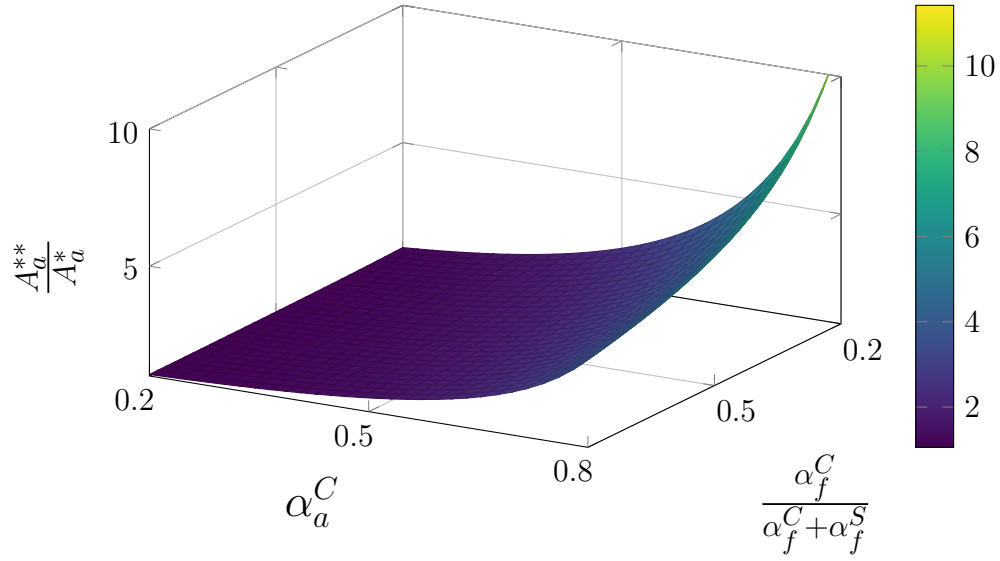
\begin{figure}[H]
	\centering
	\vspace{0.5cm}
	{
		\begin{tikzpicture}
			\begin{axis}[
				view={120}{30},
				xlabel=\(\frac{\alpha^C_{f}}{\alpha^C_{f}+\alpha^S_{f}}\),
				ylabel=\(\alpha^C_{a}\),
				zlabel=\(\frac{A_a^{**}}{A_a^*}\),
				xmin=0.2, xmax=0.8,
				ymin=0.2, ymax=0.8,
				zmin=1, zmax=10,
				xtick={0.2,0.5},
				ytick={0.2,0.5,0.8},
				grid=both,
				colormap/viridis,
				colorbar,
				tick label style={font=\normalsize},
				label style={font=\Large},
				title style={font=\Large},
				width=0.65\textwidth,
				height=0.45\textwidth
				]
				\addplot3[
				surf,
				domain=0.2:0.8,
				domain y=0.2:0.8,
				samples=30,
				samples y=30,
				] { ( ((1-x)*y + x) / ((1-y)*x) )^y };
			\end{axis}
		\end{tikzpicture}
	}
	\caption{Technological Determinants of Transition Duration ($A_a^{**}/A_a^*$)}
	\label{figureheat}
\end{figure}

Figure~\ref{figureheat} numerically illustrates these dynamics. The vertical axis tracks the relative intensity of complementary labor in the AI sector ($\alpha^C_{a}$), while the horizontal axis tracks the same relative intensity in the final-good sector ($\alpha^C_{f}$, normalized). The transition phase is longest when the AI sector is intensive in complementary labor (high $\alpha^C_{a}$) and the final-good sector is intensive in substitutable tasks (low $\alpha^C_{f}$). This prolongation occurs because a high $\alpha^C_{a}$ limits the AI sector's capacity to absorb substitutable labor (since $\alpha^S_{a}$ is low), precisely when the final-good sector has a large share of such labor to displace. 
With the reallocation channel constrained, the economy faces a structural bottleneck: the displaced workers remain ``stuck'' in the final-good sector for a much wider range of productivity levels, significantly extending the transition phase.

\section{Derivations for the Model Extensions (Section~\ref{SEC:ROBUSTNESS})}
\label{app:extensions}

We provide the derivations underlying the two extensions of Section~\ref{SEC:ROBUSTNESS}: AI used in its own production, and a general (CES) AI technology.

\subsection{AI Used in AI Production}
\label{app:ai_in_ai}

This subsection provides the full derivation of the wage gap, the Domar weight, and the multiplicity result for the AI-in-AI extension, in which AI can be used as an input in its own production. The AI and final-good technologies are
\[
Y_a = A_a\big[L^S_{a}+X_{aa}\big]^{\alpha^S_{a}}(L^C_{a})^{\alpha^C_{a}},
\qquad
Y_f = A_f\big[L^S_{f}+X_{fa}\big]^{\alpha^S_{f}}(L^C_{f})^{\alpha^C_{f}}(L^F_{f})^{\alpha^F_{f}},
\]
with $\alpha^S_{a}+\alpha^C_{a}=1$ and $\alpha^S_{f}+\alpha^C_{f}+\alpha^F_{f}=1$. Here $X_{aa}$ is AI used in AI production and $X_{fa}$ is AI used in final-good production, and AI-market clearing is $X_{aa}+X_{fa}=Y_a$. Let $Z_a\equiv L^S_{a}+X_{aa}$ and $Z_f\equiv L^S_{f}+X_{fa}$ be the substitutable bundles in the two sectors, and let $\Phi\equiv(\alpha^S_{a})^{\alpha^S_{a}}(\alpha^C_{a})^{\alpha^C_{a}}$.

\paragraph{Price relation and wage ratio.} Substitutable labor and AI are perfect substitutes within each bundle, so any equilibrium that uses both equates their prices sector-by-sector. In the final-good sector the marginal product of $Z_f$ equals $w_S$ while the firm pays $p_a$ for $X_{fa}$, hence $p_a=w_S$; in the AI sector the producer pays $w_S$ for $L^S_{a}$ and $p_a$ for $X_{aa}$, which again gives $p_a=w_S$. The AI cost function is the same Cobb-Douglas unit cost as in the baseline, so
\begin{equation}
p_a=\frac{w_S^{\alpha^S_{a}}w_C^{\alpha^C_{a}}}{A_a\Phi}=w_S
\quad\Longrightarrow\quad
\frac{w_C}{w_S}=(A_a\Phi)^{1/\alpha^C_{a}},
\label{eq:wage_ratio_aiai}
\end{equation}
exactly as in the baseline. The recursive use of AI does not alter the AI cost function and therefore leaves the wage ratio unchanged.

\paragraph{Cost shares.} With $p_f=1$ and competitive pricing, the Cobb-Douglas cost shares of the final-good sector (whose total cost equals $Y_f$) are
\[
w_S Z_f=\alpha^S_{f}Y_f,\qquad w_C L^C_{f}=\alpha^C_{f}Y_f,\qquad w_F L^F_{f}=\alpha^F_{f}Y_f,
\]
while the AI sector (whose total revenue $p_aY_a=w_SY_a$ equals its total cost) gives
\[
w_S Z_a=\alpha^S_{a}\,w_S Y_a \ \Rightarrow\ Z_a=\alpha^S_{a}Y_a,
\qquad
w_C L^C_{a}=\alpha^C_{a}\,w_S Y_a .
\]

\paragraph{Domar weight.} Let $\hat E_f\equiv X_{fa}/Z_f$ and $\hat E_a\equiv X_{aa}/Z_a$ be the AI shares of the substitutable bundle in the two sectors, so $X_{fa}=\hat E_f Z_f$ and $X_{aa}=\hat E_a Z_a$. AI-market clearing $X_{aa}+X_{fa}=Y_a$ becomes
\[
\hat E_a Z_a+\hat E_f Z_f=Y_a
\quad\Longrightarrow\quad
\hat E_a\,\alpha^S_{a}Y_a+\hat E_f\,\frac{\alpha^S_{f}Y_f}{w_S}=Y_a ,
\]
using $Z_a=\alpha^S_{a}Y_a$ and $w_S Z_f=\alpha^S_{f}Y_f$. Solving,
\[
w_S Y_a=\frac{\hat E_f}{1-\hat E_a\alpha^S_{a}}\,\alpha^S_{f}Y_f \equiv K\,\alpha^S_{f}Y_f,
\qquad
K\equiv \frac{\hat E_f}{1-\hat E_a\alpha^S_{a}} .
\]
Since $p_a=w_S$, the Domar weight of the AI sector (its sales relative to value added) is
\begin{equation}
\mathcal{D}_a=\frac{p_aY_a}{Y_f}=K\,\alpha^S_{f}.
\label{eq:domar_aiai}
\end{equation}

\paragraph{Wage gap in terms of $K$.} Summing complementary-labor demands,
\[
w_C\overline{L}^C=w_C L^C_{a}+w_C L^C_{f}=\alpha^C_{a}\,w_SY_a+\alpha^C_{f}Y_f=Y_f\big(\alpha^C_{f}+K\alpha^S_{f}\alpha^C_{a}\big),
\]
using $w_SY_a=K\alpha^S_{f}Y_f$. For substitutable labor,
\[
w_S\overline{L}^S=w_S L^S_{a}+w_S L^S_{f}=w_S(1-\hat E_a)Z_a+(1-\hat E_f)\,w_S Z_f ;
\]
substituting $Z_a=\alpha^S_{a}Y_a$, $w_S Z_f=\alpha^S_{f}Y_f$, $w_SY_a=K\alpha^S_{f}Y_f$, and the identity $\hat E_f=K(1-\hat E_a\alpha^S_{a})$ collapses the dependence on $\hat E_a,\hat E_f$ to
\[
w_S\overline{L}^S=\alpha^S_{f}Y_f\big[(1-\hat E_a)\alpha^S_{a}K+(1-\hat E_f)\big]=\alpha^S_{f}Y_f\big(1-K\alpha^C_{a}\big)=Y_f\big(\alpha^S_{f}-K\alpha^S_{f}\alpha^C_{a}\big).
\]
Dividing the two expressions yields the wage gap,
\begin{equation}
\frac{w_C}{w_S}=\frac{\overline{L}^S}{\overline{L}^C}\,\frac{\alpha^C_{f}+K\alpha^S_{f}\alpha^C_{a}}{\alpha^S_{f}-K\alpha^S_{f}\alpha^C_{a}},
\label{eq:wage_gap_aiai}
\end{equation}
which is exactly the baseline wage-gap formula with $K$ replacing $E_a/\alpha^S_{f}=e_f^{\ast}$.

\paragraph{Determination of $K$ and observational equivalence.} Equating \eqref{eq:wage_gap_aiai} with \eqref{eq:wage_ratio_aiai} and solving for $K$,
\begin{equation}
K=\frac{\hat E_f}{1-\hat E_a\alpha^S_{a}}=\frac{(A_a\Phi)^{1/\alpha^C_{a}}\,\alpha^S_{f}\overline{L}^C-\alpha^C_{f}\overline{L}^S}{\alpha^S_{f}\alpha^C_{a}\big((A_a\Phi)^{1/\alpha^C_{a}}\overline{L}^C+\overline{L}^S\big)} .
\label{eq:K_aiai}
\end{equation}
The right-hand side is identical to the expression that pins $E_a/\alpha^S_{f}=e_f^{\ast}$ in the baseline (Equation~\eqref{ea}). Hence $\mathcal{D}_a=K\alpha^S_{f}=E_a$, and the baseline wage formulas $w_S=Y_f(\alpha^S_{f}-E_a\alpha^C_{a})/\overline{L}^S$ and $w_C=Y_f(\alpha^C_{f}+E_a\alpha^C_{a})/\overline{L}^C$ deliver the same $w_S$, $w_C$, $p_a=w_S$, and $Y_f$ as the baseline during the transition. The two models are observationally equivalent in all aggregate outcomes, and the elasticity of output with respect to $A_a$ equals $\mathcal{D}_a$ in both.

\paragraph{Multiplicity.} Equation~\eqref{eq:K_aiai} pins down the composite $K=\hat E_f/(1-\hat E_a\alpha^S_{a})$ but not the pair $(\hat E_f,\hat E_a)$ separately. Any $(\hat E_f,\hat E_a)\in[0,1]^2$ satisfying $\hat E_f=K(1-\hat E_a\alpha^S_{a})$ is an equilibrium: different splits of AI output between own-use and final-good use are consistent with the same prices, wages, and output. 

\paragraph{Beyond the baseline completion threshold.} In the baseline the transition ends at $A_a^{**}$, where $e_f^{\ast}=1$ and the Domar weight $\mathcal{D}_a=\alpha^S_{f}$ becomes fixed. In the AI-in-AI specification there is no finite completion threshold: even once the final-good sector is fully automated ($\hat E_f=1$, $L^S_{f}=0$), the AI sector keeps substituting its own output for labor as $\hat E_a$ rises. The Domar weight $\mathcal{D}_a=\alpha^S_{f}/(1-\hat E_a\alpha^S_{a})$ continues to grow and the wage ratio $w_C/w_S=(A_a\Phi)^{1/\alpha^C_{a}}$ expands without bound, so AI-driven growth continues while inequality widens indefinitely. These dynamics are illustrated in Figure~\ref{fig:ai_in_ai_sim} of Section~\ref{SEC:ROBUSTNESS}.

\subsection{Beyond Cobb--Douglas: A CES AI Technology}
\label{app:ces}

This subsection derives the wage ratio reported in Section~\ref{SEC:ROBUSTNESS} when the AI sector uses a constant-elasticity-of-substitution (CES) technology, and shows that it is increasing in $A_a$ and nests the Cobb-Douglas case. The argument relies only on the equilibrium condition $p_a=w_S$ (perfect substitutability of AI output and substitutable labor in the final-good bundle) and on the AI unit-cost function; the functional form of the final-good technology plays no role, so the invariance of $w_C/w_S$ to labor supply and to final-good parameters carries over. A formal treatment under general nested CES structures is provided in \cite{kj}.

Let the AI sector produce with the CES technology
\[
Y_a=A_a\Big[\alpha^S_{a}\,(L^S_{a})^{\rho}+\alpha^C_{a}\,(L^C_{a})^{\rho}\Big]^{1/\rho},
\qquad \rho=\frac{\sigma-1}{\sigma},\quad \sigma>0,\ \sigma\neq 1,
\]
where $\sigma$ is the elasticity of substitution between substitutable and complementary labor in building AI and $\alpha^S_{a}+\alpha^C_{a}=1$. Cost minimization yields the standard CES unit-cost function
\begin{equation}
c(w_S,w_C;A_a)=\frac{1}{A_a}\Big[(\alpha^S_{a})^{\sigma}\,w_S^{1-\sigma}+(\alpha^C_{a})^{\sigma}\,w_C^{1-\sigma}\Big]^{1/(1-\sigma)} .
\label{eq:ces_cost}
\end{equation}

\paragraph{Wage ratio.} During the transition, substitutable labor and AI are both used in the final-good bundle, so their prices are equal, $p_a=w_S$, and competitive AI pricing sets $p_a=c(w_S,w_C;A_a)$. Hence $w_S=c(w_S,w_C;A_a)$. Substituting \eqref{eq:ces_cost} and raising both sides to the power $1-\sigma$,
\[
A_a^{1-\sigma}\,w_S^{1-\sigma}=(\alpha^S_{a})^{\sigma}\,w_S^{1-\sigma}+(\alpha^C_{a})^{\sigma}\,w_C^{1-\sigma}.
\]
Rearranging,
\[
\big[A_a^{1-\sigma}-(\alpha^S_{a})^{\sigma}\big]\,w_S^{1-\sigma}=(\alpha^C_{a})^{\sigma}\,w_C^{1-\sigma}
\quad\Longrightarrow\quad
\Big(\frac{w_C}{w_S}\Big)^{1-\sigma}=\frac{A_a^{1-\sigma}-(\alpha^S_{a})^{\sigma}}{(\alpha^C_{a})^{\sigma}} ,
\]
so the wage ratio is
\begin{equation}
\frac{w_C}{w_S}=\left[\frac{A_a^{1-\sigma}-(\alpha^S_{a})^{\sigma}}{(\alpha^C_{a})^{\sigma}}\right]^{1/(1-\sigma)} ,
\label{eq:ces_wage_ratio}
\end{equation}
which is the expression reported in Section~\ref{SEC:ROBUSTNESS}. It is well defined precisely when $A_a^{1-\sigma}>(\alpha^S_{a})^{\sigma}$: a lower bound on $A_a$ when $\sigma<1$ (the AI sector must be productive enough to be adopted) and an upper bound when $\sigma>1$, as discussed under completion behavior below.

\paragraph{Monotonicity in \(A_a\).}
Write \eqref{eq:ces_wage_ratio} as
\[
\left(\frac{w_C}{w_S}\right)^{1-\sigma}
=
\frac{
A_a^{1-\sigma}-(\alpha_a^S)^\sigma
}{
(\alpha_a^C)^\sigma
}.
\]
Let
\[
q=\frac{w_C}{w_S}.
\]
On the admissible interior range,
\[
A_a^{1-\sigma}-(\alpha_a^S)^\sigma>0,
\]
so the wage ratio is well-defined and positive.

Taking logs gives
\[
(1-\sigma)\ln q
=
\ln\left[
A_a^{1-\sigma}-(\alpha_a^S)^\sigma
\right]
-
\sigma\ln\alpha_a^C.
\]
Differentiating with respect to \(\ln A_a\),
\[
(1-\sigma)\frac{d\ln q}{d\ln A_a}
=
\frac{
(1-\sigma)A_a^{1-\sigma}
}{
A_a^{1-\sigma}-(\alpha_a^S)^\sigma
}.
\]
For \(\sigma\neq 1\), the factor \(1-\sigma\) cancels from both sides. Hence
\[
\frac{d\ln(w_C/w_S)}{d\ln A_a}
=
\frac{
A_a^{1-\sigma}
}{
A_a^{1-\sigma}-(\alpha_a^S)^\sigma
}.
\]
Since the denominator is positive on the admissible interior range, this
derivative is strictly positive. Therefore
\[
\frac{w_C}{w_S}
\]
is strictly increasing in \(A_a\), both for \(\sigma<1\) and for \(\sigma>1\).
Moreover,
\[
\frac{
A_a^{1-\sigma}
}{
A_a^{1-\sigma}-(\alpha_a^S)^\sigma
}
\geq 1,
\]
with strict inequality whenever \(\alpha_a^S>0\). Thus
\[
\frac{d\ln(w_C/w_S)}{d\ln A_a}\geq 1,
\]
with strict inequality when AI production uses substitutable labor.

The two cases differ in their completion behavior. If \(\sigma<1\), then
\(1-\sigma>0\), so \(A_a^{1-\sigma}\) rises without bound as \(A_a\) rises.
Therefore
\[
A_a^{1-\sigma}-(\alpha_a^S)^\sigma
\]
does not converge to zero at any finite productivity level. The wage ratio rises
with \(A_a\), but it does not diverge at a finite \(A_a\).

If \(\sigma>1\), then \(1-\sigma<0\), so \(A_a^{1-\sigma}\) decreases as
\(A_a\) rises. The admissible interior range ends when
\[
A_a^{1-\sigma}=(\alpha_a^S)^\sigma.
\]
Equivalently,
\[
A_a
=
(\alpha_a^S)^{-\sigma/(\sigma-1)}.
\]
As \(A_a\) approaches this value from below,
\[
A_a^{1-\sigma}-(\alpha_a^S)^\sigma\to 0.
\]
Because
\[
\frac{1}{1-\sigma}<0,
\]
we have
\[
\frac{w_C}{w_S}\to\infty.
\]
Thus, for \(\sigma>1\), the interior transition ends at a finite productivity
level, the analogue of the completion threshold \(A^{**}\) in the
Cobb--Douglas baseline.

\paragraph{Cobb-Douglas limit.} As $\sigma\to1$ (so $1-\sigma\to0$), \eqref{eq:ces_wage_ratio} converges to the Cobb-Douglas ratio. Writing $\varepsilon\equiv1-\sigma\to0$ and expanding,
\[
A_a^{\varepsilon}=1+\varepsilon\ln A_a+o(\varepsilon),\quad
(\alpha^S_{a})^{\sigma}=\alpha^S_{a}\big(1-\varepsilon\ln\alpha^S_{a}\big)+o(\varepsilon),\quad
(\alpha^C_{a})^{\sigma}=\alpha^C_{a}\big(1-\varepsilon\ln\alpha^C_{a}\big)+o(\varepsilon).
\]
The numerator is $A_a^{\varepsilon}-(\alpha^S_{a})^{\sigma}=\alpha^C_{a}+\varepsilon(\ln A_a+\alpha^S_{a}\ln\alpha^S_{a})+o(\varepsilon)$ (using $1-\alpha^S_{a}=\alpha^C_{a}$), so
\[
\frac{A_a^{\varepsilon}-(\alpha^S_{a})^{\sigma}}{(\alpha^C_{a})^{\sigma}}
=1+\varepsilon\left[\frac{\ln A_a+\alpha^S_{a}\ln\alpha^S_{a}}{\alpha^C_{a}}+\ln\alpha^C_{a}\right]+o(\varepsilon).
\]
Therefore
\[
\ln\frac{w_C}{w_S}=\frac{1}{\varepsilon}\ln\!\left(\frac{A_a^{\varepsilon}-(\alpha^S_{a})^{\sigma}}{(\alpha^C_{a})^{\sigma}}\right)
\xrightarrow[\varepsilon\to0]{}
\frac{1}{\alpha^C_{a}}\Big[\ln A_a+\alpha^S_{a}\ln\alpha^S_{a}+\alpha^C_{a}\ln\alpha^C_{a}\Big]
=\frac{1}{\alpha^C_{a}}\ln\!\big(A_a\Phi\big),
\]
with $\Phi=(\alpha^C_{a})^{\alpha^C_{a}}(\alpha^S_{a})^{\alpha^S_{a}}$. Hence $w_C/w_S\to(A_a\Phi)^{1/\alpha^C_{a}}$, the Cobb-Douglas baseline ratio.

\paragraph{Redistribution implications.} The redistribution analysis of Sections~\ref{SEC:INCOME_TAX} and the monopoly extension depends on the AI technology only through the fact that $w_C/w_S$ is increasing in $A_a$ during the transition. Since \eqref{eq:ces_wage_ratio} is increasing in $A_a$ on its admissible range for every $\sigma>0$, the qualitative conclusions---substitutable wages fall in absolute terms while complementary wages rise, the wage-ratio ($\gamma$) tax increases in $A_a$ while the floor-wage tax is single-peaked in $A_a$ (Proposition~\ref{taxrate}), and aggregate output rises so that Pareto-improving redistribution is feasible---carry over unchanged beyond the Cobb-Douglas case.

\section{Cost-Minimizing Monopolist versus GE-Monopolist First-Order Conditions}
\label{sec:monopoly_struc}

Consider AI production
\[
Y_a=A_a(L_a^S)^{\alpha_a^S}(L_a^C)^{\alpha_a^C},
\qquad \alpha_a^S+\alpha_a^C=1,
\]
with marginal products
\[
MP_S^a=\frac{\partial Y_a}{\partial L_a^S}
=
\alpha_a^S\frac{Y_a}{L_a^S},
\qquad
MP_C^a=\frac{\partial Y_a}{\partial L_a^C}
=
\alpha_a^C\frac{Y_a}{L_a^C}.
\]

The final-good sector uses the effective substitutable input
\[
Z=L_f^S+Y_a.
\]
Since total substitutable labor is \(\overline L^S\), and \(L_a^S\) is allocated to AI production,
\[
L_f^S=\overline L^S-L_a^S.
\]
Therefore,
\[
Z=\overline L^S-L_a^S+Y_a.
\]
Similarly,
\[
L_f^C=\overline L^C-L_a^C.
\]

For concreteness, suppose final-good production is
\[
Y_f=A_f Z^{\alpha_f^S}(L_f^C)^{\alpha_f^C}.
\]
Then the competitive final-good sector implies
\[
w_S(Z,L_f^C)
=
\alpha_f^S A_f Z^{\alpha_f^S-1}(L_f^C)^{\alpha_f^C},
\]
and
\[
w_C(Z,L_f^C)
=
\alpha_f^C A_f Z^{\alpha_f^S}(L_f^C)^{\alpha_f^C-1}.
\]
In the interior transition, final-good producers use both AI and substitutable labor, so
\[
p_a=w_S.
\]

The two formulations below differ in the input-mix condition.

\subsection*{A. Cost-minimization formulation}

In this formulation, the monopolist has market power in the AI output market, but chooses its input mix as if it were competitive in factor markets.

For a given output level \(Y_a\), the monopolist minimizes labor cost:
\[
\min_{L_a^S,L_a^C}
\; w_S L_a^S+w_C L_a^C
\]
subject to
\[
Y_a=A_a(L_a^S)^{\alpha_a^S}(L_a^C)^{\alpha_a^C}.
\]

In this cost-minimization problem, \(w_S\) and \(w_C\) are treated as the input prices relevant for the input-mix decision.

The Lagrangian is
\[
\mathcal L
=
 w_S L_a^S+w_C L_a^C
+
u\Big[
Y_a-A_a(L_a^S)^{\alpha_a^S}(L_a^C)^{\alpha_a^C}
\Big],
\]
where \(u\) is the Lagrange multiplier on the fixed-output constraint.

The first-order condition with respect to \(L_a^S\) is
\[
\frac{\partial \mathcal L}{\partial L_a^S}
=
w_S-uMP_S^a=0,
\]
so
\[
w_S=uMP_S^a.
\]
The first-order condition with respect to \(L_a^C\) is
\[
\frac{\partial \mathcal L}{\partial L_a^C}
=
w_C-uMP_C^a=0,
\]
so
\[
w_C=uMP_C^a.
\]

Thus the FOCs are
\[
w_S=u MP_S^a,
\qquad
w_C=u MP_C^a.
\]

Dividing the two FOCs gives the cost-minimizing input-mix condition:
\[
\frac{MP_S^a}{MP_C^a}
=
\frac{w_S}{w_C}
\]
or equivalently
\[
\frac{\alpha_a^S}{\alpha_a^C}\frac{L_a^C}{L_a^S}
=
\frac{w_S}{w_C}.
\]

Given these cost-minimizing input demands, the remaining equilibrium conditions and market clearing determine, for each feasible \(Y_a\), a corresponding allocation and wage vector:
\[
L_a^{S,costmin}(Y_a),\quad
L_a^{C,costmin}(Y_a),\quad
L_f^{S,costmin}(Y_a),\quad
L_f^{C,costmin}(Y_a),\quad
w_S^{costmin}(Y_a),\quad
w_C^{costmin}(Y_a).
\]
In the interior transition, final-good producers use both AI and substitutable labor, so the induced AI price is
\[
p_a^{costmin}(Y_a)=w_S^{costmin}(Y_a).
\]
Thus this formulation generates an inverse demand curve \(p_a^{costmin}(Y_a)\).

The monopolist then chooses the level of AI production:
\[
Y_a^{costmin}
\in
\arg\max_{Y_a}
\Pi^{costmin}(Y_a),
\]
where
\[
\Pi^{costmin}(Y_a)
=
p_a^{costmin}(Y_a)Y_a
-
w_S^{costmin}(Y_a)L_a^{S,costmin}(Y_a)
-
w_C^{costmin}(Y_a)L_a^{C,costmin}(Y_a).
\]
Equivalently, because each feasible \(Y_a\) induces a unique price \(p_a^{costmin}(Y_a)\), the monopolist may be viewed as choosing either \(Y_a\) or the corresponding price \(p_a^{costmin}(Y_a)\).

\subsection*{B. GE-monopoly formulation}

In the GE-monopoly formulation, the monopolist chooses inputs while anticipating the equilibrium response of wages, prices, and allocations.

Because \(p_a=w_S\) in the interior transition, monopoly profit is
\[
\Pi^m
=
p_aY_a-w_S L_a^S-w_C L_a^C
=
w_S(Y_a-L_a^S)-w_C L_a^C.
\]

Using
\[
Z=\overline L^S-L_a^S+Y_a,
\]
we can rewrite
\[
Y_a-L_a^S=Z-\overline L^S.
\]
Also,
\[
L_f^C=\overline L^C-L_a^C.
\]

Therefore, after substituting the equilibrium wage functions, monopoly profit can be viewed as a reduced-form function of the induced final-good input \(Z\) and the AI-sector complementary labor \(L_a^C\):
\[
\Pi^m
=
w_S(Z,\overline L^C-L_a^C)(Z-\overline L^S)
-
w_C(Z,\overline L^C-L_a^C)L_a^C.
\]
This expression is the same monopoly profit as above, but written directly in terms of \(Z\) and \(L_a^C\). In particular, \(Z\) summarizes the effect of \(L_a^S\) on final-good production and on the equilibrium wages.

Now hold \(L_a^C\) fixed and vary \(L_a^S\). Since \(L_a^C\) is held fixed, \(L_f^C=\overline L^C-L_a^C\) is also fixed. Hence the only channel through which \(L_a^S\) affects this reduced-form profit is through its effect on \(Z\).

A marginal increase in \(L_a^S\) has two effects on \(Z\). First, it removes one unit of substitutable labor from final-good production. Second, it raises AI output by \(MP_S^a\). Hence
\[
\frac{\partial Z}{\partial L_a^S}
=
\frac{\partial}{\partial L_a^S}
\left(\overline L^S-L_a^S+Y_a\right)
=
-1+\frac{\partial Y_a}{\partial L_a^S}
=
MP_S^a-1.
\]

Let \(\Pi_Z^m\) denote the derivative of the reduced-form monopoly profit with respect to \(Z\), holding \(L_a^C\) fixed:
\[
\Pi_Z^m
\equiv
\frac{\partial}{\partial Z}
\left[
w_S(Z,\overline L^C-L_a^C)(Z-\overline L^S)
-
w_C(Z,\overline L^C-L_a^C)L_a^C
\right].
\]
Then, by the chain rule,
\[
\frac{\partial \Pi^m}{\partial L_a^S}
=
\Pi_Z^m
\frac{\partial Z}{\partial L_a^S}.
\]
Substituting the expression for \(\partial Z/\partial L_a^S\), we get
\[
\frac{\partial \Pi^m}{\partial L_a^S}
=
\Pi_Z^m\left(MP_S^a-1\right).
\]

The interior FOC with respect to \(L_a^S\) is therefore
\[
\Pi_Z^m\left(MP_S^a-1\right)=0.
\]
On the non-degenerate branch where the marginal value of the effective substitutable input is not zero,
\[
\Pi_Z^m\neq0,
\]
the FOC implies
\[
MP_S^a=1.
\]

Using
\[
MP_S^a=\alpha_a^S\frac{Y_a}{L_a^S},
\]
the condition \(MP_S^a=1\) becomes
\[
\alpha_a^S\frac{Y_a}{L_a^S}=1.
\]

Now substitute this condition back into AI production:
\[
Y_a
=
A_a(L_a^S)^{\alpha_a^S}(L_a^C)^{\alpha_a^C}.
\]
Using \(L_a^S=\alpha_a^S Y_a\), we obtain
\[
Y_a
=
A_a(\alpha_a^S Y_a)^{\alpha_a^S}(L_a^C)^{\alpha_a^C}.
\]
Equivalently,
\[
Y_a
=
A_a(\alpha_a^S)^{\alpha_a^S}
Y_a^{\alpha_a^S}(L_a^C)^{\alpha_a^C}.
\]
Dividing both sides by \(Y_a^{\alpha_a^S}\) gives
\[
Y_a^{1-\alpha_a^S}
=
A_a(\alpha_a^S)^{\alpha_a^S}(L_a^C)^{\alpha_a^C}.
\]
By using \(1-\alpha_a^S=\alpha_a^C\) and taking both sides to the power \(1/\alpha_a^C\), we get

\[
Y_a=\kappa L_a^C,
\qquad
\kappa=\left[A_a(\alpha_a^S)^{\alpha_a^S}\right]^{1/\alpha_a^C}.
\]

Hence
\[
MP_C^a
=
\alpha_a^C\frac{Y_a}{L_a^C}
=
\alpha_a^C\kappa.
\]

Therefore the GE-monopoly input-mix condition implies
\[
\frac{MP_S^a}{MP_C^a}
=
\frac{1}{\alpha_a^C\kappa}.
\]

\bigskip
\bigskip

\noindent{Next, we compare how these two formulations differ.}

\subsection*{Where the two formulations differ}

Cost minimization requires
\[
\frac{MP_S^a}{MP_C^a}=\frac{w_S}{w_C}.
\]
The GE-monopoly formulation implies
\[
\frac{MP_S^a}{MP_C^a}=\frac{1}{\alpha_a^C\kappa}.
\]
Therefore, the two formulations coincide only if
\[
\frac{w_S}{w_C}=\frac{1}{\alpha_a^C\kappa}
\]
or equivalently
\[
w_C=\alpha_a^C\kappa w_S.
\]

But under the GE-monopoly formulation, after using \(Y_a=\kappa L_a^C\) and \(L_a^S=\alpha_a^S Y_a\), profit is
\[
\Pi^m
=
w_S(Y_a-L_a^S)-w_C L_a^C.
\]
Since
\[
Y_a-L_a^S
=
Y_a-\alpha_a^S Y_a
=
\alpha_a^C Y_a
=
\alpha_a^C \kappa L_a^C,
\]
we get
\[
\Pi^m
=
\left(\alpha_a^C\kappa w_S-w_C\right)L_a^C.
\]
Thus, positive monopoly profit requires
\[
\Pi^m>0
\quad\Longleftrightarrow\quad
w_C<\alpha_a^C\kappa w_S.
\]
Hence, whenever monopoly profit is strictly positive,
\[
\frac{w_S}{w_C}
>
\frac{1}{\alpha_a^C\kappa}
=
\frac{MP_S^a}{MP_C^a}.
\]
Therefore, the cost-minimization FOC and the GE-monopoly FOC differ for positive monopoly profit.

In the cost-minimization formulation, the monopolist chooses AI output \(Y_a\)
while anticipating the general-equilibrium allocation associated with that
output. For each feasible \(Y_a\), the input mix \((L_a^S,L_a^C)\) is determined
by cost minimization at the equilibrium wages induced by that output, together
with final-good-sector optimality and market clearing. This generates functions
\[
p_a^{costmin}(Y_a)=w_S^{costmin}(Y_a),\quad
w_C^{costmin}(Y_a),\quad
L_a^{S,costmin}(Y_a),\quad
L_a^{C,costmin}(Y_a).
\]
The monopolist then chooses \(Y_a\), equivalently the associated price
\(p_a^{costmin}(Y_a)\), to maximize
\[
\Pi^{costmin}(Y_a)
=
p_a^{costmin}(Y_a)Y_a
-
w_S^{costmin}(Y_a)L_a^{S,costmin}(Y_a)
-
w_C^{costmin}(Y_a)L_a^{C,costmin}(Y_a).
\]
Thus the monopolist anticipates how wages and allocations vary with its output
choice, but the conditional input mix is still determined by the
cost-minimization condition
\[
\frac{MP_S^a}{MP_C^a}=\frac{w_S^{costmin}(Y_a)}{w_C^{costmin}(Y_a)}.
\]

Finally, the GE-monopoly problem has a weakly larger choice set than the
cost-minimization formulation. In the GE-monopoly formulation, the monopolist
chooses directly over all feasible input pairs \((L_a^S,L_a^C)\), internalizing
the induced equilibrium wages and allocations. In the cost-minimization
formulation, the monopolist is restricted to input pairs that satisfy the
cost-minimizing input-mix condition for each \(Y_a\), namely
\[
\frac{MP_S^a}{MP_C^a}
=
\frac{w_S}{w_C},
\]
or equivalently,
\[
\frac{\alpha_a^S}{\alpha_a^C}
\frac{L_a^C}{L_a^S}
=
\frac{w_S}{w_C}.
\]
Therefore, any allocation chosen under the cost-minimization formulation is
feasible in the GE-monopoly problem, but not conversely. Hence, when both
problems are solved over the same transition domain,
\[
\Pi^{GE}\geq \Pi^{costmin}.
\]
The inequality is strict whenever the GE-monopoly optimum does not satisfy the
cost-minimizing input-mix condition; in particular, this is the relevant case
whenever positive GE-monopoly profit makes the GE and cost-minimization FOCs
differ.

\begin{figure}[t]
\centering
\includegraphics[width=\textwidth]{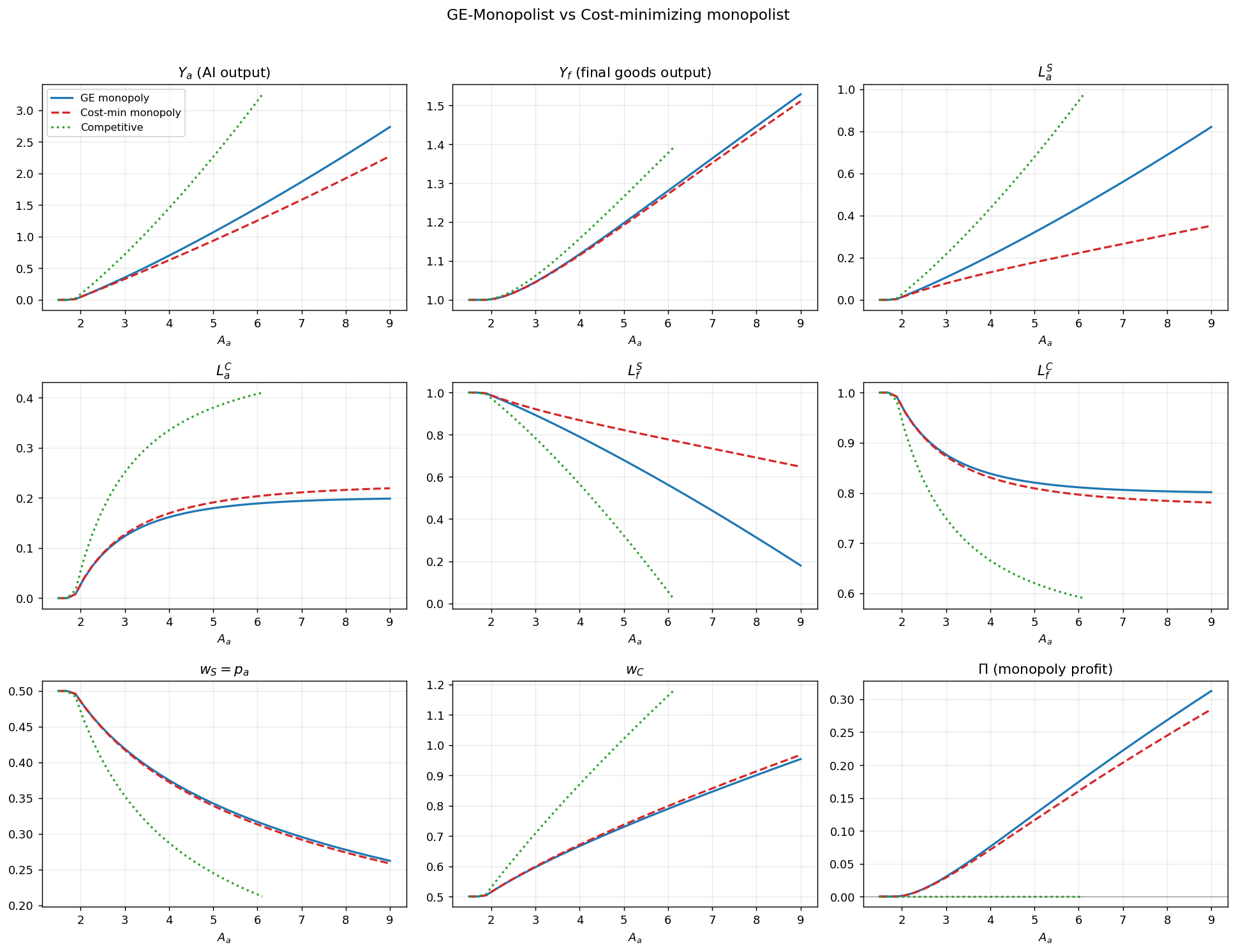}
\caption{GE-Monopoly, Cost-Minimizing Monopoly, and Competition}
\label{fig:mono_two_approaches}
 
\vspace{0.1cm}
\begin{minipage}{0.95\textwidth}
\footnotesize
\textit{Notes:} Each panel plots an equilibrium object against AI productivity $A_a$ under three closures of the AI-sector problem. The closures are the GE-monopoly of the main text (solid), the Cost-minimizing monopoly of this appendix (dashed), and the competitive benchmark (dotted). The bottom-right panel reports profit: the GE-monopoly earns the most, and the Cost-minimizing monopoly strictly less, because it is restricted to cost-minimizing input bundles. Consistent with Lemma~\ref{lem:mon_alloc}, both monopoly closures restrict scale relative to competition---$Y_a$, $L^S_{a}$, and $L^C_{a}$ are all lower. Parameters: $\alpha^S_{a}=0.30$, $\alpha^C_{a}=0.70$, $\alpha^S_{f}=\alpha^C_{f}=0.50$, $\overline{L}^S=\overline{L}^C=1$, $A_f=1$, over the range $A_a\in[1.5,9]$. The competitive economy transition terminates near $A_a\approx 6.2$.
\end{minipage}
\end{figure}

\end{document}